\documentclass[10pt,twocolumn,floatfix,prx,superscriptaddress]{revtex4-1}

\usepackage{graphicx}
\usepackage{color}
\usepackage{xcolor}
\usepackage{amssymb}
\usepackage{amsmath}
\usepackage{amsthm}
\usepackage{epsfig}
\usepackage{bm}
\usepackage[american]{babel}
\usepackage{hyperref}
\usepackage{braket}
\usepackage[T1]{fontenc} 
\usepackage{txfonts}
\usepackage{comment}
\usepackage{textcomp}
\usepackage[capitalize]{cleveref}
\usepackage{subfig}
\usepackage[export]{adjustbox}
\usepackage{extarrows}
\usepackage[font=small,labelfont=bf,
   justification=justified,
   format=plain]{caption}
\usepackage{tikz}
\usetikzlibrary{matrix,calc}
\usepackage{makecell}
\usepackage{cancel}

\newtheorem{theorem}{Theorem}
\newtheorem{lemma}{Lemma}
\newtheorem{proposition}{Proposition}
\newtheorem{definition}{Definition}
\newtheorem{corollary}{Corollary}

\DeclareGraphicsExtensions{.pdf,.png,.jpg}
\hypersetup{colorlinks=true,urlcolor=blue,breaklinks=true,citecolor=blue,linkcolor=blue,pdfstartview=FitH,pdfpagemode=UseNone}

\usepackage{algpseudocode}
\usepackage{algorithm}

\newcommand{\Mirror}[1]{
    \node[circle,draw,fill=black,inner sep=0,minimum size=3] (m#1) at (#1*0.5,0) {};
    }

\newcommand{\MirrorAngle}[2]{
    \draw (#1*0.5,0) -- (#1*0.5+0.5,0);
    \node[anchor=south] (a#1) at (#1*0.5+0.25,-0.08) {\tiny #2};
    }

\newcommand{\MirrorRing}[1]{
    \node[circle,draw,fill=none,inner sep=0,minimum size=6] (r#1) at (#1*0.5,0) {};
    }

\newcommand{\CoxeterFour}[7]{
    \begin{tikzpicture}
        \Mirror{0};
        \foreach \x in {#1} {\MirrorRing{0};}

        \MirrorAngle{0}{#2};

        \Mirror{1};
        \foreach \x in {#3} {\MirrorRing{1};}
        
        \MirrorAngle{1}{#4};

        \Mirror{2};
        \foreach \x in {#5} {\MirrorRing{2};}
        
        \MirrorAngle{2}{#6};

        \Mirror{3};
        \foreach \x in {#7} {\MirrorRing{3};}
    \end{tikzpicture}   
    }

\begin{document}

\title{Quantum Rainbow Codes: Achieving Linear Rate, Growing Distance and Transversal Non-Clifford Gates with Generalised Colour Codes}
\author{Thomas R. Scruby}
\email{t.r.scruby@gmail.com}
\affiliation{Okinawa Institute of Science and Technology, 1919-1 Tancha, Onna, Kunigami District, Okinawa 904-0412, Japan}
\author{Arthur Pesah}
\email{arthur.pesah@gmail.com}
\affiliation{Department of Physics \& Astronomy, University College London, Gower St, London WC1E 6BT, United Kingdom}
\author{Mark Webster}
\email{mark.acacia@gmail.com}
\affiliation{Department of Physics \& Astronomy, University College London, Gower St, London WC1E 6BT, United Kingdom}

\begin{abstract}
We introduce rainbow codes, a novel class of quantum error correcting codes generalising colour codes and pin codes. Rainbow codes can be defined on any $D$-dimensional simplicial complex that admits a valid $(D+1)$-colouring of its $0$-simplices. We study in detail the case where these simplicial complexes are derived from chain complexes obtained via the hypergraph product and, by reinterpreting these codes as collections of colour codes joined at domain walls, show that we can obtain code families with growing distance and number of encoded qubits as well as logical non-Clifford gates implemented by transversal application of $T$ and $T^\dag$. By combining these techniques with the quasi-hyperbolic colour codes of Zhu et al. (arXiv:2310.16982) we obtain a family of codes with transversal non-Clifford gates and parameters $[\![n,\Theta(n),\Theta(\log(n))]\!]$. This is the first example of a family of LDPC codes with linear rate, growing distance and transversal non-Clifford gates, which are necessary conditions for the magic-state distillation parameter $\gamma = \log_d(n/k)$ to be made arbitrarily small. 
In contrast to several other constructions that satisfy these requirements, our codes are natively defined on qubits, are LDPC, and have non-Clifford gates implementable by single-qubit (rather than entangling) physical operations, but are not asymptotically good. 

\end{abstract}

\maketitle

\section{Introduction}
\label{section:intro}
Quantum error correcting codes allow us to protect information from environmental noise and perform fault-tolerant quantum computation, but this comes at the cost of increased computational overheads and increased difficulty of implementing logical operations. An important theorem by Eastin and Knill tells us that no code can have a transversal implementation of a universal gate set~\cite{eastin_restrictions_2009} and even within the confines of this theorem there is a large degree of variation in how many transversal gates a code can possess. For instance, the two-dimensional surface code~\cite{kitaev_fault-tolerant_2003} admits only a transversal CNOT gate (although other operations can be performed if we relax the definition of transversal~\cite{moussa_transversal_2016}), while the two-dimensional colour code~\cite{bombin2006topological} admits transversal implementations of the full Clifford group. In higher dimensional colour codes transversal non-Clifford gates implemented by single-qubit physical operations also become possible~\cite{bombin2007topological}.

In recent years much progress has been made in constructing codes with improved encoding rate and distance~\cite{tillich_quantum_2014,breuckmann_balanced_2021,panteleev_asymptotically_2022,leverrier_quantum_2022}, and many of these constructions can be viewed as generalisations of the surface code beyond manifolds and so have gate sets which are similarly limited. It is therefore natural to ask whether colour codes can be generalised in a similar way, and in fact such a generalisation already exists in the form of pin codes~\cite{vuillot2022quantum}. However, when combined with the same product constructions that generate high rate and distance generalised surface codes, these generalised colour codes typically have only constant distance. Additionally, in higher dimensions they can possess transversal non-Clifford gates but this property is not guaranteed without further modification of the codes by e.g. puncturing techniques. 

In this work we further generalise colour codes and pin codes to obtain a class of codes we call \textit{rainbow codes}. These codes are obtained by identifying the low-weight logical operators of the pin codes and including some of them in the stabiliser group. We will see that there are a number of different choices for which operators to include (and which to remove) from the stabiliser group and these different choices produce codes with very different properties. We show that some of these choices, when combined with the hypergraph product~\cite{tillich_quantum_2014}, can result in families of codes with growing distance and number of encoded qubits as well as transversal implementations of logical non-Clifford gates. We also show that these codes can be interpreted as joinings of Euclidean colour codes at domain walls between two copies of the topological phase. By combining this construction with the recently proposed quasi-hyperbolic codes of~\cite{zhu2023non} we can obtain codes with finite asymptotic rate, non-constant distance and transversal non-Clifford gates. The existence of such codes is a necessary condition for the magic-state yield parameter $\gamma = \log_d(n/k)$~\cite{bravyi_universal_2005,bravyi_magic-state_2012} to be made arbitrarily small, and until very recently the construction of such codes was an open problem, although their existence has now been demonstrated for the case of asymptotically good non-LDPC codes with gates implemented by transversal $CCZ$~\cite{wills_constant-overhead_2024,nguyen_good_2024,golowich_asymptotically_2024}. In contrast, our codes have transversal non-Clifford gates implemented by transversal $T/T^\dag$, are LDPC, but are not asymptotically good. We note also that the logical action of our non-Clifford gate results in states with a complex entanglement structure and it is non-obvious how to use these states in magic state distillation procedures. Subsequent to the release of the first version of this work, a construction for LDPC codes that provably achieve $\gamma \rightarrow 0$ was identified by Golowich and Lin~\cite{golowich2024quantumldpccodestransversal}.

We begin this paper by reviewing the constructions of quantum colour codes and pin codes then defining rainbow codes in \cref{section:rainbow}. We then consider rainbow codes defined via the hypergraph product in \cref{section:rainbow_hgp} and show how these codes can be interpreted as joinings of colour codes on manifolds. In \cref{section:logical_gates} we consider the action of transversal non-Clifford gates on these codes and identify the cases in which they can implement logical non-Clifford operations. Finally we present several examples of both finite-size and asymptotic rainbow codes, including the previously mentioned family with linear encoding rate, in \cref{section:examples}. We also describe various other properties and transformations of these codes in the appendices, such as a method for modifying rainbow codes to reduce physical qubit count and stabiliser weight while potentially preserving $k$ and $d$ (\cref{app:edge_contraction}), an unfolding map for some classes of rainbow code (\cref{app:unfolding}) and various algorithms enabling efficient construction of rainbow codes from chain complexes (\cref{app:algorithms}).  

\section{Rainbow codes}
\label{section:rainbow}
\subsection{Colour codes and pin codes}
\label{subsection:rainbow_ccpc}

We begin by introducing colour codes and pin codes, and by fixing some terminology and notation which will be used throughout the paper.

\begin{definition}[Chain complex]
    A length-$D$ chain complex is the data of $D+1$ vector spaces $C_1,\ldots,C_D$ and linear maps $\delta_1,\ldots,\delta_D$ with $\delta_i:C_i \rightarrow C_{i-1}$, called \textbf{boundary maps}, such that $\delta_i \circ \delta_{i+1} = 0$ for all $i \in \{1,\ldots,D-1\}$. We write a chain complex with the following diagram:
    \begin{equation*}
        \mathcal{C} = C_D \overset{\delta_D}{\rightarrow} C_{D-1} ~ ... ~ C_1 \overset{\delta_1}{\rightarrow} C_0
    \end{equation*}
    The elements of each vector space $C_k$ are called $\bm{k}$\textbf{-chains}.
\end{definition}
In the rest of this paper, we will only consider chain complexes defined with finite-dimensional vector spaces over the field $\mathbb{F}_2$. 
In this case, we can equip each vector space with a basis, and represent the boundary maps as binary matrices. 
We call each basis element of $C_k$ a $\bm{k}$\textbf{-cell}.
A chain complex equipped with those bases can then be viewed as a $(D+1)$-partite graph, whose nodes are the $k$-cells, where $k$ determines the type of a node in the partition. 
The edges are then determined by the boundary maps $\delta_k$, seen as biadjacency matrices between nodes of type $k$ and $k-1$. 
Conversely, let $\mathcal{G}$ be a $(D+1)$-partite graph such that nodes of type $k$ are only connected to nodes of type $k-1$ and $k+1$, with a biadjacency matrix $\delta_k$ (for $1 \leq k \leq D$). 
Then, if the biadjacency matrices satisfy $\delta_{k-1} \circ \delta_k = 0$, we can associate a chain complex to the graph, by defining the vector spaces as $C_k=\mathbb{F}_2^{n_k}$ and the boundary maps as $\delta_k$, where $n_k$ is the number of nodes of type $k$ in $\mathcal{G}$. 

A particular type of chain complex arises when considering the cellulation of a manifold. Given a smooth $D$-dimensional manifold $\mathcal{M}$, a \textbf{cellulation} is a covering of $\mathcal{M}$ by a set of $D$-dimensional polytopes $\{P_1,\ldots,P_n\}$, sending each polytope to $\mathcal{M}$ using injective homeomorphisms $\phi_i:P_i\rightarrow \mathcal{M}$, such that for all $i,j$, $i \neq j$,  $\phi(P_i)$ and $\phi(P_j)$ are either disjoint or overlap on exactly one of their $k$-faces for any $k$ \cite{breuckmann2018phdthesis}. 
The $k$-faces of a $D$-dimensional polytope refer to its $k$-dimensional elements: vertices ($0$-faces), edges ($1$-faces), faces ($2$-faces), and so on. 
The image of a $k$-face of a polytope $P_i$ by $\phi_i$ is called a $k$-cell of the cellulation. To the cellulation of a manifold we can associate a chain complex, built from a multipartite graph called the \textbf{Hasse diagram} of the cellulation. It is defined as follows: each $k$-cell defines a node of type $k$, and there is an edge between a node of type $k$ and a node of type $k-1$ if and only if the associated $k$-cell contains the associated $(k-1)$-cell. The Hasse diagram is characterized by the following important lemma \cite[Lemma 2.9]{breuckmann2018phdthesis}.
\begin{lemma}
    Consider a manifold cellulation and let $c_{i+1}$ and $c_{i-1}$ be $(i+1)$- and $(i-1)$-cells respectively. The number of $i$-cells connected to both $c_{i+1}$ and $c_{i-1}$ in the Hasse diagram is either zero or two.
\end{lemma}
Using this lemma, it is easy to prove that the Hasse diagram has the right properties to define a chain complex, namely that successive biadjacency matrices $\delta_k$ and $\delta_{k-1}$ satisfy $\delta_{k-1} \circ \delta_k = 0$. 
Intuitively, this is equivalent to the statement that in a manifold cellulation, the boundary of the boundary of a cell is always zero.
We call such a chain complex a cell complex.
\begin{definition}[Cell complex]
    A cell complex is a chain complex obtained from the Hasse diagram of the cellulation of a manifold. Each vector space $C_k$ of this complex is naturally equipped with a basis, corresponding to the $k$-cells of the cellulation.
\end{definition}

Another type of chain complex we will use throughout this work is the simplicial complex.
\begin{definition}[Simplicial complex]
    A $D$-dimensional simplicial complex is a length-$D$ chain complex equipped with a basis of $k$-cells, also called $\bm{k}$\textbf{-simplices}, such that every $k$-simplex is connected to exactly $k$ $(k-1)$-simplices in the associated multipartite graph. 
\end{definition}
A particular kind of simplicial complex emerges from the cellulation of a manifold by simplices, that is, when all the polytopes $P_i$ are simplices. We will call such a complex a \textbf{simplicial cell complex}.

Equipped with those definitions, we are now ready to define colour codes and pin codes.

\begin{definition}[Colour code lattice]
    Consider a length-$D$ cell complex 
    \begin{equation*}
        \mathcal{C} = C_D \overset{\delta_D}{\rightarrow} C_{D-1} ~ ... ~ C_1 \overset{\delta_1}{\rightarrow} C_0
    \end{equation*}
    The subcomplex
    \begin{equation*}
        \mathcal{C}_{(1,0)} = C_1 \overset{\delta_{1}}{\rightarrow} C_0
    \end{equation*}
    can be identified with a simple (no self-loops or multi-edges) undirected graph $\mathcal{G}$. If $\mathcal{G}$ is $(D+1)$-regular and admits a $(D+1)$-colouring of its edges then we will call $\mathcal{C}_{(1,0)}$ a \textbf{colour code lattice}.
\end{definition}

We will use the term ``colour code lattice'' to refer interchangeably to the length-1 complex $C_{(1,0)}$ or the to the colourable graph $\mathcal{G}$ associated to it. In the case that $\mathcal{C}_{(1,0)}$ is a colour code lattice we can define a colour code~\cite{bombin2006topological, kubica2018abcs} on $\mathcal{C}$ by choosing a pair of integers $x,z$ such that
\begin{align}
    \label{eq:x-z-constraints}
    \begin{split}
        2 \leq x,z \leq D \\
        x+z \geq D+2
    \end{split}
\end{align}
and then assigning a qubit to each $0$-cell of $\mathcal{C}$ and an $X$/$Z$ stabiliser to each $x$-/$z$-cell. An example is shown in \cref{fig:2dcc}. Notice that the colouring of the $1$-cells also induces a colouring on the $D$-cells. The colourability and valency conditions of $\mathcal{G}$, as well as the fact the complex comes from a manifold, ensure that this forms a valid CSS code. 
When taking $x=D$ and $z=2$ colour codes admit transversal diagonal gates of the form 
$R_D = \mathrm{diag}(1,e^{i\pi/2^{D-1}})$, which belong to the $D$-level of the Clifford hierarchy~\cite{bombin2007topological, bombin2015gauge, kubica2015universal}.

An alternative way to define a $D$-dimensional colour code is using a $D$-dimensional simplicial cell complex $\mathcal{K}$ that admits a $(D+1)$-dimensional colouring of its $0$-simplices (i.e. $0$-simplices that are part of the same $1$-simplex must have different colours). To see that this is true, and that these two methods of defining colour codes are equivalent, notice that the cellulation described by $\mathcal{K}$ is dual~\footnote{recall that the dual of a cellulation replaces $D$-cells with $0$-cells, $(D-1)$-cells with $1$-cells and so on.} to a cellulation described by some $\mathcal{C}$ where $\mathcal{C}_{(1,0)}$ is a colour code lattice. This is because the valency condition of $\mathcal{G}$ means that $0$-cells of $\mathcal{C}$ are mapped to $D$-simplices of $\mathcal{K}$, and the $(D+1)$-colouring of the $D$-cells of $\mathcal{C}$ translates to a $(D+1)$-colouring of the $0$-simplices of $\mathcal{K}$. The colour code defined on $\mathcal{K}$ then has qubits on $D$-simplices and stabilisers on $(D-x)$- and $(D-z)$-simplices (in the sense that they are supported on all $D$-simplices that contain this $(D-x)$ or $(D-z)$ simplex). An example of this duality is also shown in \cref{fig:2dcc}. 

\begin{figure}
    \centering
    \includegraphics[width=.35\textwidth]{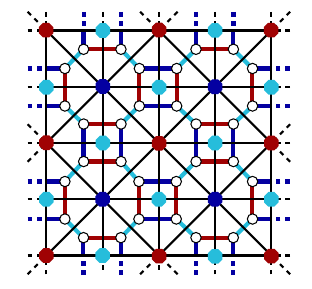}
    \caption{A square-octagon lattice can be used to define a colour code in two dimensions. Colours can be assigned to both edges and faces so that $c$-coloured faces are made from non-$c$-coloured edges. Qubits are assigned to vertices and stabilisers to faces. In the dual picture each qubit is assigned to a triangle ($2$-simplex) and stabilisers to vertices of these triangles (which inherit a colouring from the lattice).}
    \label{fig:2dcc}
\end{figure}

The definition via simplicial complexes is useful as a simplicial complex with $(D+1)$-colourable vertices can be easily obtained from any length $D$ chain complex via the following procedure. First, a $0$-simplex is associated to every $i$-cell $c_i$ of the chain complex. Next a $1$-simplex is associated to every pair of cells $c_i$ and $c_j$ such that $c_i \in c_j$, where we define this inclusion to mean that $i < j$ and there exists some length $j-i$ sequence
\begin{equation}
    c_i \in \delta_{i+1}(c_{i+1}), ~ c_{i+1} \in \delta_{i+2}(c_{i+2}), ~ ..., ~ c_{j-1} \in \delta_j(c_j).
\end{equation}
Intuitively this means that $c_i$ is a part of some higher dimensional cell $c_j$. For instance, if $c_2$ is a square then $c_1 \in c_2$ means $c_1$ is an edge of this square while $c_0 \in c_2$ means $c_0$ is a vertex of this square. Higher dimensional simplicies are similarly defined, so that an $m$-simplex is associated to a set of $m+1$ different cells $c_i,c_j,...,c_k$ where $c_i \in c_j \in ... \in c_k$. $D$-simplices are then associated with sets containing exactly one cell of each dimension such that $c_i \in c_{i+1}$ for all $0 \leq i < D$. These $D$-simplices are sometimes referred to as \textit{flags}, and the $0$-simplices of a flag are naturally $(D+1)$-coloured by the $(D+1)$ different dimensions of $i$-cell to which they correspond. In fact, the simplices shown in \cref{fig:2dcc} are the flags of a square lattice where each red $0$-simplex corresponds to a vertex of this lattice, each light blue $0$-simplex corresponds to an edge and each dark blue $0$-simplex to a square face. 

This method of constructing simplicial complexes was used in \cite{vuillot2022quantum} to obtain generalised colour codes called \textit{pin codes}. In the language of pin codes the set of all $D$-simplices containing a given $k$-simplex is called a $k$-pinned set, and stabilisers are associated to $(D-x)$- and $(D-z)$-pinned sets where $x$ and $z$ are subject to the same constraints as in \cref{eq:x-z-constraints}. These codes correspond exactly to colour codes if the simplicial complex from which they are defined corresponds to a cellulation of a manifold, but in general this does not need to be the case. In particular, the hypergraph product of classical codes can be used to obtain a chain complex and then a simplicial complex via the method described previously. Pin codes defined on these complexes were observed to have very high encoding rates but only constant distance in most cases. In addition, the transversal non-Clifford gate of high-dimensional colour codes were not typically inherited by these pin codes prior to additional modifications. 

In the following sections we introduce a further generalisation of pin codes that addresses both of these issues. 

\subsection{The simplex graph}
\label{subsection:rainbow_simplex}

\begin{figure*}
    \centering
    \includegraphics[width=.9\textwidth]{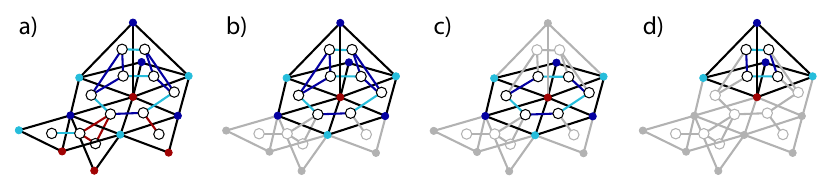}
    \caption{a) A $2$-dimensional simplicial complex that does not correspond to a cellulation of a manifold and its associated simplex graph. Vertices of the simplex graph can be part of more than one edge of the same colour. b) A $2$-maximal and c), d) $2$-rainbow subgraphs of the simplex graph, along with the associated parts of the complex.}
    \label{fig:nonmanifold_subgraphs}
\end{figure*}

The limitations of pin codes can be most easily understood if we think of them as arising not from a simplicial complex but from a closely related object that we call the \textit{simplex graph}. 

\begin{definition}
    Given a $D$-dimensional simplicial complex, $\mathcal{K}$, and a $(D+1)$-colouring of the $0$-simplices of this complex, the corresponding \textbf{simplex graph}, $\mathcal{G}_\mathcal{K}$, is the graph such that
    \begin{itemize}
        \item there is one vertex for every $D$-simplex
        \item there is an edge between vertices whenever the corresponding $D$-simplices share a $(D-1)$-simplex
        \item each edge has a colour $c$, which is the unique colour of $0$-simplex not contained in the $(D-1)$-simplex associated with this edge.
    \end{itemize}
\end{definition}

Note that to avoid confusion we will use the terms ``vertex'' and ``edge'' only when referring to objects of $\mathcal{G}_\mathcal{K}$. Vertices and edges of $\mathcal{K}$ will always be referred to as $0$- and $1$-simplices respectively. 

Perhaps the simplest example of a simplex graph is a 2D colour code lattice like \cref{fig:2dcc}, where we have a $2$-dimensional simplicial complex with a $3$-colouring defined on the $0$-simplices. There is then one graph/lattice vertex for each $2$-simplex and edges join these vertices whenever the intersection of the corresponding $2$-simplices is a $1$-simplex. Each such intersection contains two colours of $0$-simplex and the corresponding edge of the simplex graph is coloured with the third colour. More generally we have the following fact

\begin{lemma}
    \label{lemma:manifold_cellulation}
    Any $D$-dimensional simplicial cell complex with $(D+1)$-colourable $0$-simplices has its associated simplex graph correspond to a colour code lattice.
\end{lemma}

\begin{proof}
    Because the simplicial cell complex cellulates a $D$-dimensional manifold, a unique pair of $D$-simplices must meet at each $(D-1)$-simplex. The number of $(D-1)$-simplices contained in a $D$-simplex is $D+1$, so each vertex in the simplex graph is $(D+1)$-valent. A natural $(D+1)$-colouring of the edges of the graph is inherited from the colouring of the simplicial complex. These are exactly the requirements for a graph that can be identified with a colour code lattice. 
\end{proof}

For simplicial complexes that do not correspond to cellulations of manifolds there is no limit on how many $D$-simplices can meet at a $(D-1)$-simplex and so vertices in the simplex graph are not generally $(D+1)$-valent and can be part of multiple edges of the same colour. An example of such a graph is shown in \cref{fig:nonmanifold_subgraphs}.

Given some subset of colours $S = \{c_{i_1},...,c_{i_k}\}$, where $1 \leq k \leq D$, we can define two important types of subgraph of the simplex graph.

\begin{definition}
    An \textbf{$S$-maximal} subgraph, $\mathcal{M}^k$, of $\mathcal{G}_\mathcal{K}$ is a connected subgraph containing only edges with colours in $S$, and such that no more edges of colour in $S$ can be added to the subgraph while maintaining connectivity.
\end{definition}

\begin{definition}
    An \textbf{$S$-rainbow} subgraph, $\mathcal{R}^k$, of $\mathcal{G}_\mathcal{K}$ is a $k$-regular connected subgraph where each vertex is part of exactly one edge of each colour in $S$.
\end{definition}

We will also use the terms $k$-maximal and $k$-rainbow subgraph to refer to $S$-maximal or $S$-rainbow subgraphs for any choice of $S$ of size-$k$. Examples of both types of subgraph are shown in \cref{fig:nonmanifold_subgraphs}. Notice that $1$-maximal subgraphs must be cliques (fully connected subgraphs), as if multiple $D$-simplices meet at a common $(D-1)$-simplex their corresponding vertices in $\mathcal{G}_\mathcal{K}$ must all be connected by edges of the same colour. 

These subgraphs can be used to define both colour codes and pin codes. For a colour code (i.e. a code defined from a simplicial cell complex) $k$-maximal and $k$-rainbow subgraphs are equivalent since each vertex in the simplex graph is part of only one edge of each colour. Stabilisers of the colour code can then be thought of as being assigned to either $x$- and $z$-maximal subgraphs, $x$- and $z$-rainbow subgraphs or any mix of the two. For instance, every hexagonal face of the colour code lattice shown in \cref{fig:2dcc} is both a $2$-maximal and $2$-rainbow subgraph. To see the connection to pin codes, notice that the set of vertices contained in a $k$-maximal subgraph is equivalent to the set of $D$-simplices contained in a $(D-k)$-pinned set. This is because vertices in the simplex graph connected by $c_i$-coloured edges correspond to $D$-simplices which differ only by a $c_i$-coloured $0$-simplex. As a result, an $S$-maximal subgraph corresponds to a maximal set of $D$-simplices which differ from each other only by $0$-simplices with colours in $S$, and so there is a common $(D-k)$-simplex shared by all $D$-simplices in this set. For example, all $2$-simplices associated with the $2$-maximal subgraph shown in \cref{fig:nonmanifold_subgraphs} b) have a common $(2-2=0)$-simplex. We can therefore view pin codes as being defined by the assignment of $X$ and $Z$ stabilisers to all $x$- and $z$-maximal subgraphs of a simplex graph, and thus they are equivalent to colour codes when this simplex graph corresponds to a $D$-dimensional lattice.

An important property of pin codes (proposition 1 of \cite{vuillot2022quantum}) is that the nontrivial intersection of a pair of pinned sets is another pinned set. The equivalent result in terms of the simplex graph is as follows

\begin{lemma}
    \label{lemma:mm_intersection}
    The (non-empty) intersection of an $S_1$-maximal subgraph, $\mathcal{M}^{\lvert S_1 \rvert}$, and an $S_2$-maximal subgraph, $\mathcal{M}^{\lvert S_2 \rvert}$, is an $(S_1 \cap S_2)$-maximal subgraph (or multiple such subgraphs).
\end{lemma}

\begin{proof}
    Given any vertex in the intersection of $\mathcal{M}^{\lvert S_1 \rvert}$ and $\mathcal{M}^{\lvert S_2 \rvert}$, all edges of $\mathcal{G}_\mathcal{K}$ containing this vertex and with colours in $S_1 \cap S_2$ must be in both $\mathcal{M}^{\lvert S_1 \rvert}$ and $\mathcal{M}^{\lvert S_2 \rvert}$ and thus the intersection of these subgraphs is an $(S_1 \cap S_2)$-maximal subgraph.
\end{proof}

We can also show a similar result for intersections of maximal subgraphs and rainbow subgraphs

\begin{lemma}
    \label{lemma:mr_intersection}
    The (nontrivial) intersection of an $S_1$-maximal subgraph, $\mathcal{M}^{\lvert S_1 \rvert}$, and an $S_2$-rainbow subgraph, $\mathcal{R}^{\lvert S_2 \rvert}$, is an $(S_1 \cap S_2)$-rainbow subgraph (or multiple such subgraphs). 
\end{lemma}

\begin{proof}
    Given any vertex in the intersection of $\mathcal{M}^{\lvert S_1 \rvert}$ and $\mathcal{R}^{\lvert S_2 \rvert}$, all edges of $\mathcal{G}_\mathcal{K}$ containing this vertex and with colours in $S_1 \cap S_2$ must be in $\mathcal{M}^{\lvert S_1 \rvert}$ and exactly one of each colour is in $\mathcal{R}^{\lvert S_2 \rvert}$ and thus the intersection of these subgraphs is an $(S_1 \cap S_2)$-rainbow subgraph.
\end{proof}

In general there is no restriction on the possible intersection of two rainbow subgraphs, and it is fairly straightforward to construct examples of simplex graphs where $x$- and $z$-rainbow subgraphs intersect at a single vertex. 

It will also be useful to define the following operation on $S$-maximal and $S$-rainbow subgraphs, which allows us to decompose them into collections of maximal/rainbow subgraphs containing fewer colours.

\begin{definition}
    Given an $S$-maximal subgraph $\mathcal{M}^{\lvert S \rvert}$ and a set $T \subseteq S$, the $T$-division of $\mathcal{M}^{\lvert S \rvert}$, denoted $\mathcal{M}^{\lvert S \rvert}/T$, is the graph obtained by removing all edges with colours in $T$ from $\mathcal{M}^{\lvert S \rvert}$. Due to the definition of $S$-maximal subgraphs, $\mathcal{M}^{\lvert S \rvert}/T$ must be a collection of disjoint $(S \setminus T)$-maximal subgraphs, i.e,
    \begin{equation}
        \mathcal{M}^{\lvert S \rvert}/T = \mathcal{M}_1^{\lvert S \setminus T \rvert} \cup ... \cup \mathcal{M}_m^{\lvert S \setminus T \rvert}
    \end{equation}
    An analogous operation is defined for $S$-rainbow subgraphs, in which case we have
    \begin{equation}
        \mathcal{R}^{\lvert S \rvert}/T = \mathcal{R}_1^{\lvert S \setminus T \rvert} \cup ... \cup \mathcal{R}_m^{\lvert S \setminus T \rvert}
    \end{equation}
\end{definition}

At this point readers are likely asking what (if any) role rainbow subgraphs play in pin codes. The answer to this question comes from the following key result.

\begin{proposition}
    \label{prop:even_intersection}
    An $x$-maximal and $z$-rainbow subgraph (or vice versa) with non-trivial intersection must share an even number of vertices.
\end{proposition}

\begin{proof}
   By \cref{lemma:mr_intersection}, for some $\mathcal{M}^x$ and $\mathcal{R}^z$ with nontrivial intersection we have $\mathcal{M}^x \cap \mathcal{R}^z = \mathcal{R}_1^m \cup \mathcal{R}_2^m \cup ...$ for some $m$. Because $x+z \geq D+2$ (by \cref{eq:x-z-constraints}) but there are only $D+1$ different colours in the graph the two subgraphs must have at least one colour of edge in common, and so $m \geq 1$. If $m > 1$ we can choose some $T$ such that each $\mathcal{R}_i^m/T$ is a collection of disjoint $1$-rainbow subgraphs (which are just single edges connecting pairs of vertices), and otherwise $\mathcal{R}_i^m$ itself is a collection of disjoint $1$-rainbow subgraphs. $\mathcal{M}^x \cap \mathcal{R}^z$ must therefore contain an even number of vertices.
\end{proof}

An immediate consequence of this is that

\begin{corollary}
    $x$- and $z$-rainbow subgraphs can support pin code logical operators.
\end{corollary}

From \cref{prop:even_intersection} we can see that an $X$ operator supported on the vertices of an $x$-rainbow subgraph commutes with all $Z$ stabilisers (and similarly for $Z$ operators on $z$-rainbow subgraphs) and so if these operators are not stabilisers they must be logicals. We will see in \cref{section:rainbow_hgp} that this fact is the cause of the low distances observed for the majority of pin codes studied numerically in~\cite{vuillot2022quantum}, as the method used to generate these codes naturally gives rise to rainbow subgraphs of size $4$. In addition, in \cref{section:logical_gates} we will show that it prevents useful transversal non-Clifford gates in the majority of cases. 

\subsection{Rainbow codes}
\label{subsection:rainbow_rainbow}

A straightforward solution to the low distances observed in pin codes is to add operators supported on rainbow subgraphs to the stabiliser group. This leads us beyond pin codes to a more general code family that we call \textit{rainbow codes}. 

\begin{definition}
    A $D$-dimensional \textbf{rainbow code} is a CSS code with qubits associated to the vertices of a simplex graph with $D+1$ colours and stabiliser generators associated to a subset of the $x$-maximal and $x$-rainbow subgraphs (for $X$ stabilisers) and $z$-maximal and $z$-rainbow subgraphs (for $Z$ stabilisers) of this graph, for $x$ and $z$ satisfying \cref{eq:x-z-constraints}.
\end{definition}

This definition is extremely broad and includes both colour codes and pin codes, as well as many less interesting codes such as the trivial code with empty stabiliser group. In addition to these we define the following classes of rainbow codes.

\begin{definition}
    A \textbf{generic} rainbow code has $X$ stabilisers supported on all $x$-maximal subgraphs and $Z$ stabilisers supported on all $z$-rainbow subgraphs. 
\end{definition}

\begin{definition}
    An \textbf{anti-generic} rainbow code has $X$ stabilisers supported on all $x$-rainbow subgraphs and $Z$ stabilisers supported on all $z$-maximal subgraphs.
\end{definition}

\begin{definition}
    A \textbf{mixed} rainbow code has $X$ stabilisers supported on both $x$-maximal and $x$-rainbow subgraphs and $Z$ stabilisers supported on both $z$-maximal and $z$-rainbow subgraphs.
\end{definition}

Note that in general mixed rainbow codes are not required to have $X$ and $Z$ stabilisers supported on \textit{all} $x$- and $z$-rainbow subgraphs as such operators would not commute in general. The definitions of these different classes of rainbow code are summarised in \cref{tab:rainbow_classes}. Notice that all of these classes can be thought of as generalisations of colour codes as maximal and rainbow subgraphs in colour codes are equivalent. To understand how the properties of these classes can differ from each other we will need to look at some more specific examples, such as those presented in the next section. We also remark that, for all stabiliser assignments we will consider, the sets of all relevant maximal and rainbow subgraphs are efficiently computable and algorithms for finding them are presented in \cref{app:algorithms}.

\begin{table}
    \centering
    \begin{tabular}{c|c|c}
          & $X$-stabilisers & $Z$-stabilisers \\
         \hline
         pin & all $x$-maximal & all $z$-maximal \\
         generic & all $x$-maximal & all $z$-rainbow \\
         anti-generic & all $x$-rainbow & all $z$-maximal \\
         mixed & \makecell{some $x$-maximal \\ and $x$-rainbow} & \makecell{some $z$-maximal \\ and $z$-rainbow}
    \end{tabular}
    \caption{Stabiliser assignments to maximal and rainbow subgraphs in four classes of rainbow code. All four classes can be viewed as generalisations of colour codes.}
    \label{tab:rainbow_classes}
\end{table}

Finally, we comment on the stabiliser weights of these codes. These weights are upper-bounded by the size of the largest $D$-maximal subgraph, or equivalently by the largest number of $D$-simplices sharing a common $0$-simplex. In general this can be arbitrarily large but in more specific cases it can be bounded, such as in the case of the codes considered in the next section.

\section{Rainbow codes from hypergraph products and gluings of colour codes}
\label{section:rainbow_hgp}
Given a cell complex that defines a cellulation of a manifold, a colour code on this manifold can be defined using a mapping from the chain complex to a corresponding simplicial complex~\cite{bombin2007exact,zhu2023non}. Similar techniques were used in \cite{vuillot2022quantum} to define pin codes on more general chain complexes obtained from the hypergraph product of classical codes~\cite{tillich_quantum_2014}, and we could use these mappings and the results of the previous section to obtain other classes of rainbow code in an identical fashion. However, it turns out that there is an alternative perspective on the hypergraph product -- and on the codes derived from it -- that provides considerably more intuition into the structure and properties of these codes. In what follows we will show that any $D$-dimensional rainbow code defined via the hypergraph product can be viewed as a collection of $D$-dimensional colour codes joined together at $(D-1)$-dimensional domain walls. This makes understanding the code parameters and action of logical gates much more straightforward, and also implies connections between these codes and other constructions involving some forms of code gluing, such as welded codes~\cite{michnicki20143d}, defect networks~\cite{song2023topological}, and layer codes~\cite{williamson_layer_2024}.

We begin by studying gluing operations on graphs and their interplay with the hypergraph product in \cref{subsection:hgp_glue}. In \cref{subsection:hgp_flag} we show how to obtain flag graphs from the output of the hypergraph product and demonstrate the effects of gluing operations on these graphs. In \cref{subsection:hgp_join} we show how, when all inputs to the product are cycles, the resulting flag graph can be understood as a collection of colour code lattices joined at objects we call \textit{seams}, and finally we show how to assign stabilisers to these seams and study their effect on colour code logical operators in \cref{subsection:hgp_seam}.

\subsection{Gluings of graphs}
\label{subsection:hgp_glue}

To start, we define the hypergraph product (HGP) and show how its action commutes with gluing operations on graphs. This product is commonly defined in terms of tensor products of chain complexes but it can also be understood as a Cartesian product of bipartite graphs~\footnote{If we label the two kinds of vertex type-$0$ and type-$1$ then the adjacency matrix with type-$0$ vertices as rows and type-$1$ vertices as columns is equivalent to a representation of the boundary map from $1$-cells to $0$-cells in the chain complex picture}. Recall that a graph $\mathcal{G}$ is a set of vertices $V_\mathcal{G}$ and a set of edges $E_\mathcal{G}$, where an edge is an (unordered) pair of vertices $\{g_1,g_2\}$. Also, the Cartesian product of two graphs is defined as 

\begin{definition}
    Given two graphs $\mathcal{G} = (V_\mathcal{G},E_\mathcal{G})$ and $\mathcal{H} = (V_\mathcal{H}, E_\mathcal{H})$, the Cartesian product $\mathcal{G} \square \mathcal{H}$ is a graph with vertex set given by the Cartesian product of sets, $V_\mathcal{G} \times V_\mathcal{H}$, and edges between vertices $(g_1,h_1)$ and $(g_2,h_2)$ iff
    \begin{itemize}
        \item $g_1 = g_2$ and $\{h_1,h_2\} \in E_\mathcal{H}$ or 
        \item $h_1 = h_2$ and $\{g_1,g_2\} \in E_\mathcal{G}$
    \end{itemize}
\end{definition}

We then define the neighbourhood of a vertex, $\mathcal{N}(g_1)$, to be the set of all vertices $g_i \in V_\mathcal{G}$ such that $\{g_1,g_i\} \in E_\mathcal{G}$, and the \textit{gluing} of two vertices to be the following operation.

\begin{definition}
    Given two vertices $g_1$ and $g_2$ of a graph $\mathcal{G}$, where $g_1 \notin \mathcal{N}(g_2)$, the \textbf{gluing} $g_1 \leftarrow g_2 : \mathcal{G}$ removes $g_2$ from $V_\mathcal{G}$ and replaces each edge $\{g_2,g_i\} \in E_\mathcal{G}$ with an edge $\{g_1,g_i\}$.
\end{definition}

Gluing is both associative and commutative, and so we can think of multiple gluings happening simultaneously without needing to specify an order. 

\begin{lemma}
    \label{lemma:gluing_assoc_comm}
    Gluing is associative and commutative. 
\end{lemma}

\begin{proof}
    We can consider $g_1 \leftarrow g_2$ to be an operation on the neighbourhoods of $g_1$ and $g_2$. Specifically, it maps $\mathcal{N}(g_1)$ to $\mathcal{N}(g_1) \cup \mathcal{N}(g_2)$ and $\mathcal{N}(g_2)$ to the empty set, with all other neighbourhoods preserved up to relabellings. Since the union of sets is associative and commutative gluing also has both of these properties. 
\end{proof}

This lets us define \textit{gluing along a line} in a graph obtained from the Cartesian product in the following way

\begin{definition}
    Given a graph $\mathcal{G} \square \mathcal{H}$, the operation $g_1 \Leftarrow g_2 : \mathcal{G} \square \mathcal{H}$ is a gluing $((g_1,h_i) \leftarrow (g_2,h_i) : \mathcal{G} \square \mathcal{H}) ~ \forall ~ h_i \in V_\mathcal{H}$
\end{definition}

\noindent where the order of these gluings does not matter by \cref{lemma:gluing_assoc_comm}. This operation generalises straightforwardly to \textit{gluing along a hyperplane} in a graph $\mathcal{G} = \mathcal{A} \square \mathcal{B} \square ...$ since we can use the associativity and commutativity of the Cartesian product to write 
\begin{equation}
    \mathcal{G} = \mathcal{F} \square (\mathcal{A} \square \mathcal{B} \square ...) = \mathcal{F} \square \mathcal{H}
\end{equation}
for any factor $\mathcal{F}$ in the product, and then use the above definition to glue as $f_1 \Leftarrow f_2 : \mathcal{F} \square \mathcal{H}$. Additionally, because any graph $\mathcal{G}$ can be written as the Cartesian product $\mathcal{G} \square \mathcal{I}$ (where $\mathcal{I}$ is the graph with a single vertex and no edges) we can define 
\begin{equation}
    (g_1 \Leftarrow g_2 : \mathcal{G}) := (g_1 \Leftarrow g_2 : \mathcal{G} \square \mathcal{I}) = (g_1 \leftarrow g_2 : \mathcal{G}).
\end{equation}
We can then prove the following key result

\begin{figure}
    \centering
    \includegraphics[width=\linewidth]{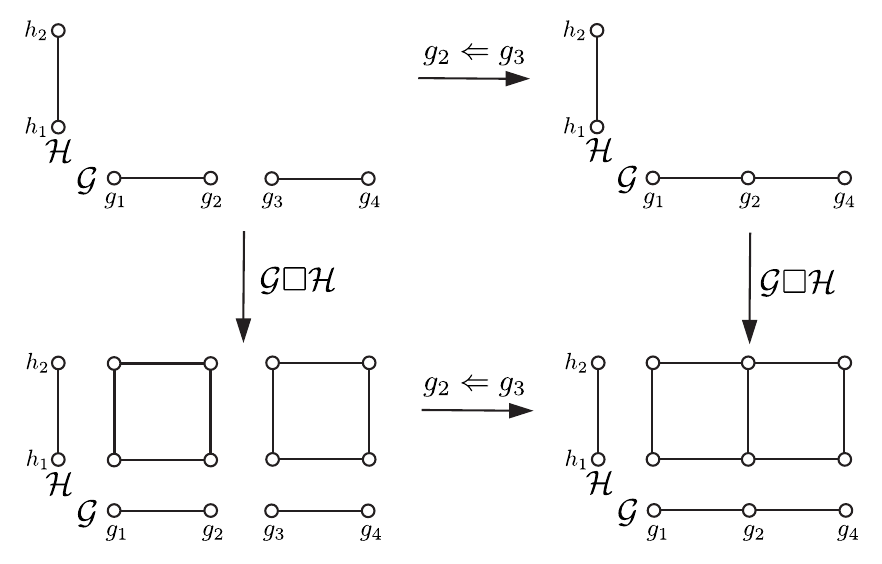}
    \caption{Commutativity of hyperplane gluing and Cartesian product.}
    \label{fig:commute}
\end{figure}

\begin{proposition}
    \label{prop:glue_prod_commute}
    Hyperplane gluing commutes with the Cartesian product, i.e.,
    \begin{equation}
        (g_1 \Leftarrow g_2 : \mathcal{G} \square \mathcal{H}) = (g_1 \Leftarrow g_2 : \mathcal{G} ) \square \mathcal{H}
    \end{equation}
\end{proposition}

\begin{proof}
    First we can show that the two graphs have the same vertex set. The vertex set of $\mathcal{G} \square \mathcal{H}$ is $V_\mathcal{G} \times V_\mathcal{H}$ and the gluing $g_1 \Leftarrow g_2 : \mathcal{G} \square \mathcal{H}$ removes from this set all vertices of the form $\{g_2, h_j\}$, giving a vertex set
    \begin{equation}
        \{\{g_i, h_j\} ~ | ~ (g_i \neq g_2) \in V_\mathcal{G}, h_j \in V_\mathcal{H}\}.
    \end{equation}
    On the other hand, $g_1 \Leftarrow g_2 : \mathcal{G}$ results in a vertex set $V_\mathcal{G}' = \{g_i ~ | ~ (g_i \neq g_2) \in V_\mathcal{G}\}$ and the product $(g_1 \Leftarrow g_2 : \mathcal{G}) \square \mathcal{H}$ then has vertex set 
    \begin{equation}
        \{\{g_i, h_j\} ~ | ~ g_i \in V_\mathcal{G}', h_j \in V_\mathcal{H}\}
    \end{equation}
    which is equivalent to the above. 

    Secondly, we can show that each vertex has the same neighbourhood in both graphs (equivalent to showing that the edge sets are equivalent). First, notice that for a vertex $(g_i,h_j)$ of $\mathcal{G} \square \mathcal{H}$, the neighbourhood $\mathcal{N}((g_i,h_j))$ can be split into two parts, $N^g$ and $N^h$, where $N^g = \mathcal{N}(g_i) \times h_j$ and $N^h = g_i \times \mathcal{N}(h_j)$. Then, recalling the effect of gluing on neighbourhoods described in the proof of \cref{lemma:gluing_assoc_comm}, for $g_1 \Leftarrow g_2 : \mathcal{G} \square \mathcal{H}$ we have that
    \begin{equation}
        \begin{split}
            \mathcal{N}((g_1,h_j)) \mapsto (\mathcal{N}(g_1) &\times h_j) \cup (g_1 \times    \mathcal{N}(h_j)) \\
            &\cup (\mathcal{N}(g_2) \times h_j) \cup (g_2 \times \mathcal{N}(h_j))
        \end{split}
    \end{equation}
    but because $g_1 \Leftarrow g_2 : \mathcal{G} \square \mathcal{H}$ glues the vertices in the set $g_1 \times \mathcal{N}(h_j)$ to those in $g_2 \times \mathcal{N}(h_j)$ this is really
    \begin{equation}
        \mathcal{N}(g_i,h_j) \mapsto ((\mathcal{N}(g_1) \cup \mathcal{N}(g_2)) \times h_j) \cup (g_1 \times \mathcal{N}(h_j)) 
    \end{equation}
    where we have used the fact that, for sets $A,B,C$, $(A \times C) \cup (B \times C) = ((A \cup B) \times C)$. For $(g_1 \Leftarrow g_2 : \mathcal{G}) \square \mathcal{H}$, we instead have that 
    \begin{equation}
        \mathcal{N}(g_1) \mapsto \mathcal{N}(g_1) \cup \mathcal{N}(g_2)
    \end{equation}
    in $\mathcal{G}$, and then after taking the product with $\mathcal{H}$
    \begin{equation}
        \begin{split}
            \mathcal{N}((g_1,h_j)) &= N^g \cup N^h \\
            &= ((\mathcal{N}(g_1) \cup \mathcal{N}(g_2)) \times h_j) \cup (g_1 \times \mathcal{N}(h_j)) 
        \end{split}
    \end{equation}
    which is the same as above, completing the proof.
\end{proof}

\noindent An example of this commutation is shown in \cref{fig:commute}. 

We also define an opposite operation to gluing, which is \textit{ungluing}

\begin{definition}
    Given a vertex $g_1$ of a graph $\mathcal{G}$ and a set of edges $U \subseteq \mathcal{N}(g_1)$, an \textbf{ungluing} $g_1 \overset{U}{\leftrightarrow} g_2 : \mathcal{G}$ adds a new vertex $g_2$ to $V_\mathcal{G}$ and replaces all edges $\{g_1,g_i\} \in U$ with edges $\{g_2,g_i\}$
\end{definition}

Notice that unlike gluing, ungluing is not uniquely defined as we can distribute the edges originally connected to $g_1$ between $g_1$ and $g_2$ in multiple different ways. This means that for any choice of $U$ we are guaranteed to have
\begin{equation}
    \label{eq:unglue_glue}
    g_1 \leftarrow g_2 : (g_1 \overset{U}{\leftrightarrow} g_2 : \mathcal{G}) = \mathcal{G}
\end{equation}
but if we reverse the order then in general
\begin{equation}
    \label{eq:glue_unglue}
    g_1 \overset{U}{\leftrightarrow} g_2 : (g_1 \leftarrow g_2 : \mathcal{G}) \neq \mathcal{G}.
\end{equation}
Finally, we define ungluing along a hyperplane as
\begin{definition}
    Given a graph $\mathcal{G}\square\mathcal{H}$ and a set $U \in \mathcal{N}(g_1)$, the operation $g_1 \overset{U}{\Leftrightarrow} g_2 : \mathcal{G}\square\mathcal{H}$ is equivalent to the operation $(g_1 \overset{U}{\leftrightarrow} g_2: \mathcal{G})\square\mathcal{H}$.
\end{definition}

\noindent which commutes with the Cartesian product by definition. The fact that 
\begin{equation}
    \label{eq:unglue_prod_glue}
    \begin{split}
        &g_1 \Leftarrow g_2 : (g_1 \overset{U}{\Leftrightarrow} g_2 : \mathcal{G} \square \mathcal{H}) \\
        =~&g_1 \Leftarrow g_2 : ((g_1 \overset{U}{\leftrightarrow} g_2 : \mathcal{G}) \square \mathcal{H})\\
        =~&\mathcal{G} \square \mathcal{H}
    \end{split}
\end{equation}
is then guaranteed by \cref{prop:glue_prod_commute} and \cref{eq:unglue_glue}. 

The insight provided by \cref{eq:unglue_prod_glue} that instead of studying the product of a few very complex graphs, we can equivalently ``unglue'' these graphs into simpler graphs, take products of these and then glue together the results. However, we still need a way to define simplicial complexes from these graphs, as well as way to understand the effect of gluing on these complexes. 

\subsection{Flags and their gluings}
\label{subsection:hgp_flag}

The next step is to find a way to obtain simplex graphs from the graphs produced by the product, and understand how they are affected by gluing operations on the input graphs.

Given a bipartite graph $\mathcal{G}$ we can define two types of vertex, which we call level 0 and level 1, and write $g_i^0$ and $g_i^1$. In a graph $\mathcal{G} \square \mathcal{H}$ where $\mathcal{G}$ and $\mathcal{H}$ are both bipartite we then have three levels of vertex:

\begin{itemize}
    \item level 0 vertices are products of two level 0 vertices, $(g_i^0,h_j^0)^0$,
    \item level 1 vertices are products of a level 0 and a level 1 vertex, $(g_i^1,h_j^0)^1$ or $(g_i^0,h_j^1)^1$,
    \item level 2 vertices are products of two level 1 vertices, $(g_i^1,h_j^1)^2$.
\end{itemize}

More generally, for the Cartesian product of $D$ bipartite graphs we will have $D+1$ levels of vertex in the output, with the level of a given output vertex equal to the sum of the levels of all its corresponding input vertices. Typically we will be focused on a single input graph $\mathcal{G}$ (i.e. as the target of a gluing) and in this case we will use an abbreviated notation $(g_i^l,...)^L$ where $l \in \{0,1\}$, $L$ is the level of the product vertex and the $...$ represents input vertices from all other input graphs. A $D$-dimensional simplicial complex can then be obtained from this graph by associating a $D$-simplex with every path of length $D+1$ that consists of a sequence of vertices with levels $0,1,...,D$. In keeping with the terminology of \cite{vuillot2022quantum} we will refer to these paths as \textit{flags} and write them as $(D+1)$-uples $((g_i^0,...)^0,...,(g_{i'}^1,...)^D)$. Two examples for $D=2$ and $D=3$ are shown in \cref{fig:flags}. When discussing simplicial complexes obtained in this manner we will use the term ``flag graph'' (or $\mathcal{G}_\mathcal{F}$) to refer to the simplex graph and ``product graph'' (or $\mathcal{G}_\square$) to refer to the graph with levelled vertices obtained from the Cartesian product. Additionally, the vertex levels in the product graph will be used to label the edge colours in the flag graph, so that e.g. two flags which share product graph vertices with levels $1,2,...,D$ will correspond to flag graph vertices connected by an edge with colour $c_0$. 

\begin{figure}
    \centering
    \includegraphics[width=.8\linewidth]{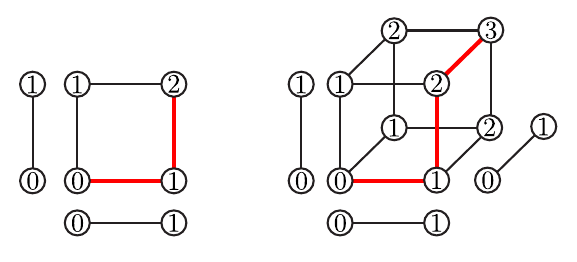}
    \caption{Examples of flags (paths of red edges) in Cartesian products of two and three graphs. Numbers show vertex levels.}
    \label{fig:flags}
\end{figure}

\begin{figure*}[b]
    \centering
    \includegraphics[width=.75\textwidth]{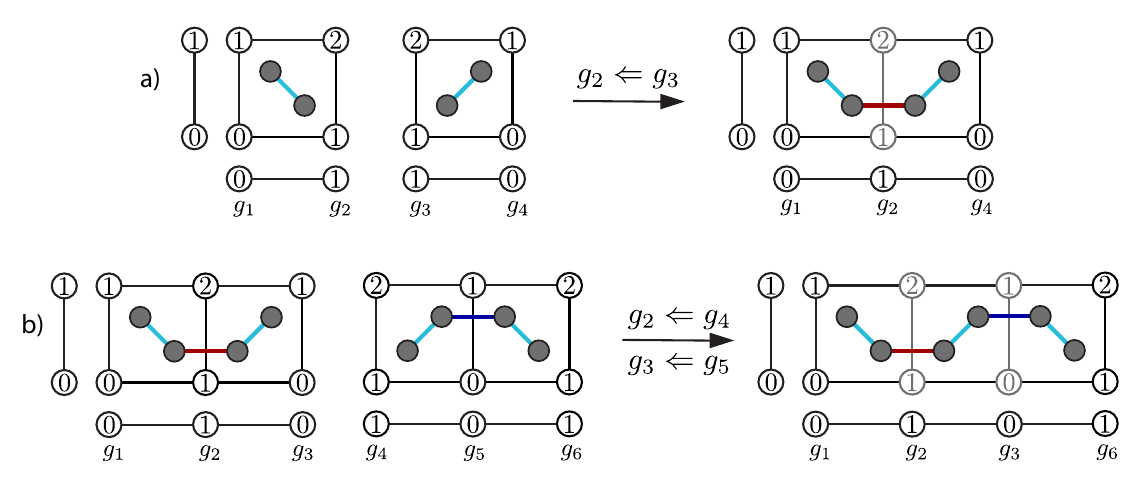}
    \caption{Gluings on product graphs (numbered vertices and black edges) and their effects on the corresponding flag graphs (dark vertices and coloured edges). a) A single gluing that joins two flags at $D-1$ vertices in the product graph and so adds a new $c_0$-coloured edge between the corresponding vertices of the flag graph. b) A pair of gluings that glue together pairs of flags in the product graph and so glue together the corresponding vertices in the flag graph. Seams in the product graph are shown in both cases by grey edges/vertices.}
    \label{fig:flag_gluings}
\end{figure*}

When considering gluings of bipartite graphs, bipartiteness is only preserved when gluings are performed between vertices of the same level, so we will consider only these types of gluings. The flag graph is modified by gluings of the product graph only when flags in the product graph are glued together at $D$ or $D+1$ of their vertices, with the former creating new edges in the flag graph and the latter gluing together vertices of the flag graph. Notice, however, that this second case cannnot be accomplished with a single hyperplane gluing as, for example, the gluing $g_i^0 \Leftarrow g_j^0$ can non-trivially affect at most $D$ vertices of any flag, as a flag where all $D$ vertices are of the form $(g_i^0,...)$ or $(g_j^0,...)$ would not contain a level $D$ vertex. On the other hand, we can identify two distinct cases where flags can be joined at $D$ vertices by a single gluing. Either we have a gluing $g_i^0 \Leftarrow g_j^0$ and a pair of flags
\begin{equation*}
    \begin{split}
       &F_1 = ((g_i^0,...)^0,...,(g_i^0,...)^{D-1},(g_{i'}^1,...)^D) \\
       &F_2 = ((g_j^0,...)^0,...,(g_j^0,...)^{D-1},(g_{j'}^1,...)^D),
    \end{split}
\end{equation*}
or we have a gluing $g_{i'}^1 \Leftarrow g_{j'}^1$ and a pair of flags
\begin{equation*}
    \begin{split}
        &F_1 = ((g_i^0,...)^0,(g_{i'}^1,...)^1,...,(g_{i'}^1,...)^D)\\
        &F_2 = ((g_j^0,...)^0,(g_{j'}^1,...)^1,...,(g_{j'}^1,...)^D).
    \end{split}
\end{equation*}
In the first case a new $c_D$ edge in the flag graph will be created between the vertices corresponding to $F_1$ and $F_2$, while in the second case a new $c_0$ edge will be created. No other colours of edge can be created by gluing. An example of such a gluing is shown in \cref{fig:flag_gluings} a). We then fix the following terminology. 

\begin{definition}
    Given a graph $g_i^l \Leftarrow g_j^l : \mathcal{G}_\square$ the \textbf{type-$l$ seam} associated to the gluing is the set of vertices of the form $(g_i^l,...)^L$. Additionally, a flag is said to lie on this seam if $D$ of its $D+1$ vertices are in the seam.
\end{definition}

Intuitively the seam is the set of vertices of $g_i^l \Leftarrow g_j^l : \mathcal{G}_\square$ that were modified by the gluing. We can then see that the new edges created by the gluing will always be between flags which lie on the seam. 

As alluded to above, it is possible to use multiple gluings to glue one flag fully onto another, and thus glue together a pair of vertices in the flag graph. This is accomplished using a pair of gluings $g_i^0 \Leftarrow g_j^0$ and $g_{i'}^1 \Leftarrow g_{j'}^1$ where $g_{i'}^1 \in \mathcal{N}(g_i^0)$ and $g_{j'}^1 \in \mathcal{N}(g_j^0)$. This creates a pair of seams, as in \cref{fig:flag_gluings} b). 

We also care about the action of ungluing on flag graphs. From \cref{eq:unglue_prod_glue} we know that $\mathcal{G}_\square$ and $g_1^l \Leftarrow g_2^l : (g_1^l \overset{U}{\Leftrightarrow} g_2^l : \mathcal{G}_\square)$ are identical and so have identical flag graphs. Since $g_1^l \Leftarrow g_2^l$ adds edges of colour $c_0$ or $c_D$, $g_1^l \overset{U}{\Leftrightarrow} g_2^l$ must therefore delete edges of colour $c_0$ or $c_D$. 

\subsection{Joining colour code lattices}
\label{subsection:hgp_join}

Now that we have a method for obtaining a simplex graph, the next step is to show that this graph is equivalent to a collection of colour code lattices which will be joined together by the gluing. 

Consider a set of $D$ cycle graphs of even length, $\{\mathcal{O}_1,\mathcal{O}_2,...,\mathcal{O}_D\}$. We will write the vertices of $\mathcal{O}_a$ as $o_{a,i}^l$, with $i$ and $l$ representing index and level as above (bipartiteness of the graph is guaranteed by the even length). The Cartesian product of these graphs, $\mathcal{G}_\square = \mathcal{O}_1 \square \mathcal{O}_2 \square ... \square \mathcal{O}_D$, is a $D$-dimensional hypercubic lattice on a $D$-dimensional torus, with each vertex having coordinates and level $(o_{1,i_1}^{l_1},o_{2,i_2}^{l_2},...,o_{D,i_D}^{l_D})^L$ with $L=l_1+\ldots+l_D$.

\begin{lemma}
    \label{lemma:cc_lattice}
    The flag graph $\mathcal{G}_\mathcal{F}$ of a graph $\mathcal{G}_\square = \mathcal{O}_1 \square \mathcal{O}_2 \square ... \square \mathcal{O}_D$ is a $D$-dimensional colour code lattice.
\end{lemma}

\begin{proof}
    Because $\mathcal{G}_\square$ is a cellulation of a $D$-dimensional manifold the simplicial complex described by $\mathcal{G}_\mathcal{F}$ is also a cellulation of this same manifold (it is simply a subdivision of the hypercubic lattice into simplices), and by \cref{lemma:manifold_cellulation} $\mathcal{G}_\mathcal{F}$ is then a $D$-dimensional colour code lattice. 
\end{proof}

\begin{figure*}
    \centering
    \includegraphics[width=\textwidth]{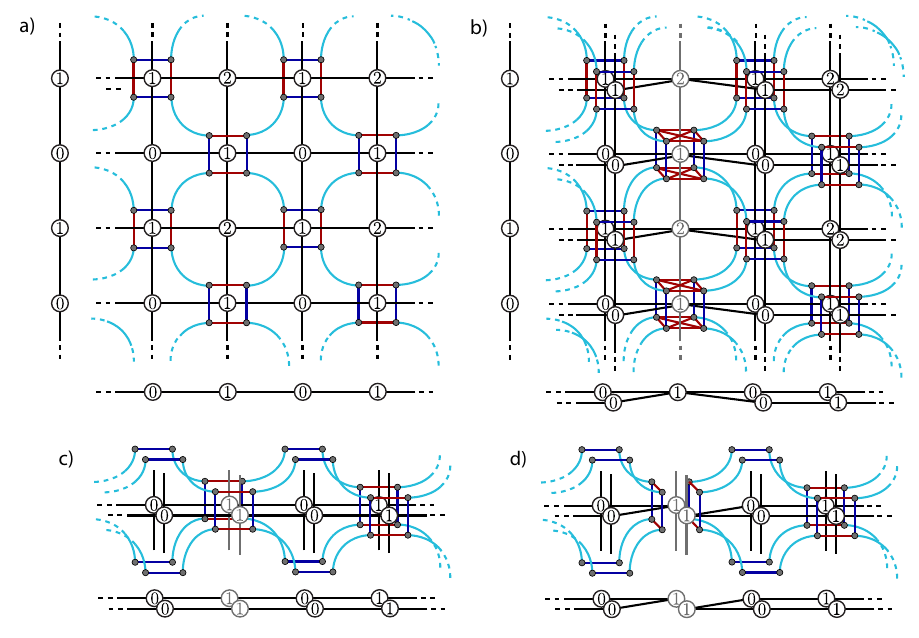}
    \caption{a) A flag graph equivalent to a 2D colour code lattice produced from the product of a pair of cycles. b) A pair of 2D colour code lattices joined at a type-1 seam. c) and d) Subsections of split-equivalent lattices obtained from different splits of the lattice in b). In c) we have a pair of disjoint lattices and in d) we have a single large lattice.}
    \label{fig:glued_ccs}
\end{figure*}

An example of this is shown in \cref{fig:glued_ccs} a). This statement is equivalent to results in other works where $\mathcal{G}_\square$ is interpreted as a chain complex rather than a graph~\cite{bombin2007exact,vuillot2022quantum}. 

\begin{figure*}
    \centering
    \includegraphics[width=\textwidth]{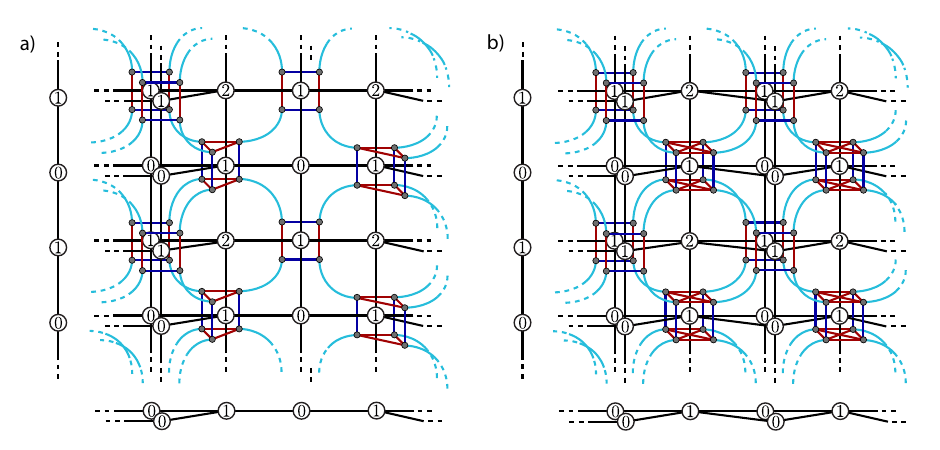}
    \caption{a) Flag graph obtained from a gluing of two cycles at a set of common edges, resuling in a pair of non-splittable seams. b) Flag graph obtained from a gluing of three cycles at a common set of edges resulting in a pair of splittable seams. This graph can be split into two (but not three) colour code lattices.}
    \label{fig:glued_css_2}
\end{figure*}

Now consider a graph $\mathcal{G}_\square = (\mathcal{O}_1 \cup \mathcal{O}_2) \square \mathcal{O}_3 \square ... \square \mathcal{O}_{D+1}$ which is equivalent to a pair of disjoint $D$-dimensional hypercubic lattices, and so has a flag graph equivalant to a pair of disjoint $D$-dimensional colour code lattices. The gluing $o_{1,i}^l \Leftarrow o_{2,j}^l : \mathcal{G}_\square$ which glues together these two hypercubic lattices then also modifies the associated colour code lattices/flag graph in accordance with the discussion in the previous section. Specifically, it creates new $c_0$ or $c_D$ edges (depending on $l=1$ or $0$) between vertices of $\mathcal{G}_\mathcal{F}$ that correspond to flags lying on the seam assocaited to the gluing, and so joins the whole flag graph into a single connected component. An example is shown in \cref{fig:glued_ccs} b), where we can see that the flag graph post-gluing remains locally identical to a colour code lattice everywhere except in the neighbourhood of the seam. We refer to this operation as the \textit{join} of two colour code lattices.

\begin{definition}
    Given a pair of product graphs $\mathcal{G}_\square^1$ and $\mathcal{G}_\square^2$ whose flag graphs are colour code lattices, a \textbf{join} of these lattices is the transformation induced by the gluing $g_i^l \Leftarrow g_j^l: (\mathcal{G}_\square^1 \cup \mathcal{G}_\square^2)$ for $g_i^l \in \mathcal{G}_\square^1$ and $g_j^l \in \mathcal{G}_\square^2$.
\end{definition}

We can also define an opposite operation, which is the \textit{split} of two colour code lattices. 

\begin{definition}
    Given a pair of joined colour code lattices defined from a graph $\mathcal{G} = g_i^l \Leftarrow g_j^l: (\mathcal{G}_\square^1 \cup \mathcal{G}_\square^2)$, a \textbf{split} of these lattices is the transformation induced by the ungluing $g_i^l \overset{U}\Leftrightarrow g_j^l : \mathcal{G}$, where $\lvert U \rvert = 2$.
\end{definition}

Notice that, due to \cref{eq:glue_unglue}, it is not guaranteed that performing a join and then a split recovers the original pair of colour code lattices. However, the requirement that $\lvert U \rvert = 2$ ensures that we recover either a pair of colour code lattice or a single larger colour code lattice, as this ungluing necessarily divides the underlying glued cycle graphs back into a disjoint union of cycle graphs. We can also convince ourselves of this diagramatically by considering a gluing of two cycles

\begin{center}
    \begin{tikzpicture}[scale=0.45]
        \filldraw (0,0) circle (3pt);
        \filldraw (-1,-1) circle (3pt);
        \filldraw (1,-1) circle (3pt);
        \filldraw (0,-2) circle (3pt);
        \draw (0,0) -- (-1,-1);
        \draw (0,0) -- (1,-1);
        \draw (-1,-1) -- (0,-2);
        \draw (1,-1) -- (0,-2);

        \filldraw (0,-2.5) circle (3pt);
        \filldraw (-1,-3.5) circle (3pt);
        \filldraw (1,-3.5) circle (3pt);
        \filldraw (0,-4.5) circle (3pt);
        \draw (0,-2.5) -- (-1,-3.5);
        \draw (0,-2.5) -- (1,-3.5);
        \draw (-1,-3.5) -- (0,-4.5);
        \draw (1,-3.5) -- (0,-4.5);

        \draw[->] (2,-2.25) -- (3,-2.25);

        \filldraw (5,-0.25) circle (3pt);
        \filldraw (4,-1.25) circle (3pt);
        \filldraw (6,-1.25) circle (3pt);
        \filldraw (5,-2.25) circle (3pt);
        \filldraw (4,-3.25) circle (3pt);
        \filldraw (6,-3.25) circle (3pt);
        \filldraw (5,-4.25) circle (3pt);
        \draw (5,-0.25) -- (4,-1.25);
        \draw (5,-0.25) -- (6,-1.25);
        \draw (4,-1.25) -- (5,-2.25);
        \draw (6,-1.25) -- (5,-2.25);
        \draw (5,-2.25) -- (4,-3.25);
        \draw (5,-2.25) -- (6,-3.25);
        \draw (4,-3.25) -- (5,-4.25);
        \draw (6,-3.25) -- (5,-4.25);
    \end{tikzpicture}
\end{center}

\noindent and the set of possible ungluings satisfying $\lvert U \rvert = 2$

\begin{center}
    \begin{tikzpicture}[scale=0.45]
        \filldraw (0,0) circle (3pt);
        \filldraw (-1,-1) circle (3pt);
        \filldraw (1,-1) circle (3pt);
        \filldraw (0,-2) circle (3pt);
        \draw (0,0) -- (-1,-1) -- (0,-2) -- (1,-1) -- cycle;

        \filldraw (0,-2.5) circle (3pt);
        \filldraw (-1,-3.5) circle (3pt);
        \filldraw (1,-3.5) circle (3pt);
        \filldraw (0,-4.5) circle (3pt);
        \draw (0,-2.5) -- (-1,-3.5) -- (0,-4.5) -- (1,-3.5) --cycle;

        \filldraw (3,0) circle (3pt);
        \filldraw (2,-1) circle (3pt);
        \filldraw (4,-1) circle (3pt);
        \filldraw (3,-2) circle (3pt);
        
        \filldraw (3,-2.5) circle (3pt);
        \filldraw (2,-3.5) circle (3pt);
        \filldraw (4,-3.5) circle (3pt);
        \filldraw (3,-4.5) circle (3pt);
        \draw (3,0) -- (2,-1) -- (3,-2.5) -- (4,-1) -- cycle;
        \draw (3,-2) -- (2,-3.5) -- (3,-4.5) -- (4,-3.5) -- cycle;

        \filldraw (6,0) circle (3pt);
        \filldraw (5,-1) circle (3pt);
        \filldraw (7,-1) circle (3pt);
        \filldraw (6,-2) circle (3pt);

        \filldraw (6,-2.5) circle (3pt);
        \filldraw (5,-3.5) circle (3pt);
        \filldraw (7,-3.5) circle (3pt);
        \filldraw (6,-4.5) circle (3pt);
        \draw (6,0) -- (5,-1) -- (6,-2) -- (5,-3.5) -- (6,-4.5) -- 
            (7,-3.5) -- (6,-2.5) -- (7,-1) -- cycle;

        \filldraw (9,0) circle (3pt);
        \filldraw (8,-1) circle (3pt);
        \filldraw (10,-1) circle (3pt);
        \filldraw (9,-2) circle (3pt);

        \filldraw (9,-2.5) circle (3pt);
        \filldraw (8,-3.5) circle (3pt);
        \filldraw (10,-3.5) circle (3pt);
        \filldraw (9,-4.5) circle (3pt);
        \draw (9,0) -- (10,-1) -- (9,-2) -- (10,-3.5) -- (9,-4.5) -- 
            (8,-3.5) -- (9,-2.5) -- (8,-1) -- cycle;

        \filldraw (12,0) circle (3pt);
        \filldraw (11,-1) circle (3pt);
        \filldraw (13,-1) circle (3pt);
        \filldraw (12,-2) circle (3pt);

        \filldraw (12,-2.5) circle (3pt);
        \filldraw (11,-3.5) circle (3pt);
        \filldraw (13,-3.5) circle (3pt);
        \filldraw (12,-4.5) circle (3pt);
        \draw (12,0) -- (11,-1) -- (12,-2) -- (13,-3.5) -- (12,-4.5) --
            (11,-3.5) -- (12,-2.5) -- (13,-1) -- cycle;

        \filldraw (15,0) circle (3pt);
        \filldraw (14,-1) circle (3pt);
        \filldraw (16,-1) circle (3pt);
        \filldraw (15,-2) circle (3pt);

        \filldraw (15,-2.5) circle (3pt);
        \filldraw (14,-3.5) circle (3pt);
        \filldraw (16,-3.5) circle (3pt);
        \filldraw (15,-4.5) circle (3pt);
        \draw (15,0) -- (16,-1) -- (15,-2) -- (14,-3.5) -- (15,-4.5) -- 
            (16,-3.5) -- (15,-2.5) -- (14,-1) -- cycle;
    \end{tikzpicture}
\end{center}

Sets of colour code lattices that can be mapped into each other by performing a join and then a split will be referred to as \textbf{split-equivalent}. Examples are shown in \cref{fig:glued_ccs} c) and d).

Finally we consider more complex gluings. Consider (n-1) gluings of $n$ cycles at a common vertex
\begin{equation}
    \prod_{a=2}^n \left(o_{1,i_a}^l \Leftarrow o_{a,i_a}^l\right) : O_1 \cup \left(\bigcup_{a=2}^n O_a\right)
\end{equation}
where $\prod$ is used to mean composition of gluings. These gluings generalise the gluing of two cycles at single vertex that we have considered previously, and are significant as their only action on the corresponding flag graphs is to create new edges, and the vertices of the flag graphs are always preserved. We can think of them as joins of multiple colour code lattices, and these lattices (or a set that is split-equivalent) can be recovered by an appropriate sequence of ungluings. We call such seams \textit{splittable}.

\begin{definition}
    A seam in a $(D+1)$-coloured flag graph is \textbf{splittable} if there exists a sequence of ungluings such that all flags lying on the seam are part of exactly one edge of each colour in the flag graph produced by the ungluing.
\end{definition}

A set of joined colour codes can therefore be split into a set of disjoint colour codes \textit{iff} all seams are splittable (since a colour code lattice is defined by being $(D+1)$-valent and $(D+1)$-colourable). An example where this is not possible (and thus that contains unsplittable seams) is given by the gluing

\begin{center}
    \begin{tikzpicture}[scale=0.45]
        \filldraw (0,0) circle (3pt);
        \filldraw (1,0) circle (3pt);
        \filldraw (2,0) circle (3pt);
        \filldraw (0,1) circle (3pt);
        \filldraw (2,1) circle (3pt);
        \filldraw (0,2) circle (3pt);
        \filldraw (1,2) circle (3pt);
        \filldraw (2,2) circle (3pt);
        \draw (0,0) -- (2,0) -- (2,2) -- (0,2) -- cycle;

        \filldraw (3,0) circle (3pt);
        \filldraw (4,0) circle (3pt);
        \filldraw (5,0) circle (3pt);
        \filldraw (3,1) circle (3pt);
        \filldraw (5,1) circle (3pt);
        \filldraw (3,2) circle (3pt);
        \filldraw (4,2) circle (3pt);
        \filldraw (5,2) circle (3pt);
        \draw (3,0) -- (5,0) -- (5,2) -- (3,2) -- cycle;

        \draw[->] (6,1) -- (7,1);

        \filldraw (8,0) circle (3pt);
        \filldraw (9,0) circle (3pt);
        \filldraw (10,0) circle (3pt);
        \filldraw (11,0) circle (3pt);
        \filldraw (12,0) circle (3pt);
        \filldraw (8,1) circle (3pt);
        \filldraw (10,1) circle (3pt);
        \filldraw (12,1) circle (3pt);
        \filldraw (8,2) circle (3pt);
        \filldraw (9,2) circle (3pt);
        \filldraw (10,2) circle (3pt);
        \filldraw (11,2) circle (3pt);
        \filldraw (12,2) circle (3pt);
        \draw (8,0) -- (12,0) -- (12,2) -- (8,2) -- cycle;
        \draw (10,0) -- (10,2);
    \end{tikzpicture}
\end{center}

\noindent where two cycles are glued at three vertices/two edges. Because two of the vertices in the resulting graph have degree three there is no way to recover a pair of cycles using the ungluing operation that we have defined. We show in \cref{fig:glued_css_2} the effect of such a gluing on a pair of colour code lattices. Because the gluing is not reversible there is no splitting of this flag graph that produces a set of disjoint colour codes the flag graph contains unsplittable seams (corresponding to the two degree-3 vertices in the base graph). We generalise this idea with the following statement

\begin{lemma}
\label{lemma:splittable_seams}
    If all vertices in a graph $\mathcal{G}_\square = \mathcal{G}_1 \square \mathcal{G}_2 \square ... \square \mathcal{G}_D$ are of even degree then all seams are splittable. 
\end{lemma}

\begin{proof}
    A \textit{cycle decomposition} of a graph $\mathcal{G}$ is a partitioning of the edges of $\mathcal{G}$ into cycles. Such a decomposition is known to exist \textit{iff} every vertex in the graph is of even degree. Such a decomposition then implies the existence of an ungluing of $\mathcal{G}$ into disjoint cycles. Accordingly, if all vertices in $\mathcal{G}_\square$ are of even degree then there is a sequence of ungluing $U$ that unglues each factor $\mathcal{G}_i$ into a set of disjoint cycles. $U : \mathcal{G}_\square$ then has a flag graph equivalent to a collection of disjoint colour code lattices and thus all seams in $\mathcal{G}_\square$ are splittable.
\end{proof}

Finally, notice that given a flag graph $\mathcal{G}_\mathcal{F}$ obtained from a joining of colour code lattices it is not always possible to recover these original lattices by splitting even if all seams are splittable. For an example consider

\begin{center}
    \begin{tikzpicture}[scale=0.4]
        \filldraw (0,0) circle (3pt);
        \filldraw (1,0) circle (3pt);
        \filldraw (2,0) circle (3pt);
        \filldraw (0.4,0.75) circle (3pt);
        \filldraw (2.4,0.75) circle (3pt);
        \filldraw (0.8,1.5) circle (3pt);
        \filldraw (1.8,1.5) circle (3pt);
        \filldraw (2.8,1.5) circle (3pt);
        \draw (0,0) -- (2,0) -- (2.8,1.5) -- (0.8,1.5) -- cycle;

        \filldraw (3,0) circle (3pt);
        \filldraw (4,0) circle (3pt);
        \filldraw (5,0) circle (3pt);
        \filldraw (3.4,0.75) circle (3pt);
        \filldraw (5.4,0.75) circle (3pt);
        \filldraw (3.8,1.5) circle (3pt);
        \filldraw (4.8,1.5) circle (3pt);
        \filldraw (5.8,1.5) circle (3pt);
        \draw (3,0) -- (5,0) -- (5.8,1.5) -- (3.8,1.5) -- cycle;

        \filldraw (2.5,0.25) circle (3pt);
        \filldraw (2.9,1) circle (3pt);
        \filldraw (3.3,1.75) circle (3pt);
        \filldraw (2.5,1.25) circle (3pt);
        \filldraw (3.3,2.75) circle (3pt);
        \filldraw (2.5,2.25) circle (3pt);
        \filldraw (2.9,3) circle (3pt);
        \filldraw (3.3,3.75) circle (3pt);
        \draw (2.5,0.25) -- (3.3,1.75) -- (3.3,3.75) -- (2.5,2.25) -- cycle;

        \draw[->] (6,1) -- (7,1);
        
        \filldraw (7,0) circle (3pt);
        \filldraw (8,0) circle (3pt);
        \filldraw (9,0) circle (3pt);
        \filldraw (10,0) circle (3pt);
        \filldraw (11,0) circle (3pt);
        \filldraw (7.4,0.75) circle (3pt);
        \filldraw (9.4,0.75) circle (3pt);
        \filldraw (11.4,0.75) circle (3pt);
        \filldraw (7.8,1.5) circle (3pt);
        \filldraw (8.8,1.5) circle (3pt);
        \filldraw (9.8,1.5) circle (3pt);
        \filldraw (10.8,1.5) circle (3pt);
        \filldraw (11.8,1.5) circle (3pt);
        \draw (7,0) -- (11,0) -- (11.8,1.5) -- (7.8,1.5) -- cycle;
        \draw (9,0) -- (9.8,1.5);
        \filldraw (9,1) circle (3pt);
        \filldraw (9,2) circle (3pt);
        \filldraw (9.4,2.75) circle (3pt);
        \filldraw (9.8,3.5) circle (3pt);
        \filldraw (9.8,2.5) circle (3pt);
        \draw (9,0) -- (9,2) -- (9.8,3.5) -- (9.8,1.5);

        \draw[->] (12,1) -- (13,1);

        \filldraw (13,0) circle (3pt);
        \filldraw (14,0) circle (3pt);
        \filldraw (15,0) circle (3pt);
        \filldraw (13.4,0.75) circle (3pt);
        \filldraw (13.8,1.5) circle (3pt);
        \filldraw (14.8,1.5) circle (3pt);
        \filldraw (15.8,1.5) circle (3pt);
        \filldraw (15,1) circle (3pt);
        \filldraw (15,2) circle (3pt);
        \filldraw (15.4,2.75) circle (3pt);
        \filldraw (15.8,3.5) circle (3pt);
        \filldraw (15.8,2.5) circle (3pt);
        \draw (13,0) -- (15,0) -- (15,2) -- (15.8,3.5) -- (15.8,1.5) -- (13.8,1.5) -- cycle;

        \filldraw (16,0) circle (3pt);
        \filldraw (17,0) circle (3pt);
        \filldraw (18,0) circle (3pt);
        \filldraw (16.4,0.75) circle (3pt);
        \filldraw (18.4,0.75) circle (3pt);
        \filldraw (16.8,1.5) circle (3pt);
        \filldraw (17.8,1.5) circle (3pt);
        \filldraw (18.8,1.5) circle (3pt);
        \draw (16,0) -- (18,0) -- (18.8,1.5) -- (16.8,1.5) -- cycle;
    \end{tikzpicture}
\end{center}

\noindent where three cycles are glued together at a set of edges. There is no ungluing that can recover these three cycles, but we \textit{can} recover a pair of cycles (as shown). A flag graph obtained from three colour code lattices by a similar gluing is shown in \cref{fig:glued_css_2} b). Even though all the seams are splittable we cannot recover these three lattices via splitting, we can only obtain a pair of lattices. 

\subsection{Seams, stabilisers and logical operators}
\label{subsection:hgp_seam}

Finally we want to understand how to assign stabilisers to the maximal and rainbow subgraphs that exist at seams, and how this choice transforms the logical operators of the joined codes. 

 Consider a pair of $D$-dimensional colour codes joined at a single splittable seam, as in \cref{fig:glued_ccs}. In unjoined colour codes stabilisers can be assigned to either $x$/$z$-maximal or $x$/$z$-rainbow subgraphs as these two types of subgraph are equivalent, but at the seam where two colour code lattices have been joined this is no longer the case. What possible assignments of stabilisers to subgraphs exist at the seam, and how do these assignments affect the logical operator structure of the resulting code?
 
 Recall that the single-qubit logical $Z$ operators of a colour code on a $D$-dimensional torus (for $x=D$ and $z=2$) have a canonical basis in which each operator is supported on the vertices of a set of $c_i$ coloured edges running in a direction $d_j$ around a specific handle of the torus~\cite{kubica2018abcs}. We will write logical $Z$ operators of a colour code $C_m$ as $\overline{Z}_{(c_i,d_j)}^m$. Joining a pair of these codes at a seam then corresponds to joining the two torii at a hyperplane, and so logical $Z$ operators from each code can be described as being either \textit{parallel} or $\textit{normal}$ to this seam. The interactions of each type of logical operator with the seam are encapsulated by the following results, 

 \begin{lemma}
     \label{lemma:parallel_to_seam}
     Given two colour codes $C_1$ and $C_2$ joined at a splittable seam, logical $Z$ operators $\overline{Z}_{(c_i,d_j)}^1$ and $\overline{Z}_{(c_i,d_j)}^2$ with $d_j$ parallel to the seam are equivalent up to composition with $Z$ operators supported on $z$-rainbow subgraphs, but not with $Z$ operators supported on $z$-maximal subgraphs.
 \end{lemma}

\begin{proof}
    Because this seam is splittable $C_1$ and $C_2$ are split-equivalent to a single large colour code, in which $\overline{Z}_{(c_i,d_j)}^1$ and $\overline{Z}_{(c_i,d_j)}^2$ must be equivalent up to composition with $Z$ operators supported on $z$-rainbow subgraphs as these are $Z$ stabilisers of this code and $\overline{Z}_{(c_i,d_j)}^1$ and $\overline{Z}_{(c_i,d_j)}^2$ are logically equivalent in this code. These same rainbow subgraphs must then also exist in the joined code because splitting can only delete edges and not create them. 

    To see that the same is not true for maximal subgraphs notice that splitting the joined codes back into the original pair also splits all maximal subgraphs into disjoint collections of rainbow subgraphs (which are supports of colour code stabilisers). Assume that there is an operator $Z_\mathcal{M}$ supported on maximal subgraphs that transforms $\overline{Z}_{(c_i,d_j)}^1$ into $\overline{Z}_{(c_i,d_j)}^2$. Splitting these maximal subgraphs into rainbow subgraphs then partitions this support into the supports of stabilisers of the original colour codes, but the product of these stabilisers with $\overline{Z}_{(c_i,d_j)}^1$ cannot be $\overline{Z}_{(c_i,d_j)}^2$ and so $Z_\mathcal{M}$ cannot exist.
\end{proof}

\begin{lemma}
    \label{lemma:normal_to_seam}
    Given two colour codes $C_1$ and $C_2$ joined at a splittable seam, logical $Z$ operators $\overline{Z}_{(c_i,d_j)}^1$ and $\overline{Z}_{(c_i,d_j)}^2$ with $d_j$ normal to the seam anticommute with $X$ operators supported on $x$-rainbow subgraphs at the seam, but not with $X$ operators supported on $x$-maximal subgraphs.
\end{lemma}

\begin{proof}
    As before, we use the fact that $C_1$ and $C_2$ are split-equivalent to a single large colour code. As $\overline{Z}_{(c_i,d_j)}^1$ and $\overline{Z}_{(c_i,d_j)}^2$ are not logical operators of this code (but their product is) there must be $X$ stabilisers of this code (supported on $x$-rainbow subgraphs) that anticommute with these operators and then, once again, we know that these subgraphs must also exist in the joined code as splitting cannot create edges. 

    To see that these $Z$ logicals commute with all $X$ operators supported on $x$-maximal subgraphs note that, as in the previous proof, we can split any such subgraph into $x$-rainbow subgraphs of the original codes (which are supports of $X$ stabilisers of these codes). Any $Z$ logical of code $C_1$ or $C_2$ commutes with all $X$ stabilisers of codes $C_1$ and $C_2$ and so also commutes with products of these stabilisers, of which operators supported on $x$-maximal subgraphs must be a subset. 
\end{proof}


Each of the four classes of codes defined in \cref{tab:rainbow_classes} describes a different assignment of stabilisers to the subgraphs at the seam and in \cref{tab:classes_at_seam} we summarise how logical operators of the original codes are modified in each case (based on \cref{lemma:parallel_to_seam} and \cref{lemma:normal_to_seam}). The first three are fairly straightforward but it is interesting to discuss the mixed case in more detail, as the exact assignment of stabilisers for this case has not yet been explicitly defined. 

\begin{table}[]
    \centering
    \begin{tabular}{c|c|c}
        & $d_j$ parallel & $d_j$ normal \\
        \hline \\
        pin & $\overline{Z}_{(c_i,d_j)}^1 \neq \overline{Z}_{(c_i,d_j)}^2$ & $\overline{Z}_{(c_i,d_j)}^1$ and $\overline{Z}_{(c_i,d_j)}^2$ \\
        generic & $\overline{Z}_{(c_i,d_j)}^1 = \overline{Z}_{(c_i,d_j)}^2$ & $\overline{Z}_{(c_i,d_j)}^1$ and $\overline{Z}_{(c_i,d_j)}^2$ \\
        anti-generic & $\overline{Z}_{(c_i,d_j)}^1 \neq \overline{Z}_{(c_i,d_j)}^2$ & $\overline{Z}_{(c_i,d_j)}^1 \times \overline{Z}_{(c_i,d_j)}^2$ \\
        mixed & depends on $c_i$ & depends on $c_i$ \\
    \end{tabular}
    \caption{Modification of $Z$ logicals in a pair of joined colour codes for various stabiliser assignment strategies at the seam. Parallel logicals either remain distinct or become equivalent. Normal logicals either remain distinct or must be combined. In mixed codes we can have a mix of these behaviours depending on stabiliser and logical colours, as discussed in the main text.}
    \label{tab:classes_at_seam}
\end{table}

\begin{lemma}
    \label{lemma:mixed_assignment}
    The following is a valid assignment of stabilisers to subgraphs for a pair of colour code lattices joined at a splittable seam. For $X$ stabilisers
    \begin{itemize}
        \item $\{c_0,...,c_{D-1}\}$-rainbow subgraphs
        \item $\{c_1,...,c_D\}$-rainbow subgraphs
        \item $D$-maximal subgraphs for all other colourings 
    \end{itemize}
    and for $Z$ stabilisers 
    \begin{itemize}
        \item $\{c_0,c_D\}$-maximal subgraphs
        \item $2$-rainbow subgraphs for all other colourings
    \end{itemize}
\end{lemma}

\begin{proof}
Recall that all vertices in these joined lattices can be part of only one edge of each colour apart from $c_0$ and $c_D$, as these are the only colours of edge that can be created by the joining. This means that $D$-rainbow and $2$-rainbow subgraphs can only have odd intersection when their only shared colour is $c_0$ or $c_D$ (or both), as if they share a vertex $v$ and another colour of edge, $c_i$, then they must also share the unique $c_i$-coloured edge connected to $v$, resulting in an even intersection. We can then see that $2$-rainbow subgraphs with colouring $\{c_0,c_D\}$ can have odd intersection with all colours of $D$-rainbow subgraph, while $D$-rainbow subgraphs with colouring $\{c_0,...,c_{i-1},c_{i+1},...,c_D\}$ can have odd intersection with $2$-rainbow subgraphs coloured $\{c_0,c_i\}$ or $\{c_i,c_D\}$, as well as $\{c_0,c_D\}$. The choice of stabiliser assignment described above therefore results in a valid, commuting stabiliser group.
\end{proof}

\begin{lemma}
    \label{lemma:mixed_seam_logicals}
    The effect of the stabiliser assignment described in \cref{lemma:mixed_assignment} is that, for parallel $d_j$
    \begin{itemize}
        \item $\overline{Z}_{(c_i,d_j)}^1$ and $\overline{Z}_{(c_i,d_j)}^2$ are always logical operators.
        \item $\overline{Z}_{(c_i,d_j)}^1$ and $\overline{Z}_{(c_i,d_j)}^2$ are equivalent for $c_i \neq c_0,c_D$
        \item $\overline{Z}_{(c_i,d_j)}^1$ and $\overline{Z}_{(c_i,d_j)}^2$ are inequivalent for $c_i = c_0,c_D$
    \end{itemize}
    and for normal $d_j$
    \begin{itemize}
        \item $\overline{Z}_{(c_i,d_j)}^1$ and $\overline{Z}_{(c_i,d_j)}^2$ are inequivalent logical operators for $c_i \neq c_0, c_D$
        \item Only $\overline{Z}_{(c_i,d_j)}^1 \times \overline{Z}_{(c_i,d_j)}^2$ is a logical operator for $c_i = c_0, c_D$
    \end{itemize}
\end{lemma}

\begin{proof}
    For the case of parallel $d_j$, recall that in a colour code, in order to deform a $Z$ logical supported on $c_i$ edges across a 2D subregion of the code we must take the product of this logical with all $Z$ stabilisers supported on $\{c_i,c_j\}$-rainbow subgraphs within this region. Recalling the proof of \cref{lemma:parallel_to_seam} we can then see that if $\{c_0,c_D\}$-rainbow subgraphs of the joined lattice are not supports of $Z$ stabilisers, $c_0$ and $c_D$-coloured logicals parallel to the seam cannot be equivalent. 
    
    For the case of normal $d_j$ recall that $c_i$-coloured $Z$ string operators in a colour code anticommute with $X$ stabilisers only when these stabiliers are supported on subgraphs not containing edges of colour $c_i$. By \cref{lemma:normal_to_seam} we then have that only $\overline{Z}_{(c_i,d_j)}^m$ with $c_i = c_0$ or $c_D$ anticommute with $X$ stabilisers at the seam. $\overline{Z}_{c_i,d_j}^m$ for all other $c_i$ are therefore logical operators in the joined code, whereas for $c_i=c_0$ or $c_D$ only products $\overline{Z}_{(c_i,d_j)}^1 \times \overline{Z}_{(c_i,d_j)}^2$ are logical operators. 
\end{proof}

\begin{figure}
    \centering
    \includegraphics[width=.5\textwidth]{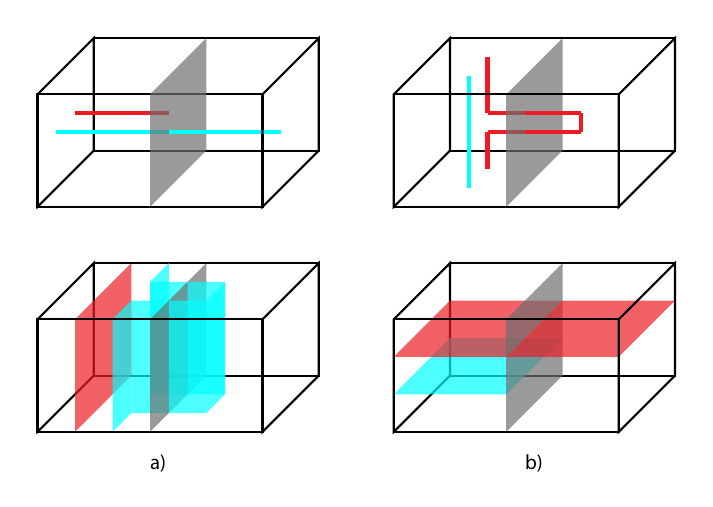}
    \caption{Interactions of logical operators (red and blue) with a splittable seam (grey) in a pair of joined 3D colour codes. 
    a) \textit{(above)} Stringlike $Z$ logicals normal to the seam. The red logical ($\overline{Z}_{(r,d_j)}^1$) can commute with all stabilisers while being supported only on one side of the seam while the blue logical ($\overline{Z}_{(b,d_j)}^1 \times \overline{Z}_{(b,d_j)}^2$) must be supported on both sides. 
    a) \textit{(below)} Membranelike $X$ logicals anticommuting with the $Z$ logicals above. The blue logical can be freely deformed through the seam ($\overline{X}_{(b,d_j)}^1 = \overline{X}_{(b,d_j)}^2$) while the red logical cannot ($\overline{X}_{(r,d_j)}^1 \neq \overline{X}_{(r,d_j)}^2$). 
    b) \textit{(above)} Stringlike $Z$ logicals parallel to the seam. The red logical can be freely deformed through the seam ($\overline{Z}_{(r,d_j)}^1 = \overline{Z}_{(r,d_j)}^2$) while the blue one cannot ($\overline{Z}_{(b,d_j)}^1 \neq \overline{Z}_{(b,d_j)}^2$). 
    b) \textit{(below)} Membranelike $X$ logicals anticommuting with the $Z$ logicals above. The blue logical ($\overline{X}_{(b,d_j)}^1$) can commute with all stabilisers while being supported only on one side of the seam while the red logical ($\overline{X}_{(r,d_j)}^1 \times \overline{X}_{(r,d_j)}^2$) must be supported on both sides.}
    \label{fig:logical_transformations_seam}
\end{figure}

Throughout this discussion we have only considered the $Z$ logical operators, but the transformations of $X$ logical operators can be straightforwardly inferred from the fact that the commutation relations between $X$ and $Z$ logical operator pairs must be preserved. For instance, consider logical $Z$ operators $\overline{Z}_{(c_i,d_j)}^1$ and $\overline{Z}_{(c_i,d_j)}^2$ and corresponding logical $X$ operators $\overline{X}_{(c_i,d_j)}^1$ and $\overline{X}_{(c_i,d_j)}^2$ such that
\begin{equation}
    [\overline{Z}_{(c_i,d_j)}^m,\overline{X}_{(c_i,d_j)}^n] = 
    \begin{cases}
        -1 \textrm{ for } m = n \\
        1 \textrm{ for } m \neq n
    \end{cases}
\end{equation}
where $X$ operators are supported on $(D-1)$-dimensional hypermembranes containing all colours of edge except $c_i$ and normal to direction $d_j$. If after a join we have $\overline{Z}_{(c_i,d_j)}^1 = \overline{Z}_{(c_i,d_j)}^2$ then $\overline{X}_{(c_i,d_j)}^1$ and $\overline{X}_{(c_i,d_j)}^2$ do not individually have consistent commutation relations with this operator and so only $\overline{X}_{(c_i,d_j)}^1 \times \overline{X}_{(c_i,d_j)}^2$ is a logical $X$ operator. Examples of all possible logical operator transformations at seam are shown in \cref{fig:logical_transformations_seam}.

All results in this section generalise straightforwardly to the case of more than two colour codes joined at a splittable seam. In this case \cref{lemma:parallel_to_seam} and \cref{lemma:normal_to_seam} hold for all pairs of colour codes meeting at this seam, and so e.g. for parallel $Z$ logicals in generic codes we have $\overline{Z}_{(c_i,d_j)}^1 = \overline{Z}_{(c_i,d_j)}^2 = \overline{Z}_{(c_i,d_j)}^3 =...$, and for normal $Z$ logicals in anti-generic codes $\overline{Z}_{(c_i,d_j)}^1, \overline{Z}_{(c_i,d_j)}^2, \overline{Z}_{(c_i,d_j)}^3, ...$ all anticommute with $X$ stabilisers so that only $\overline{Z}_{(c_i,d_j)}^1 \times \overline{Z}_{(c_i,d_j)}^2 \times \overline{Z}_{(c_i,d_j)}^3 \times ...$ is a logical operator.

We also want to consider cases such as \cref{fig:glued_css_2} b), which can be interpreted as a joining of either a pair or a triple of colour codes. While we know \textit{how} colour code logicals transform at the seams in this case there is an ambiguity about \textit{which} colour code logicals we should consider (logicals of the pair of codes or of the triple). If we label the two options for cycle graphs used in the product as 

\begin{center}
    \begin{tikzpicture}[scale=0.4]
        \filldraw (0,0) circle (3pt);
        \filldraw (1,0) circle (3pt);
        \filldraw (2,0) circle (3pt);
        \filldraw (0.4,0.75) circle (3pt);
        \filldraw (2.4,0.75) circle (3pt);
        \filldraw (0.8,1.5) circle (3pt);
        \filldraw (1.8,1.5) circle (3pt);
        \filldraw (2.8,1.5) circle (3pt);
        \draw (0,0) -- (2,0) -- (2.8,1.5) -- (0.8,1.5) -- cycle;
        \node at (1.4,0.7) {$\mathcal{O}_1$};

        \filldraw (3,0) circle (3pt);
        \filldraw (4,0) circle (3pt);
        \filldraw (5,0) circle (3pt);
        \filldraw (3.4,0.75) circle (3pt);
        \filldraw (5.4,0.75) circle (3pt);
        \filldraw (3.8,1.5) circle (3pt);
        \filldraw (4.8,1.5) circle (3pt);
        \filldraw (5.8,1.5) circle (3pt);
        \draw (3,0) -- (5,0) -- (5.8,1.5) -- (3.8,1.5) -- cycle;
        \node at (4.4,0.7) {$\mathcal{O}_3$};

        \filldraw (2.5,0.25) circle (3pt);
        \filldraw (2.9,1) circle (3pt);
        \filldraw (3.3,1.75) circle (3pt);
        \filldraw (2.5,1.25) circle (3pt);
        \filldraw (3.3,2.75) circle (3pt);
        \filldraw (2.5,2.25) circle (3pt);
        \filldraw (2.9,3) circle (3pt);
        \filldraw (3.3,3.75) circle (3pt);
        \draw (2.5,0.25) -- (3.3,1.75) -- (3.3,3.75) -- (2.5,2.25) -- cycle;
        \node at (2,3) {$\mathcal{O}_2$};

        \node at (7.4,1) {and};

        \filldraw (9,0) circle (3pt);
        \filldraw (10,0) circle (3pt);
        \filldraw (11,0) circle (3pt);
        \filldraw (9.4,0.75) circle (3pt);
        \filldraw (9.8,1.5) circle (3pt);
        \filldraw (10.8,1.5) circle (3pt);
        \filldraw (11.8,1.5) circle (3pt);
        \filldraw (11,1) circle (3pt);
        \filldraw (11,2) circle (3pt);
        \filldraw (11.4,2.75) circle (3pt);
        \filldraw (11.8,3.5) circle (3pt);
        \filldraw (11.8,2.5) circle (3pt);
        \draw (9,0) -- (11,0) -- (11,2) -- (11.8,3.5) -- (11.8,1.5) -- (9.8,1.5) -- cycle;
        \node at (9.4,2.5) {$\mathcal{O}_1 \times \mathcal{O}_2$};

        \filldraw (12,0) circle (3pt);
        \filldraw (13,0) circle (3pt);
        \filldraw (14,0) circle (3pt);
        \filldraw (12.4,0.75) circle (3pt);
        \filldraw (14.4,0.75) circle (3pt);
        \filldraw (12.8,1.5) circle (3pt);
        \filldraw (13.8,1.5) circle (3pt);
        \filldraw (14.8,1.5) circle (3pt);
        \draw (12,0) -- (14,0) -- (14.8,1.5) -- (12.8,1.5) -- cycle;
        \node at (13.4,0.7) {$\mathcal{O}_3$};
    \end{tikzpicture}
\end{center}

\noindent (where $\times$ is a composition of cycles by taking the symmetric difference of edges) then we can label the corresponding sets of colour code lattices as $\{C_1,C_2,C_3\}$ and $\{C_{1\times2}, C_3\}$. If we then consider a generic rainbow code defined on the joined lattices we can see that, in the former case, we have $\overline{Z}_{(c_i,d_j)}^1$, $\overline{Z}_{(c_i,d_j)}^2$ and $\overline{Z}_{(c_i,d_j)}^3$ all as independent logical operators for $d_j$ normal to both seams and for any $c_i$. On the other hand, in the latter case we have just $\overline{Z}_{(c_i,d_j)}^{1\times2}$ and $\overline{Z}_{(c_i,d_j)}^3$. This is therefore not a complete basis, and in fact $\overline{Z}_{(c_i,d_j)}^{1\times2} = \overline{Z}_{(c_i,d_j)}^1 \times \overline{Z}_{(c_i,d_j)}^2$ as can be checked using \cref{fig:glued_css_2} b). We therefore conclude that in order to identify a complete basis for colour code logical operators in such a code we must identify a complete cycle basis for the graphs used as input to the product and consider the colour code lattices obtained from these cycles. In contrast, the number of physical qubits in this code is equal to the number of physical qubits in $C_{1\times2} \cup C_3$ (compare \cref{fig:glued_css_2} b) to \cref{fig:glued_ccs} a)), which is fewer than the number of qubits we would have in $C_1 \cup C_2 \cup C_3$. These codes thus have the potential for improved encoding rates relative to disjoint collections of colour codes.

\subsection{Hypergraph product rainbow codes}

We now have all the tools we need to study rainbow codes obtained from the hypergraph product. Specifically, we will consider codes defined on the flag graph of a product graph
\begin{equation}
    \mathcal{G}_\square = \mathcal{G}_1 \square \mathcal{G}_2 \square ... \square \mathcal{G}_D := \scalebox{1.5}{$\square$}_{k=1}^D \mathcal{G}_k
\end{equation}
where $\mathcal{G}_k$ are all connected, bipartite and of even degree. Each $\mathcal{G}_k$ can be alternatively written as 
\begin{equation}
    \{ g_i \Leftarrow g_j ~ \forall ~ (g_i,g_j) \in P_k\} : \mathcal{G}_k'
\end{equation}
where $P_k$ is a set of vertex pairs and $\mathcal{G}_k'$ is a graph whose connected components are even-length cycle graphs in one-to-one correspondence with the elements of a fundamental cycle basis of $\mathcal{G}_k$. The size of such a basis, also called the ``circuit rank'' of $\mathcal{G}_k$, is given by
\begin{equation}
    \label{eq:cycle_basis_size}
    n_c^k = n_e^k - n_v^k + n_{cc}^k
\end{equation}
where $n_e^k$, $n_v^k$ and $n_{cc}^k$ are the numbers of edges, vertices and connected components in $\mathcal{G}_k$. In our case this reduces to
\begin{equation}
    n_c^k = n_e^k - n_v^k + 1
\end{equation}
because we only consider connected $\mathcal{G}_k$. We then use the commutativity of gluing and the Cartesian product to rewrite $\mathcal{G}_\square$ as 
\begin{equation}
    \begin{split}
    \mathcal{G}_\square 
    &= \scalebox{1.5}{$\square$}_{k=1}^D (\{g_i \Leftarrow g_j ~ \forall ~ (g_i,g_j) \in P_k\} : \mathcal{G}_k') \\
    &= \{g_i \Leftarrow g_j ~ \forall ~ (g_i,g_j) \in \scalebox{1.5}{$\cup$}_k^D P_k \} : (\scalebox{1.5}{$\square$}_{k=1}^D \mathcal{G}_k')
    \end{split}
\end{equation}
Each $\mathcal{G}_k'$ is a collection of disjoint cycles and so each connected component of $\scalebox{1.5}{$\square$}_k^D \mathcal{G}_k'$ is a $D$-dimensional hypercubic lattice on a $D$-dimensional torus and the flag graph of each of these components is a $D$-dimensional colour code lattice by \cref{lemma:cc_lattice}. If we label the elements of the cycle basis of each $\mathcal{G}_k'$ as $\mathcal{O}_a^k$ (for $1 \leq a \leq n_c^k$) then each colour code lattice is indexed by $(\mathcal{O}_{a_1}^1, \mathcal{O}_{a_2}^2,...,\mathcal{O}_{a_D}^D)$ and these can be interpreted as coordinates, giving a natural arrangement of these lattices in a $D$-dimensional grid e.g. \cref{fig:cc_grid}. Each of these lattices defines a ($x=D$ and $z=2$) colour code with logical $Z$ operators which we can index as $(\mathcal{O}_{a_1}^1, \mathcal{O}_{a_2}^2, ..., \mathcal{O}_{a_D}^D; c_i, d_j)$ where $c_i$ and $d_j$ are colour and direction as before, with $d_k$ being the direction in the grid associated with the cycles of $\mathcal{G}_k'$. For a given code and fixed $d_j$ only $D$ logicals of each colour are independent, so we have $D^2$ logical qubits per code.

\begin{figure}
    \centering
    \includegraphics[width=\linewidth]{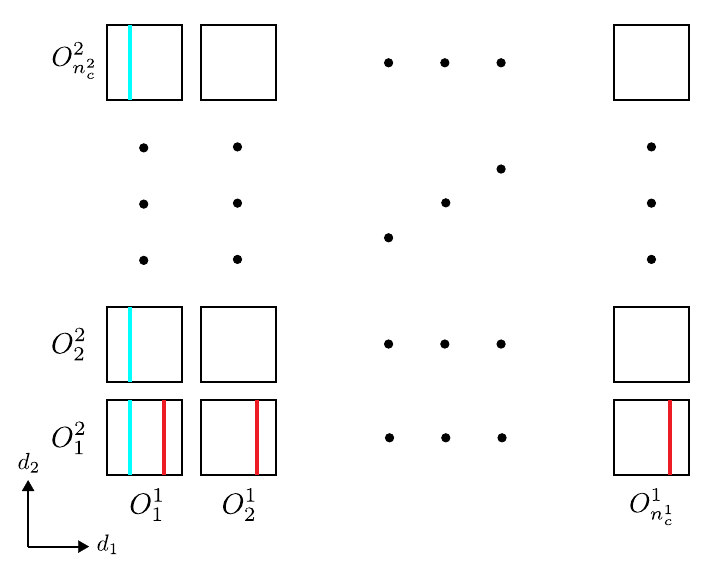}
    \caption{A 2D grid of colour codes obtained by defining a rainbow code on the flag graph of a Cartesian product of $n_c^1$ and $n_c^2$ disjoint cycles. The red lines show stringlike logicals of individual codes which can become associated by the joining of these codes while the blue line shows logicals which can be merged into a single logical by this joining.}
    \label{fig:cc_grid}
\end{figure}

The gluings $\{g_i \Leftarrow g_j\ ~ \forall ~ (g_i,g_j) \in P_k\}$ join all the disjoint cycles of $\mathcal{G}_k'$ into a single connected component and so join together all colour code lattices lying along each $k$-directional line of the grid (e.g. $\{g_i \Leftarrow g_j\ ~ \forall ~ (g_i,g_j) \in P_1\}$ join the lattices of each row of \cref{fig:cc_grid} while $\{g_i \Leftarrow g_j\ ~ \forall ~ (g_i,g_j) \in P_2\}$ join the lattices of each column). Because the $\mathcal{G}_k$ contain only even-degree vertices the resulting seams are all splittable seams by \cref{lemma:splittable_seams}, and we know how colour code logical operators transform at each of these seams by \cref{lemma:parallel_to_seam} and \cref{lemma:normal_to_seam}. This allows us to understand the resulting structure of these operators at the level of the grid. Explicitly, for each of our previously defined classes of codes, we have the following:

\textbf{Pin:} In pin codes all the logical operators of the original colour codes remain distinct after the joining, and in addition we have one logical operator for each independent $x$- or $z$-rainbow subgraph at each seam. These subgraphs are local features of the flag graph, explaining the linear rate and constant distance observed in numerical studies of hypergraph product pin codes. 

\textbf{Generic:} In these codes, for any choice of cycle $\mathcal{O}_{a_k}^k$, all logical $Z$ operators of the form $(\mathcal{O}_{a_1}^1,...,\mathcal{O}_{a_k}^k,...,\mathcal{O}_{a_D}^D; c_i, d_k)$ for fixed $c_i$ are equivalent, essentially giving one stringlike logical of each colour associated to each $\mathcal{O}_{a_k}^k$ (red lines in \cref{fig:cc_grid}). Recalling that in each colour code only $D$ colours of logical are independent for fixed $d_j$ this gives a number of independent logical operators
\begin{equation}
    n_L = \sum_i D n_c^i.
\end{equation}
The weights of these logicals are the same as the weights of the original colour code logicals, and so the distances of these codes are linear in the girth (minimum cycle length) of the input graphs $\mathcal{G}_k$.

\textbf{Anti-generic:} In these codes, for any choice of cycles $\mathcal{O}_{a_1}^1,...,\mathcal{O}_{a_{k-1}}^{k-1},\mathcal{O}_{a_{k+1}}^{k_1},...,\mathcal{O}_{a_D}^D$, all logical $Z$ operators of the form $(\mathcal{O}_{a_1}^1,...,\mathcal{O}_{a_k}^k,...,\mathcal{O}_{a_D}^D; c_i, d_k)$ are merged into a single logical operator, giving one logical of each colour for each $k$-directional line of the grid (blue lines in \cref{fig:cc_grid}). The number of independent logicals in this case is then
\begin{equation}
    n_L = \sum_i \prod_{j \neq i} Dn_c^j 
\end{equation}
and the weights of these logicals are the sums of the weights of all their constituent colour code logicals. The distances of these codes are then linear in both the girth and size of cycle basis of the input graphs. 

\textbf{Mixed:} As might be expected, in these codes we see a mix of the behaviour of the previous two cases. Specifically, we see the same behaviour as in the generic codes for all colours of logical except $c_0$ and $c_D$. As the $c_D$-coloured logicals are not independent we have a number of encoded qubits
\begin{equation}
    \label{eq:mixed_k}
    n_L = \sum_i \big((D-1)n_c^i + \prod_{j \neq i} n_c^j\big)
\end{equation}
and distance linear in the girth of the input graphs. However, we also note that for fixed input graphs the distance of a mixed code defined on the resulting flag graph will be twice the distance of the generic code defined on this same graph (as long as all $\mathcal{G}_k'$ contain more than one connected component). This is due to the fact that the $c_0$- and $c_D$-coloured logical $Z$ operators of a colour code have half the weight of the logicals of all other colours, but in the mixed code the $c_0$/$c_D$ logicals must be extended across multiple component colour codes and so the new lowest-weight logicals will be those of other colours.

All of the families of codes discussed above are LDPC. This follows from the fact that the weights of the stabiliser generators are upper-bounded by the size of the largest $D$-maximal subgraph, which, in the language of pin codes, is the size of the largest $1$-pinned set. It was shown in~\cite{vuillot2022quantum} that pin codes on simplicial complexes obtained from the hypergraph product of classical LDPC codes are themselves LDPC, and so this is also true of rainbow codes defined on the same complexes.

\section{Logical gates of HGP rainbow codes}
\label{section:logical_gates}
Now that we understand the structure and parameters of various classes of rainbow codes obtained from the hypergraph product we can examine the logical operations available in these codes. In particular, we are interested in logical non-Clifford gates that can be implemented by transversal application of $T/T^\dag$. In order for a code to possess such a gate we require the following properties, which are essentially a rephrasing of the triorthogonality conditions of~\cite{bravyi_magic-state_2012} and which have been equivalently presented in a number of other sources e.g.~\cite{bombin_transversal_2018,rengaswamy_optimality_2020,scruby_non-pauli_2022,webster_xp_2022}, but may be unfamiliar to some readers in this form.

\begin{lemma}
    \label{lemma:triorthogonal_conditions}
    The following are necessary and sufficient conditions for a quantum CSS code to possess a transversal non-Clifford gate implemented by application of physical $T$/$T^\dag$ to a bipartition of the qubits. 
    \begin{enumerate}
        \item The non-trivial intersection of any pair of $X$ stabilisers is the support of a $Z$ stabiliser.
        \item The non-trivial intersection of an $X$ stabiliser and an $X$ logical is the support of a $Z$ stabiliser.
        \item The intersection of at least one pair of $X$ logicals is the support of a $Z$ logical, and is either a $Z$ logical or stabiliser for all other non-trivially intersecting pairs. 
        \item For all $X$ stabilisers, the number of $T$ and $T^\dag$ applied to qubits in the support of this stabiliser are equal mod $8$.
        \item For all pairs of $X$ operators (i.e logicals and/or stabilisers) whose intersection is a $Z$ stabiliser, the number of $T$ and $T^\dag$ applied to qubits in this intersection are equal mod $4$.
    \end{enumerate}
\end{lemma}

\begin{proof}
Let  $\mathbf{a}$ be a length $n$ binary vector representing the bipartition of the qubits.
Let $W:=T(\mathbf{a})T^\dag(\mathbf{a}) = T(2\mathbf{a} -1)$ where $T(\mathbf{a}) := \prod_{0 \le i <n}T_i^{\mathbf{a}[i]}$.
Consider the action of $W$ on the CSS code with X checks $\{X(\mathbf{x}):\mathbf{x} \in S_X\}$, X-logicals $L_X$, Z-stabiliser generators $S_Z$, and Z-logicals $L_Z$. Let $M_Z$ be a generating set of logical identities modulo 4 of Section 6.2 of \cite{XP} so that $S(\mathbf{z})$ is a logical identity for all $\mathbf{z} \in \braket{M_Z}$ and let $\omega$ be a 16th root of unity such that $\omega^{16} = 1$.

Due to Proposition B.3 of \cite{CSSLO},  $W$ is a diagonal logical operator if and only if the group commutator $[ [ X(\mathbf{x_i}),W]]$ is a logical identity for all rows $ \mathbf{x_i} $ of $ {S_X}$. 
Calculating the group commutator using the identity in Table 4 of \cite{XP}:
\begin{align}[ [ X(\mathbf{x}_i), T(2\mathbf{a} - 1)] ] &= \omega^{2(2|\mathbf{a x}_i| - |\mathbf{x}_i|)}S(-2\mathbf{a x}_i + \mathbf{x}_i).
\end{align}
Hence, we require that both:
\begin{align}
2|\mathbf{ax}_i| &= |\mathbf{x}_i| \mod 8\text{; and}\\\mathbf{x}_i-2\mathbf{a x}_i &\in \braket{M_Z}.\end{align}
The first condition is equivalent to \textit{4} in the Lemma. 
Turning to the second condition, 
$\mathbf{x}-2\mathbf{a x}_i \in \braket{M_Z}$ if and only if the group commutator $[ [ X(\mathbf{x}_j),S(\mathbf{x}_i-2\mathbf{a x}_i)]]$ is a logical identity for all $ \mathbf{x}_j \in \braket{S_X,L_X}$.
Due to Proposition E.14 of \cite{XP}, it is sufficient to consider only $\mathbf{x}_j$ which are rows and sums of pairs of rows from the  matrix $\begin{pmatrix}S_X\\L_X\end{pmatrix}$.
Calculating the group commutator:
\begin{align}[ [ X(\mathbf{x}_j), S(\mathbf{x}_i-2\mathbf{a x}_i)] ] &= \omega^{4(|\mathbf{x}_i\mathbf{x}_j|-2|\mathbf{x}_i\mathbf{x}_j\mathbf{a}|)}Z(2\mathbf{x}_i\mathbf{x}_j\mathbf{a} -\mathbf{x}_i\mathbf{x}_j)\\&=\omega^{4(|\mathbf{x}_i\mathbf{x}_j|-2|\mathbf{x}_i\mathbf{x}_j\mathbf{a}|)}Z(\mathbf{x}_i\mathbf{x}_j).
\end{align}
Hence we require that both:
\begin{align}
2|\mathbf{x}_i\mathbf{x}_j\mathbf{a}| &= |\mathbf{x}_i\mathbf{x}_j|  \mod 4\text{; and}\\\mathbf{x}_i\mathbf{x}_j &\in \braket{S_Z}. \end{align}
The first condition is equivalent to \textit{5} and the second condition is equivalent to \textit{1} and \textit{2} in the Lemma.
By a similar argument, $W$ is a logical identity iff $\mathbf{x}_i\mathbf{x}_j \in \braket{S_Z}, \forall \mathbf{x}_i,\mathbf{x}_j \in \braket{L_X}$ which is equivalent to \textit{3}.
\end{proof}

It is well known that these requirements are satisfied by the 3D colour code on a 3-torus, with the specific logical action being $CCZ$ between triples of logical qubits whose colour and direction are all distinct, i.e.
\begin{align*}
    &\{(c_0,d_0), (c_1,d_1), (c_2,d_2)\}, \\
    &\{(c_0,d_1), (c_1,d_2), (c_2,d_0)\} \\
    &\mathrm{and}\\
    &\{(c_0,d_2), (c_1,d_0), (c_2,d_1)\}
\end{align*}
More generally, we can see that \textit{4} and \textit{5} are can be satisfied by any rainbow code obtained from a joining of colour code lattices where all seams are splittable seams. This is because any $3$-rainbow subgraph at such a seam can be viewed as the support of an $X$ stabiliser generator in one of the joined colour codes, and any $3$-maximal subgraph can be viewed as the support of a product of these generators. As \textit{4} and \textit{5} are satisfied for the $X$ stabilisers they must be satisfied in this case also. For \textit{1-3} we need to consider our four cases separately. We will focus specifically on the case of $D=3$, but as with the colour code the generalisation to higher dimensions is straightforward.

\textbf{Pin:} By \cref{lemma:mm_intersection} the intersections of $X$ stabilisers are supports of $Z$ stabilisers, satisfying \textit{1}. However, because $x$-rainbow subgraphs support $X$ logical operators, by \cref{lemma:mr_intersection} we have that the intersections of $X$ logicals and $X$ stabilisers can be $Z$ logicals and so \textit{2} is not generally satsified in these codes. 

\textbf{Generic:} By \cref{lemma:mm_intersection} the intersections of $X$ stabilisers are $2$-maximal subgraphs, but the supports of $Z$ stabilisers are $2$-rainbow subgraphs. Thus \textit{1} is only satisfied if we can always decompose these $2$-maximal subgraphs into $2$-rainbow subgraphs by deleting some of the edges. Fortunately this is true in our case because all seams are splittable seams and by deleting all edges created by the gluing we recover a disjoint collection of colour code lattices. This same property also allows us to split $3$-maximal subgraphs into disjoint collections of $3$-rainbow subgraphs which are $X$ stabilisers of the component colour codes. The logical operators of the rainbow code are also just logical operators of the component colour codes and, since we know \textit{2} and \textit{3} are satisfied in these colour codes, they must also be satisfied in the rainbow code. 

\textbf{Anti-generic:} \textit{1} is not satisfied in this case, as $X$ stabiliser intersections are $2$-rainbow subgraphs but $Z$ stabilisers are supported on $2$-maximal subgraphs. 

\textbf{Mixed:} \textit{1} is satisfied here by the same argument as for the generic codes, except that in some cases no decomposition into rainbow subgraphs is required. Notice that we do not run into the same issues as with the anti-generic codes here as $\{c_0,c_D\}$-rainbow subgraphs cannot be intersections of $X$ stabilisers. \textit{2} is also satisfied by the same argument as for the generic codes. For \textit{3}, the only possible issue would be if a pair of colour code $X$ logicals could intersect on the support of a $c_0$- or $c_D$-coloured $Z$ logical of a single colour code, as these are not logicals of the rainbow code (only products from multiple colour codes are). However, this would only be possible for a pair of $X$ logicals that share colour $c_0$ or $c_D$, and individual colour code $X$ logicals containing these colours are also not logicals of the rainbow code (only products from multiple colour codes are). We therefore conclude that \textit{3} is also satisfied in mixed codes.

\section{Examples}
\label{section:examples}
Mixed HGP rainbow codes have emerged from this discussion as the most interesting class, having better encoding rates than generic codes while still possessing transversal non-Clifford gates. We can now study some explicit examples of these codes in order to understand them more concretely.

\subsection{Figure-eight graphs}

For our first example we study the product $\mathcal{G}_\square = \mathcal{G}_8 \square \mathcal{G}_8 \square \mathcal{G}_8$ where $\mathcal{G}_8$ is the figure-of-eight graph

\begin{center}
    \begin{tikzpicture}[scale=0.8]
        \draw (0,0) -- (1,-1) -- (2,0) -- (3,-1) -- (4,0) --
            (3,1) -- (2,0) -- (1,1) -- cycle;
        \node[circle,fill=white,draw=black] at (0,0) {1};
        \node[circle,fill=white,draw=black] at (1,-1) {0};
        \node[circle,fill=white,draw=black] at (1,1) {0};
        \node[circle,fill=white,draw=black] at (2,0) {1};
        \node[circle,fill=white,draw=black] at (3,-1) {0};
        \node[circle,fill=white,draw=black] at (3,1) {0};
        \node[circle,fill=white,draw=black] at (4,0) {1};
    \end{tikzpicture}
\end{center}

which is equivalent to a pair of length-4 cycles glued at a single level-$1$ vertex, and so by ungluing these cycles and taking the Cartesian product we can obtain a flag graph equivalent to eight disjoint 3D colour code lattices on 3-torii, each of which have parameters $[\![384,9,4]\!]$. The flag grap of $\mathcal{G}_\square$ is then equivalent to the joining of these eight lattices at splittable seams as in \cref{fig:eight_codes}, which also shows some example logical operator structures. 

\begin{figure}
    \centering
    \includegraphics[width=\linewidth]{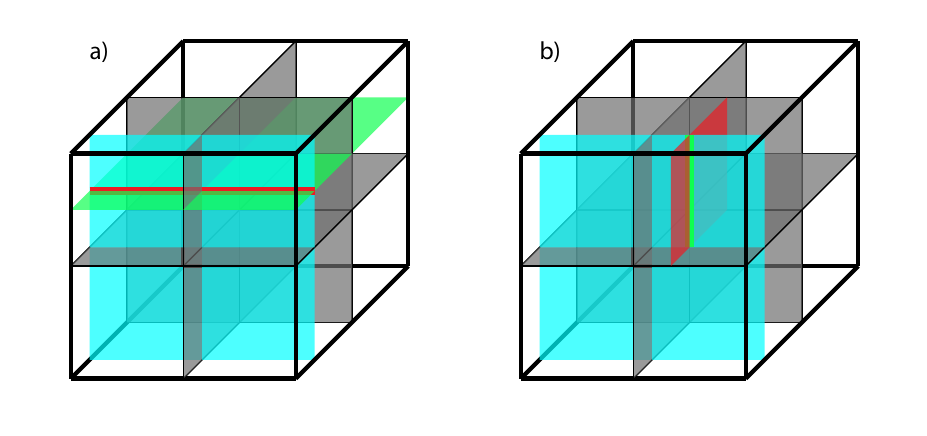}
    \caption{Structure of the mixed HGP rainbow code obtained from the product of three figure-of-eight graphs. Each octant of the cube is equivalent to a 3D colour code on a torus and each grey plane denotes a seam along which two lattices are joined. In a) we show logical $X$ operators of colours $c_1$ and $c_2$ (membranes) intersecting at a $c_0$ logical $Z$ operator (string). In b) we show logical $X$ operators of colours $c_0$ an $c_1$ intersecting at a logical $Z$ operator of colour $c_2$.}
    \label{fig:eight_codes}
\end{figure}

\begin{figure}
    \centering
    \includegraphics[width=.9\linewidth]{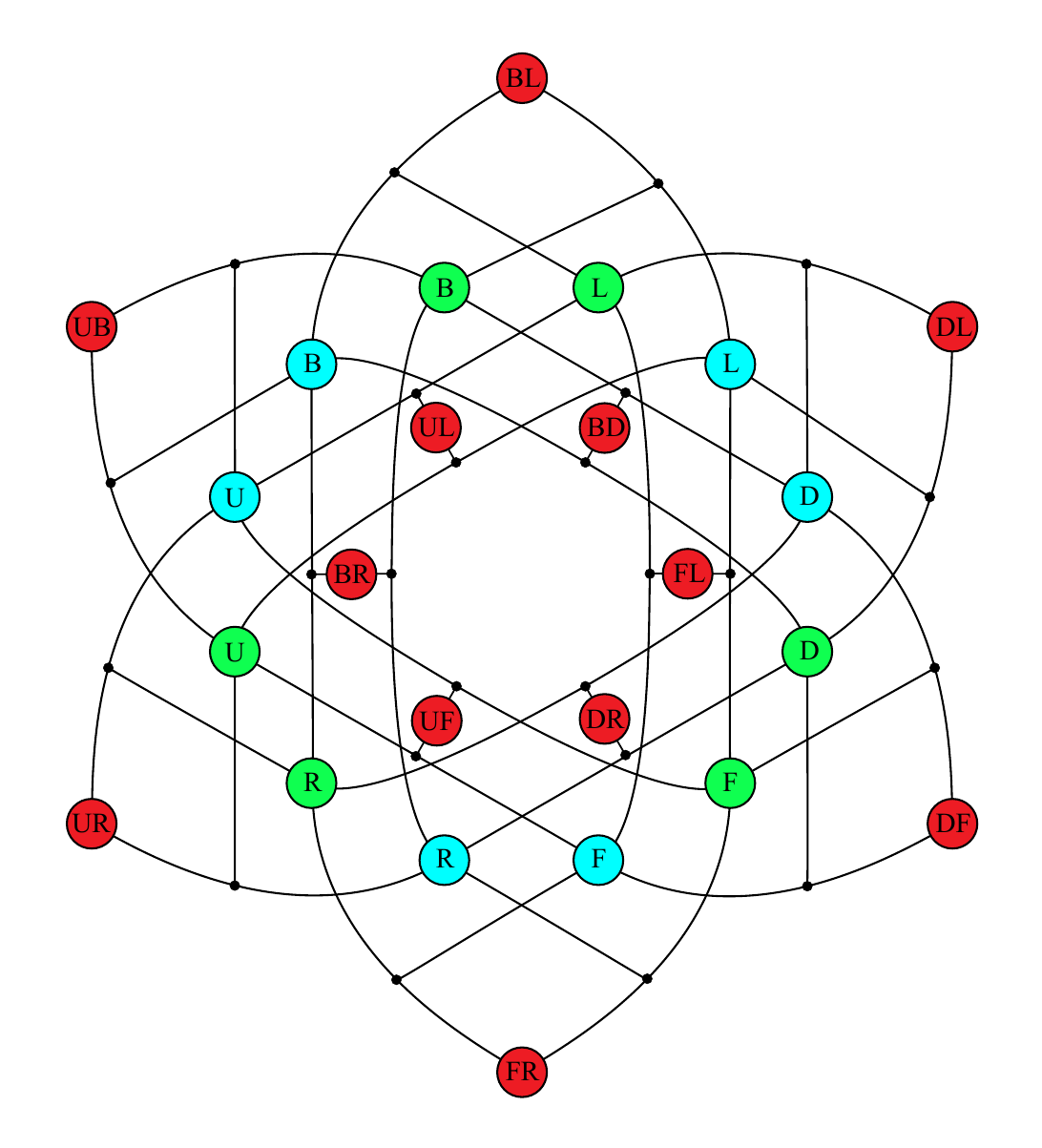}
    \caption{Interaction graph for the mixed HGP rainbow code defined from the product of three figure-of-eight graphs. Under the action of transversal $T/T^\dag$ logical $CCZ$ gates are applied between the triple of coloured vertices connected to each black vertex.}
    \label{fig:interaction_graph}
\end{figure}

We can count the number of physical qubits in this case as it is simply eight times the number of qubits in a single one of the 3D colour codes, so $n=3072$. The number of logical qubits is $k=24$ by \cref{eq:mixed_k} (as all $n_c^i = 2$) and the distance is twice the distance of any of the component colour codes so $d=8$, so we have a $[\![3072,24,8]\!]$ code. In contrast, the 3D colour code defined from a product of three length-8 cycle graphs is a $[\![3072,9,8]\!]$ code, although this is not the most efficient colour code for this $k$ and $d$ as there is also a $[\![768,9,8]\!]$ code that is not directly obtained from the hypergraph product (we show in the \cref{app:edge_contraction} how to obtain this code from the $[\![3072,9,8]\!]$ code). 

We can also try to understand the action of transversal $T/T^\dag$ on this code. With respect to the partitioning of the cube shown in \cref{fig:eight_codes} we can label each component colour code as being either front or back (F or B), left or right (L or R), and up or down (U or D). Logical operators of colour $c_1$ or $c_2$ in the rainbow code can be described by a single one of these labels while logical operators of colour $c_0$ can be described by a pair (e.g. the logicals shown in \cref{fig:eight_codes} are the $c_1$F, $c_2$U and $c_0$UF operators). We can use this to draw an interaction graph between the logical qubits describing the logical action of transversal $T/T^\dag$ (\cref{fig:interaction_graph}). We can see from this graph that two logical $CCZ$ gates are applied to each $c_0$ logical qubit while four $CCZ$ gates are applied to each $c_1$ and $c_2$ logical qubit. This example is available in the \href{https://github.com/m-webster/CSSLO/blob/main/rainbow_codes/A%20figure-8.ipynb}{linked Jupyter notebook}.

\subsection{Fully connected bipartite graphs}

The previous code does not offer any improvement in rate relative to its component colour codes (although it does have improved distance). This is because, as discussed in the previous section, the number of physical qubits in these codes depends on the number of cycles in a cycle decomposition of the input graphs whereas the number of logical qubits depends on the number of cycles in a fundamental cycle basis. For the figure-of-eight graph there are the same number of cycles in both cases, but this does not have to be the case. For our next example we consider the graph

\begin{center}
    \begin{tikzpicture}
        \draw (0,0) -- (0,1);
        \draw (0,0) -- (1,0);
        \draw (0,0) -- (2,1);
        \draw (0,0) -- (1,2);

        \draw (2,0) -- (0,1);
        \draw (2,0) -- (1,0);
        \draw (2,0) -- (2,1);
        \draw (2,0) -- (1,2);

        \draw (0,2) -- (0,1);
        \draw (0,2) -- (1,0);
        \draw (0,2) -- (2,1);
        \draw (0,2) -- (1,2);

        \draw (2,2) -- (0,1);
        \draw (2,2) -- (1,0);
        \draw (2,2) -- (2,1);
        \draw (2,2) -- (1,2);  
        \node[circle,fill=white,draw=black] at (0,0) {0};
        \node[circle,fill=white,draw=black] at (2,0) {0};
        \node[circle,fill=white,draw=black] at (0,2) {0};
        \node[circle,fill=white,draw=black] at (2,2) {0};
        \node[circle,fill=white,draw=black] at (1,0) {1};
        \node[circle,fill=white,draw=black] at (0,1) {1};
        \node[circle,fill=white,draw=black] at (1,2) {1};
        \node[circle,fill=white,draw=black] at (2,1) {1};    
    \end{tikzpicture}
\end{center}

\noindent where all type-0 nodes are connected to all type-1 nodes. This has a cycle decomposition

\begin{center}
    \begin{tikzpicture}
        \draw (0,0) -- (2,0) -- (2,2) -- (0,2) -- cycle;
        \node[circle,fill=white,draw=black] at (0,0) {0};
        \node[circle,fill=white,draw=black] at (2,0) {0};
        \node[circle,fill=white,draw=black] at (0,2) {0};
        \node[circle,fill=white,draw=black] at (2,2) {0};
        \node[circle,fill=white,draw=black] at (1,0) {1};
        \node[circle,fill=white,draw=black] at (0,1) {1};
        \node[circle,fill=white,draw=black] at (1,2) {1};
        \node[circle,fill=white,draw=black] at (2,1) {1};

        \draw (4,0) -- (5,2) -- (6,0) -- (4,1) -- 
            (6,2) -- (5,0) -- (4,2) -- (6,1) -- cycle;
        \node[circle,fill=white,draw=black] at (4,0) {0};
        \node[circle,fill=white,draw=black] at (6,0) {0};
        \node[circle,fill=white,draw=black] at (4,2) {0};
        \node[circle,fill=white,draw=black] at (6,2) {0};
        \node[circle,fill=white,draw=black] at (5,0) {1};
        \node[circle,fill=white,draw=black] at (4,1) {1};
        \node[circle,fill=white,draw=black] at (5,2) {1};
        \node[circle,fill=white,draw=black] at (6,1) {1};
    \end{tikzpicture}
\end{center}

\noindent whereas the size of a fundamental cycle basis (by \cref{eq:cycle_basis_size}) is $n_c = 16 - 8 + 1 = 9$. By \cref{eq:mixed_k} we then have 297 logical qubits in a product of three of these codes. To count the number of physical qubits we can start by noticing that, for a product of $D$ $d$-regular graphs, the number of flags associated to each level-$0$ node in the product graph is 
\begin{equation}
    Dd \times (D-1)d \times (D-2)d \times ... = D!d^D
\end{equation}
This is because the first (level-$0$) vertex in each flag is $(g_{1,i}^0,g_{2,i}^0,...,g_{D,i}^0)^0$ and then, to find a connected level-$1$ vertex, we need to change one of the $g_{j,i}^0$ to $g_{j,i'}^1$. There are $D$ choices of $j$ and then $d$ choices of vertex for each $j$, each with form $((g_{1,i}^0,...g_{j,i}^1,...,g_{D,i}^0)^1$. For the level-$2$ vertex we have $(D-1)$ choices of $j$ and $d$ choices of vertex for each $j$ and so on, eventually giving $D!d^D$. The total number of level-$0$ nodes in the product graph is 
\begin{equation}
    n_0 = \prod_i n_0^i
\end{equation}
where $n_0^i$ is the number of level-$0$ nodes in input graph $i$. The total number flags is then
\begin{equation}
    \label{eq:n_flags}
    n = n_0D!d^D
\end{equation}
and in this case we have $D=3,d=n_0^i=4$ so $n=4^3 \times 3! \times 4^3 = 24576$. The distance of this code, as with all mixed HGP codes, is twice the minimum girth of any input graph, which is $2\times4=8$, so we have a code with parameters $[\![24576,297,8]\!]$. To achieve the same $k$ and $d$ with 3D colour codes we would need $33$ copies of the $[\![768,9,8]\!]$ code, which would require $25344$ physical qubits. This example is available in the \href{https://github.com/m-webster/CSSLO/blob/main/rainbow_codes/B%20fully-connected.ipynb}{linked Jupyter notebook}.

\subsection{Expander graphs}

Next we consider an asymptotic construction based on $d$-regular bipartite expander graphs for even $d$, such as those presented in~\cite{lubotzky_ramanujan_1988}. Hypergraph product codes defined using such graphs can have fairly good parameters and so we might expect the corresponding rainbow codes to also scale well. Unfortunately, the parameters of these codes actually turn out to be quite bad (worse than those of a ``code'' corresponding to $\Theta(n)$ copies of the Euclidean 3D colour code). Specifically, as $n_0^i$ is linear in the size $n$ of the input graphs a code obtained from the product of three of these graphs will have $\Theta(n^3)$ physical qubits. The number of independent cycles in these graphs is also linear in $n$ as the number of edges in a $d$-regular graph is $nd/2$, and so by \cref{eq:mixed_k} the number of logical qubits is $\Theta(n^2)$. Finally, these graphs have girth $O(\log(n))$ and so the distance has the same scaling and we have a code family with parameters $[\![n,\Theta(n^{2/3}),\Theta(\log(n)]\!]$ after rescaling $n^3 \rightarrow n$. 

\subsection{Expander graph and a hyperbolic cellulation}

Finally we present a family of finite rate and non-constant distance codes combining HGP rainbow codes with the quasi-hyperbolic 3D colour codes of~\cite{zhu2023non}, which were obtained from manifolds corresponding to the product of a 2D hyperbolic manifold and a circle. This is compatible with our construction as cellulations of the hyperbolic manifold can be represented as graphs with three levels of vertex, but because these graphs are not themselves products of bipartite graphs the parameters of the resulting code will not be limited by \cref{eq:mixed_k}. The hyperbolic manifold has area $A$, genus $g = \Theta(A)$ and systole $\Theta(\log(A))$, and so a topological code defined on this manifold has parameters $[\![n,\Theta(n),\Theta(\log(n))]\!]$. The circle has length $\Theta(\log(A))$ and so the code defined on the product manifold has parameters $[\![n,\Theta(n/\log(n)),\Theta(\log(n))]\!]$ as the product with the circle increases the number of physical qubits by a factor of $\Theta(\log(n))$ while only increasing the number of logical qubits by a constant factor. We can visualise this manifold as in \cref{fig:hyperbolic} a), where the $x$ and $y$ directions have a hyperbolic metric and the $z$ direction has a Euclidean metric. A colour code defined on this manifold (using the same subdivision of a chain complex into flags that we have described above) supports a transversal non-Clifford gate implemented by $T/T^\dag$ and whose logical action is $CCZ$ between triples of qubits with logical $Z$ operators $\{Z_{c_i},Z_{c_j},Z_{c_k}\}$ where $Z_{c_i}$ and $Z_{c_j}$ surround a common handle in the 2D manifold while $Z_{c_k}$ is the unique $z$-direction logical of colour $c_k \neq c_i \neq c_j$. Examples of $Z$ and $X$ operators for such a triple are also shown in \cref{fig:hyperbolic} a).  Readers who desire a more thorough explanation of these properties can consult section V A of~\cite{zhu2023non}.

\begin{figure}
    \centering
    \includegraphics[width=\linewidth]{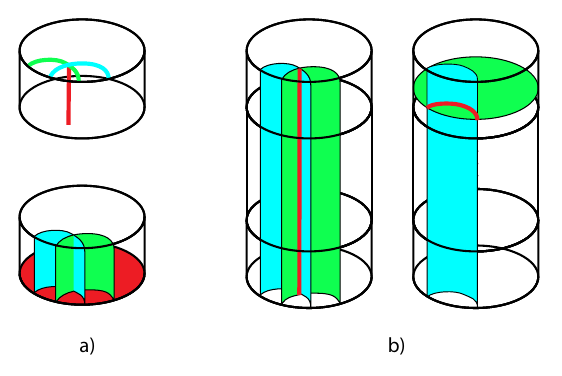}
    \caption{a) Examples of logical $Z$ (above) and $X$ (below) operators in a quasi-hyperbolic colour code as defined in~\cite{zhu2023non}. The circular cross section in the $(x,y)$ plane is a closed 2D hyperbolic manifold and contains an extensive number of handles, while the $z$ direction is Euclidean with top and bottom faces of the cylinder associated. The triple of $X$ logicals shown will be acted on by logical $CCZ$ when transversal $T/T^\dag$ is applied to the code. b) Examples of logical operators in a product of a 2D hyperbolic manifold and an expander graph. We choose a convention red = $c_0$, green = $c_1$, blue = $c_2$. (Left) Nontrivial $c_1$ and $c_2$ membranes perpendicular to seams meet at a nontrivial $c_0$ string. (Right) a nontrivial perpendicular $c_1$ and parallel $c_2$ membrane meet at a nontrivial $c_0$ string. All such $c_0$ strings originating from different quasi-hyperbolic codes correspond to the supports of independent logical operators.}
    \label{fig:hyperbolic}
\end{figure}

To achieve linear encoding rate we can instead consider the product of a $d$-regular (for even $d$) bipartite expander graph of size $s$ and a 2D hyperbolic manifold with area $A$. A rainbow code defined on the output of this product would have a number of physical qubits $\Theta(As)$. Using the fact that the expander graph contains $\Theta(s)$ independent cycles we can view this code as a joining of $\Theta(s)$ copies of the quasi-hyperbolic code described above, and prior to being joined these copies encode $\Theta(As)$ logical qubits. This product then results in a code with linear rate as long as only a constant fraction of logical qubits are lost in the joining. We can see that this is true when using the mixed stabiliser assignment as distinct $c_0$-coloured logical $Z$ operators parallel to the seams (i.e. in the plane of the hyperbolic manifold) are not associated by the joining, and there are $\Theta(A)$ such operators in each  quasi-hyperbolic code. The distance of this code will be the minimum of $\Theta(\log(A))$ and $\Theta(\log(s))$ so we can obtain the best relative distance by setting $A = s$, which gives a code with parameters $[\![n,\Theta(n),\Theta(\log(n))]\!]$. \Cref{fig:hyperbolic} b) shows examples of triples of logical operators of this code with intersections that allow for transformation by logical $CCZ$.

\section{Discussion}
\label{section:discussion}
In this work, we have introduced a general construction for defining quantum codes on any $D$-dimensional simplicial complex with $(D+1)$-colourable vertices. In cases where this complex describes a cellulation of a manifold, we recover the standard topological colour code, and in more general cases we have seen that the resulting codes can often be interpreted as copies of the colour code joined together at domain walls. This both makes understanding the properties of these more general codes straightforward and provides a promising method of constructing new code families. 

Although we have focused mostly on simplicial complexes obtained from the hypergraph product, there exist many other methods of generating suitable complexes, e.g. coset complexes \cite{KAUFMAN2023103696} or more sophisticated product constructions \cite{breuckmann_balanced_2021}. The potential of these other kinds of complex is exemplified by our example of a family of constant rate and growing distance codes with transversal non-Clifford gates, which were derived from a complex obtained by combining a graph product with a cellulation of a hyperbolic manifold. 
This construction gives the first family of quantum LDPC codes defined on qubits with linear rate, growing distance and transversal non-Clifford gates, which are necessary conditions to achieve $\gamma \rightarrow 0$. 
We also note that we have studied only a small subset of the many possible stabiliser assignments permitted by our construction, and that there may be other choices that result in codes with comparable or superior properties. The identification of more sophisticated ways to choose or understand such assignments is left as an open problem. 

Finally, we have not yet considered the question of how to decode rainbow codes. The decoding of colour codes often relies on their mapping to toric codes, and while we have shown the existence of a similar map for specific classes of rainbow codes (generic and coming from a length-2 chain complex), generalizing this map to higher-dimensional and mixed rainbow codes could be valuable for the development of a decoder.

\section*{Acknowledgments}

T. R. S. acknowledges support from the JST Moonshot R\&D Grant [grant number JPMJMS2061]. AP is supported by the Engineering and Physical Sciences Research Council (EP/S021582/1). MW is supported by the Engineering and Physical Sciences Research Council on Robust and Reliable Quantum Computing (RoaRQ), Investigation 011 [grant reference EP/W032635/1] and by the Engineering and Physical Sciences Research Council [grant number EP/S005021/1]. The authors acknowledge valuable discussions with Michael Vasmer, Armanda Quintavalle, Nikolas Breuckmann, Christophe Vuillot, Tim Hosgood, Guanyu Zhu and Louis Golowich.

\bibliographystyle{unsrt}
\bibliography{ref}

@article{bombin2006topological,
  title={Topological quantum distillation},
  author={Bombin, Hector and Martin-Delgado, Miguel Angel},
  journal={Physical review letters},
  volume={97},
  number={18},
  pages={180501},
  year={2006},
  publisher={APS}
}

@article{bombin2007topological,
  title={Topological computation without braiding},
  author={Bombin, Hector and Martin-Delgado, Miguel-Angel},
  journal={Physical review letters},
  volume={98},
  number={16},
  pages={160502},
  year={2007},
  publisher={APS}
}

@article{bombin2007exact,
  title={Exact topological quantum order in D= 3 and beyond: Branyons and brane-net condensates},
  author={Bombin, H and Martin-Delgado, MA},
  journal={Physical Review B},
  volume={75},
  number={7},
  pages={075103},
  year={2007},
  publisher={APS}
}

@article{yoshida2010framework,
  title={Framework for classifying logical operators in stabilizer codes},
  author={Yoshida, Beni and Chuang, Isaac L},
  journal={Physical Review A},
  volume={81},
  number={5},
  pages={052302},
  year={2010},
  publisher={APS}
}

@article{bombin2015gauge,
  title={Gauge color codes: optimal transversal gates and gauge fixing in topological stabilizer codes},
  author={Bomb{\'\i}n, H{\'e}ctor},
  journal={New Journal of Physics},
  volume={17},
  number={8},
  pages={083002},
  year={2015},
  publisher={IOP Publishing}
}

@article{kubica2015universal,
  title={Universal transversal gates with color codes: A simplified approach},
  author={Kubica, Aleksander and Beverland, Michael E},
  journal={Physical Review A},
  volume={91},
  number={3},
  pages={032330},
  year={2015},
  publisher={APS}
}

@article{kubica2015unfolding,
  title={Unfolding the color code},
  author={Kubica, Aleksander and Yoshida, Beni and Pastawski, Fernando},
  journal={New Journal of Physics},
  volume={17},
  number={8},
  pages={083026},
  year={2015},
  publisher={IOP Publishing}
}

@phdthesis{kubica2018abcs,
  title={The ABCs of the color code: A study of topological quantum codes as toy models for fault-tolerant quantum computation and quantum phases of matter},
  author={Kubica, Aleksander Marek},
  year={2018},
  school={California Institute of Technology}
}

@article{vuillot2022quantum,
  title={Quantum pin codes},
  author={Vuillot, Christophe and Breuckmann, Nikolas P},
  journal={IEEE Transactions on Information Theory},
  volume={68},
  number={9},
  pages={5955--5974},
  year={2022},
  publisher={IEEE}
}

@article{zhu2023non,
  title={Non-Clifford and parallelizable fault-tolerant logical gates on constant and almost-constant rate homological quantum LDPC codes via higher symmetries},
  author={Zhu, Guanyu and Sikander, Shehryar and Portnoy, Elia and Cross, Andrew W and Brown, Benjamin J},
  journal={arXiv preprint arXiv:2310.16982},
  year={2023}
}

@article{tillich_quantum_2014,
	title = {Quantum {LDPC} {Codes} {With} {Positive} {Rate} and {Minimum} {Distance} {Proportional} to the {Square} {Root} of the {Blocklength}},
	volume = {60},
	issn = {1557-9654},
	url = {https://ieeexplore.ieee.org/document/6671468},
	doi = {10.1109/TIT.2013.2292061},
	number = {2},
	urldate = {2024-07-08},
	journal = {IEEE Transactions on Information Theory},
	author = {Tillich, Jean-Pierre and Zémor, Gilles},
	month = feb,
	year = {2014},
	note = {Conference Name: IEEE Transactions on Information Theory},
	pages = {1193--1202},
}

@article{bravyi_magic-state_2012,
	title = {Magic-state distillation with low overhead},
	volume = {86},
	url = {https://link.aps.org/doi/10.1103/PhysRevA.86.052329},
	doi = {10.1103/PhysRevA.86.052329},
	number = {5},
	urldate = {2024-02-21},
	journal = {Physical Review A},
	author = {Bravyi, Sergey and Haah, Jeongwan},
	month = nov,
	year = {2012},
	note = {Publisher: American Physical Society},
	pages = {052329},
}

@misc{bombin_transversal_2018,
	title = {Transversal gates and error propagation in {3D} topological codes},
	url = {http://arxiv.org/abs/1810.09575},
	doi = {10.48550/arXiv.1810.09575},
	urldate = {2024-07-18},
	publisher = {arXiv},
	author = {Bombin, Hector},
	month = oct,
	year = {2018},
	note = {arXiv:1810.09575 [quant-ph]},
}

@article{scruby_non-pauli_2022,
	title = {Non-{Pauli} errors in the three-dimensional surface code},
	volume = {4},
	issn = {2643-1564},
	url = {https://link.aps.org/doi/10.1103/PhysRevResearch.4.043052},
	doi = {10.1103/PhysRevResearch.4.043052},
	language = {en},
	number = {4},
	urldate = {2024-07-18},
	journal = {Physical Review Research},
	author = {Scruby, Thomas R. and Vasmer, Michael and Browne, Dan E.},
	month = oct,
	year = {2022},
	pages = {043052},
}

@article{webster_xp_2022,
	title = {The {XP} {Stabiliser} {Formalism}: a {Generalisation} of the {Pauli} {Stabiliser} {Formalism} with {Arbitrary} {Phases}},
	volume = {6},
	shorttitle = {The {XP} {Stabiliser} {Formalism}},
	url = {https://quantum-journal.org/papers/q-2022-09-22-815/},
	doi = {10.22331/q-2022-09-22-815},
	language = {en-GB},
	urldate = {2024-07-18},
	journal = {Quantum},
	author = {Webster, Mark A. and Brown, Benjamin J. and Bartlett, Stephen D.},
	month = sep,
	year = {2022},
	note = {Publisher: Verein zur Förderung des Open Access Publizierens in den Quantenwissenschaften},
	pages = {815},
}

@article{vasmer_three-dimensional_2019,
	title = {Three-dimensional surface codes: {Transversal} gates and fault-tolerant architectures},
	volume = {100},
	issn = {2469-9926, 2469-9934},
	shorttitle = {Three-dimensional surface codes},
	url = {http://arxiv.org/abs/1801.04255},
	doi = {10.1103/PhysRevA.100.012312},
	number = {1},
	urldate = {2022-01-25},
	journal = {Physical Review A},
	author = {Vasmer, Michael and Browne, Dan E.},
	month = jul,
	year = {2019},
	note = {arXiv: 1801.04255},
	pages = {012312},
}

@article{XP,
  doi = {10.22331/q-2022-09-22-815},
  url = {https://doi.org/10.22331/q-2022-09-22-815},
  title = {The {XP} {S}tabiliser {F}ormalism: a {G}eneralisation of the {P}auli {S}tabiliser {F}ormalism with {A}rbitrary {P}hases},
  author = {Webster, Mark A. and Brown, Benjamin J. and Bartlett, Stephen D.},
  journal = {{Quantum}},
  issn = {2521-327X},
  publisher = {{Verein zur F{\"{o}}rderung des Open Access Publizierens in den Quantenwissenschaften}},
  volume = {6},
  pages = {815},
  month = sep,
  year = {2022}
}

@article{CSSLO,
doi = {10.1088/1367-2630/acfc5f},
url = {https://dx.doi.org/10.1088/1367-2630/acfc5f},
year = {2023},
month = {oct},
publisher = {IOP Publishing},
volume = {25},
number = {10},
pages = {103018},
author = {Mark A Webster and Armanda O Quintavalle and Stephen D Bartlett},
title = {Transversal diagonal logical operators for stabiliser codes},
journal = {New Journal of Physics}
}

@article{eastin_restrictions_2009,
	title = {Restrictions on {Transversal} {Encoded} {Quantum} {Gate} {Sets}},
	volume = {102},
	url = {https://link.aps.org/doi/10.1103/PhysRevLett.102.110502},
	doi = {10.1103/PhysRevLett.102.110502},
	number = {11},
	urldate = {2021-06-16},
	journal = {Physical Review Letters},
	author = {Eastin, Bryan and Knill, Emanuel},
	month = mar,
	year = {2009},
	note = {Publisher: American Physical Society},
	pages = {110502},
}

@article{kitaev_fault-tolerant_2003,
	title = {Fault-tolerant quantum computation by anyons},
	volume = {303},
	issn = {0003-4916},
	url = {https://www.sciencedirect.com/science/article/pii/S0003491602000180},
	doi = {10.1016/S0003-4916(02)00018-0},
	language = {en},
	number = {1},
	urldate = {2021-06-16},
	journal = {Annals of Physics},
	author = {Kitaev, A. Yu.},
	month = jan,
	year = {2003},
	pages = {2--30},
}

@article{moussa_transversal_2016,
	title = {Transversal {Clifford} gates on folded surface codes},
	volume = {94},
	issn = {2469-9926, 2469-9934},
	url = {http://arxiv.org/abs/1603.02286},
	doi = {10.1103/PhysRevA.94.042316},
	number = {4},
	urldate = {2024-08-22},
	journal = {Physical Review A},
	author = {Moussa, Jonathan E.},
	month = oct,
	year = {2016},
	pages = {042316},
}

@misc{panteleev_asymptotically_2022,
	title = {Asymptotically {Good} {Quantum} and {Locally} {Testable} {Classical} {LDPC} {Codes}},
	url = {http://arxiv.org/abs/2111.03654},
	doi = {10.48550/arXiv.2111.03654},
	urldate = {2024-02-21},
	publisher = {arXiv},
	author = {Panteleev, Pavel and Kalachev, Gleb},
	month = jan,
	year = {2022},
	note = {arXiv:2111.03654 [quant-ph]},
}

@article{breuckmann_balanced_2021,
	title = {Balanced {Product} {Quantum} {Codes}},
	volume = {67},
	issn = {0018-9448, 1557-9654},
	url = {http://arxiv.org/abs/2012.09271},
	doi = {10.1109/TIT.2021.3097347},
	number = {10},
	urldate = {2024-08-22},
	journal = {IEEE Transactions on Information Theory},
	author = {Breuckmann, Nikolas P. and Eberhardt, Jens N.},
	month = oct,
	year = {2021},
	note = {arXiv:2012.09271 [quant-ph]},
	pages = {6653--6674},
}

@misc{leverrier_quantum_2022,
	title = {Quantum {Tanner} codes},
	url = {http://arxiv.org/abs/2202.13641},
	doi = {10.48550/arXiv.2202.13641},
	urldate = {2024-08-22},
	publisher = {arXiv},
	author = {Leverrier, Anthony and Zémor, Gilles},
	month = sep,
	year = {2022},
	note = {arXiv:2202.13641 [quant-ph]},
}

@article{bravyi_universal_2005,
	title = {Universal {Quantum} {Computation} with ideal {Clifford} gates and noisy ancillas},
	volume = {71},
	issn = {1050-2947, 1094-1622},
	url = {http://arxiv.org/abs/quant-ph/0403025},
	doi = {10.1103/PhysRevA.71.022316},
	number = {2},
	urldate = {2024-08-22},
	journal = {Physical Review A},
	author = {Bravyi, Sergei and Kitaev, Alexei},
	month = feb,
	year = {2005},
	note = {arXiv:quant-ph/0403025},
	pages = {022316},
}

@misc{wills_constant-overhead_2024,
	title = {Constant-{Overhead} {Magic} {State} {Distillation}},
	url = {http://arxiv.org/abs/2408.07764},
	doi = {10.48550/arXiv.2408.07764},
	urldate = {2024-08-22},
	publisher = {arXiv},
	author = {Wills, Adam and Hsieh, Min-Hsiu and Yamasaki, Hayata},
	month = aug,
	year = {2024},
	note = {arXiv:2408.07764 [quant-ph]},
}

@misc{golowich_asymptotically_2024,
	title = {Asymptotically {Good} {Quantum} {Codes} with {Transversal} {Non}-{Clifford} {Gates}},
	url = {http://arxiv.org/abs/2408.09254},
	doi = {10.48550/arXiv.2408.09254},
	urldate = {2024-08-22},
	publisher = {arXiv},
	author = {Golowich, Louis and Guruswami, Venkatesan},
	month = aug,
	year = {2024},
	note = {arXiv:2408.09254 [quant-ph]},
}

@misc{nguyen_good_2024,
	title = {Good binary quantum codes with transversal {CCZ} gate},
	url = {http://arxiv.org/abs/2408.10140},
	doi = {10.48550/arXiv.2408.10140},
	urldate = {2024-08-22},
	publisher = {arXiv},
	author = {Nguyen, Quynh T.},
	month = aug,
	year = {2024},
	note = {arXiv:2408.10140 [quant-ph]},
}

@misc{williamson_layer_2024,
	title = {Layer {Codes}},
	url = {http://arxiv.org/abs/2309.16503},
	doi = {10.48550/arXiv.2309.16503},
	urldate = {2024-08-22},
	publisher = {arXiv},
	author = {Williamson, Dominic J. and Baspin, Nouédyn},
	month = may,
	year = {2024},
	note = {arXiv:2309.16503 [quant-ph]},
}

@article{rengaswamy_optimality_2020,
	title = {On {Optimality} of {CSS} {Codes} for {Transversal} \${T}\$},
	volume = {1},
	issn = {2641-8770},
	url = {http://arxiv.org/abs/1910.09333},
	doi = {10.1109/JSAIT.2020.3012914},
	number = {2},
	urldate = {2024-08-22},
	journal = {IEEE Journal on Selected Areas in Information Theory},
	author = {Rengaswamy, Narayanan and Calderbank, Robert and Newman, Michael and Pfister, Henry D.},
	month = aug,
	year = {2020},
	note = {arXiv:1910.09333 [quant-ph]},
	pages = {499--514},
}

@article{lubotzky_ramanujan_1988,
	title = {Ramanujan graphs},
	volume = {8},
	issn = {1439-6912},
	url = {https://doi.org/10.1007/BF02126799},
	doi = {10.1007/BF02126799},
	language = {en},
	number = {3},
	urldate = {2024-08-22},
	journal = {Combinatorica},
	author = {Lubotzky, A. and Phillips, R. and Sarnak, P.},
	month = sep,
	year = {1988},
	pages = {261--277},
}

@article{kubica2023efficient,
  title={Efficient color code decoders in $ d \geq 2$ dimensions from toric code decoders},
  author={Kubica, Aleksander and Delfosse, Nicolas},
  journal={Quantum},
  volume={7},
  pages={929},
  year={2023},
  publisher={Verein zur F{\"o}rderung des Open Access Publizierens in den Quantenwissenschaften}
}

@article{kesselring2018boundaries,
  title={The boundaries and twist defects of the color code and their applications to topological quantum computation},
  author={Kesselring, Markus S and Pastawski, Fernando and Eisert, Jens and Brown, Benjamin J},
  journal={Quantum},
  volume={2},
  pages={101},
  year={2018},
  publisher={Verein zur F{\"o}rderung des Open Access Publizierens in den Quantenwissenschaften}
}

@article{michnicki20143d,
  title={3d topological quantum memory with a power-law energy barrier},
  author={Michnicki, Kamil P},
  journal={Physical review letters},
  volume={113},
  number={13},
  pages={130501},
  year={2014},
  publisher={APS}
}

@article{song2023topological,
  title={Topological defect network representations of fracton stabilizer codes},
  author={Song, Zijian and Dua, Arpit and Shirley, Wilbur and Williamson, Dominic J},
  journal={PRX Quantum},
  volume={4},
  number={1},
  pages={010304},
  year={2023},
  publisher={APS}
}

@article{KAUFMAN2023103696,
    author = {Tali Kaufman and Izhar Oppenheim},
    doi = {https://doi.org/10.1016/j.ejc.2023.103696},
    issn = {0195-6698},
    journal = {European Journal of Combinatorics},
    note = {40th Anniversary Edition},
    pages = {103696},
    title = {High dimensional expanders and coset geometries},
    url = {https://www.sciencedirect.com/science/article/pii/S0195669823000136},
    volume = {111},
    year = {2023},
}

@misc{nielsen_quantum_2010,
	title = {Quantum {Computation} and {Quantum} {Information}: 10th {Anniversary} {Edition}},
	shorttitle = {Quantum {Computation} and {Quantum} {Information}},
	url = {https://www.cambridge.org/highereducation/books/quantum-computation-and-quantum-information/01E10196D0A682A6AEFFEA52D53BE9AE},
	journal = {Higher Education from Cambridge University Press},
	author = {Nielsen, Michael A. and Chuang, Isaac L.},
	month = dec,
	year = {2010},
	doi = {10.1017/CBO9780511976667},
	note = {ISBN: 9780511976667
Publisher: Cambridge University Press},
}

@misc{breuckmann2018phdthesis,
      title={PhD thesis: Homological Quantum Codes Beyond the Toric Code}, 
      author={Nikolas P. Breuckmann},
      year={2018},
      eprint={1802.01520},
      archivePrefix={arXiv},
      primaryClass={quant-ph},
      url={https://arxiv.org/abs/1802.01520}, 
}

@misc{golowich2024quantumldpccodestransversal,
      title={Quantum LDPC Codes with Transversal Non-Clifford Gates via Products of Algebraic Codes}, 
      author={Louis Golowich and Ting-Chun Lin},
      year={2024},
      eprint={2410.14662},
      archivePrefix={arXiv},
      primaryClass={quant-ph},
      url={https://arxiv.org/abs/2410.14662}, 
}

\appendix
\clearpage
\onecolumngrid

\section{Edge contraction}
\label{app:edge_contraction}
The smallest 3D colour code obtainable from the HGP rainbow code construction described previously is a $[\![384,9,4]\!]$ code obtained using a product of three length-$4$ cycle graphs. However, a more efficient colour code with parameters $[\![96,9,4]\!]$ defined on a different lattice is known to exist. Is there a way to obtain this code from a product construction, and if so, can a similar approach be used to generate more efficient rainbow codes? 

To answer these questions we need to define a new operation on simplex graphs, which we call \textit{edge contraction}.

\begin{definition}
    Given a simplex graph $\mathcal{G}$ with colours $\{c_0,...,c_k\}$ the \textbf{$c_i$-contraction}, $\mathrm{Cont}(c_i):\mathcal{G}$, is the operation
    \begin{equation}
        \{g_i \Leftarrow g_j ~ \forall ~ (g_i,g_j) \in P\} : (\mathcal{G}/\{c_i\})
    \end{equation}
    where $P$ is the set of edges of a spanning forest of $\mathcal{G}$ that contains only $c_i$ edges. We say that the graph $\mathrm{Cont}(c_i):\mathcal{G}$ has colours $\{c_0,...,\cancel{c_i},...,c_k\}$.
\end{definition}

In other words, $\mathrm{Cont}(c_i)$ removes from $\mathcal{G}$ all $c_i$-coloured edges and the glues together all vertices which were previously part of the same $\{c_i\}$-maximal subgraph. Another way to understand it is as the contraction of all $c_i$-coloured edges (hence the name), so that all vertices previously connected by $c_i$ edges become associated and all $c_i$ edges in the graph are removed. As contractions for different $c_i$ only affect edges of different colours these operations necessarily commute and we can also talk more generally about $\{c_i,c_j,...\}$-contractions ($\mathrm{Cont}(c_i,c_j,...)$) without needing to specify an order. An example of $c_0$-contraction in a 2D colour code lattice is shown in \cref{fig:2d_contraction}.

\begin{figure}[b]
    \centering
    \includegraphics[width=0.5\linewidth]{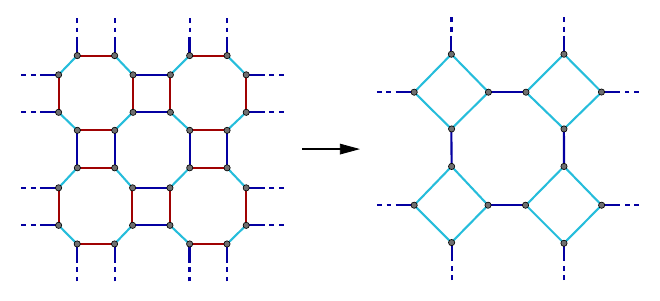}
    \caption{$c_0$-contraction in a 2D colour code lattice. Colours are $c_0=$ dark red, $c_1=$ light blue, $c_2=$ dark blue. The result is a ``rotated'' 2D colour code with the same $k$ and $d$ as the original code but half as many physical qubits.}
    \label{fig:2d_contraction}
\end{figure}

Contraction also defines a mapping of each maximal or rainbow subgraph via the contraction of all $c_i$ edges in this subgraph. For example, in \cref{fig:2d_contraction} the light blue/dark blue octagons are $\{c_1,c_2\}$-maximal/rainbow subgraphs just as in the original code while the light blue diamonds are $\{\cancel{c_0},c_1\}$-maximal/rainbow subgraphs and the dark blue edges are $\{\cancel{c_0},c_2\}$-maximal/rainbow subgraphs. This defines a corresponding mapping on the stabilisers of the original code but in general these operators will not commute in the contracted code. For example, the intersection of a $\{\cancel{c_0},c_1\}$- and $\{\cancel{c_0},c_2\}$-maximal/rainbow subgraph in \cref{fig:2d_contraction} can be only a single vertex. We must therefore choose some commuting subset of these operators to define the stabilisers for a code on the contracted lattice. In the example of \cref{fig:2d_contraction} the most obvious choice is the operators defined on $\{\cancel{c_0},c_1\}$- and $\{c_1,c_2\}$-maximal/rainbow subgraphs. We then obtain a rotated version of the original code with parameters $[\![16,4,4]\!]$ as opposed to the $[\![32,4,4]\!]$ uncontracted code. 

The $[\![384,9,4]\!]$ and $[\![96,9,4]\!]$ 3D colour codes are also related by an edge contraction ($\mathrm{Cont}(c_0,c_3)$) as shown in \cref{fig:3d_contraction}. The $X$ stabilisers of the contracted code are $\{\cancel{c_0},c_1,c_2\}$- and $\{c_1,c_2,\cancel{c_3}\}$-maximal/rainbow subgraphs and the $Z$ stabilisers are $\{\cancel{c_0},c_1\}$-, $\{c_2,\cancel{c_3}\}$- and $\{c_1,c_2\}$-maximal/rainbow subgraphs. Because we have contracted two colours of edge we get a factor of four reduction in the number of physical qubits ($384/4=96$). There is also an intermediate lattice obtained by performing only one of the two contractions (it does not matter which one) which defines a code with parameters $[\![192,9,4]\!]$. We can also contract either of the remaining edge colours to obtain a octahedral-cuboctahedral lattice which supports a triple of 3D surface codes related to the 3D colour code by an unfolding map, and which admit a transversal $CCZ$ gate~\cite{vasmer_three-dimensional_2019}. The stabilisers of these three codes can also be obtained from a suitable choice of contracted subgraphs. Construction of 3D colour codes with the various edge contractions can be explored in the \href{https://github.com/m-webster/CSSLO/blob/main/rainbow_codes/C%20colour-code-edge-contraction.ipynb}{linked Jupyter notebook}.

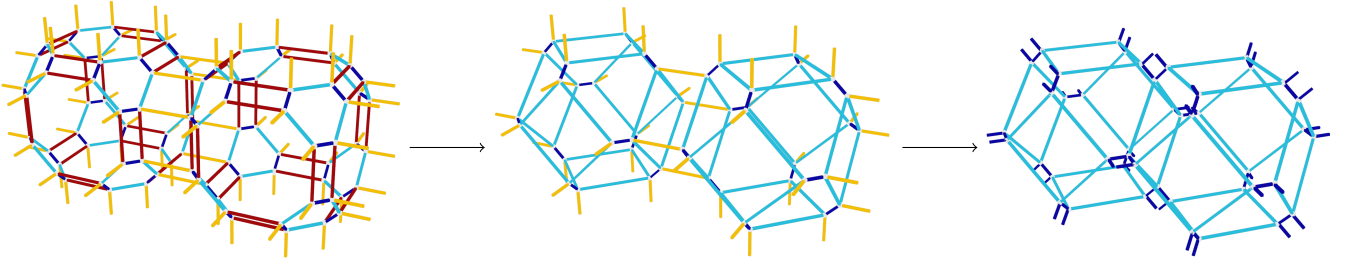
\begin{figure*}
    \centering
        \includegraphics[width=.3\textwidth]{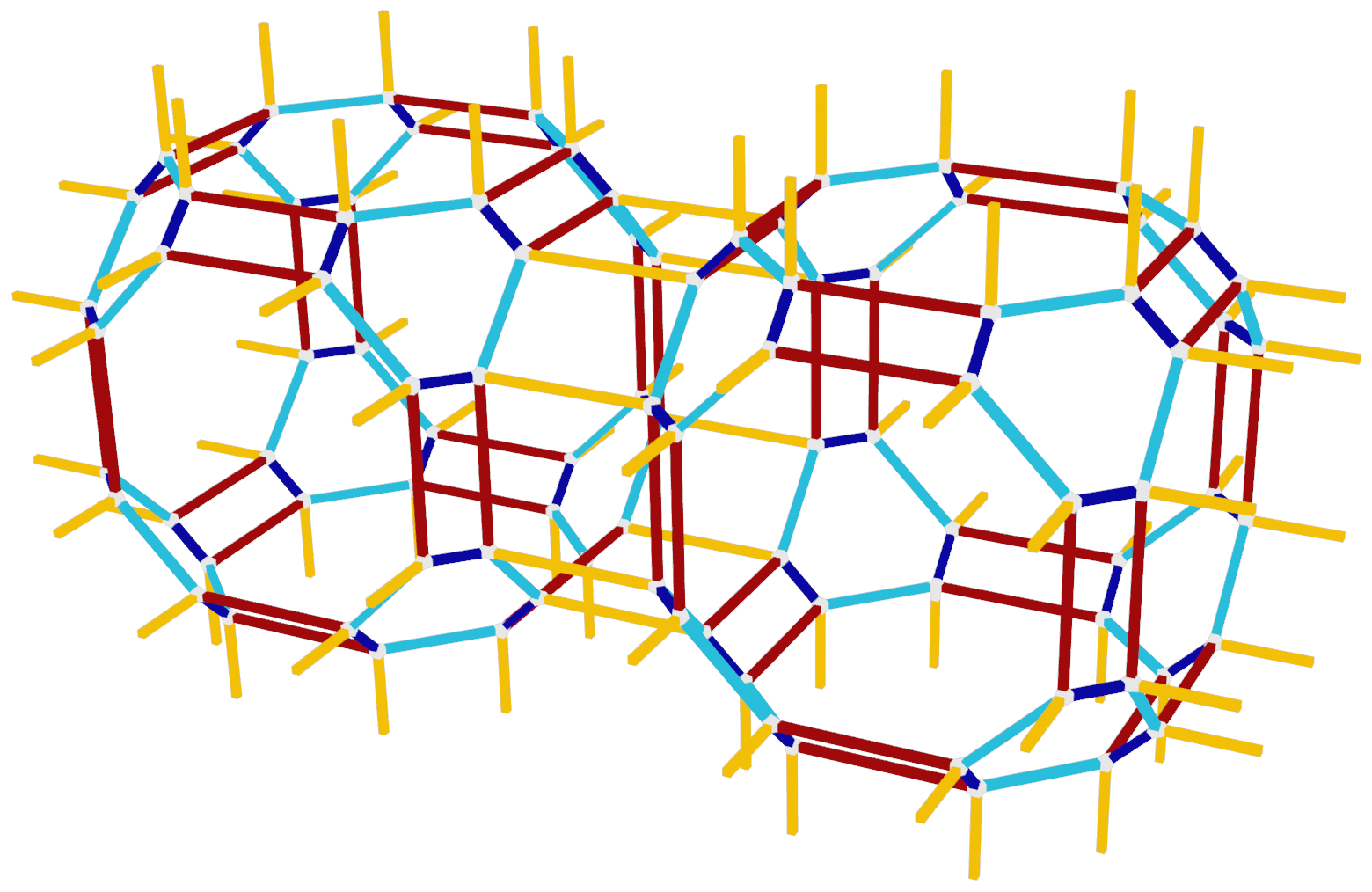}
        \begin{tikzpicture}
            \draw[->] (0,0) -- (1,0);
            \draw[draw=white] (0,0) -- (0,-1.5);
        \end{tikzpicture}
        \includegraphics[width=.3\textwidth]{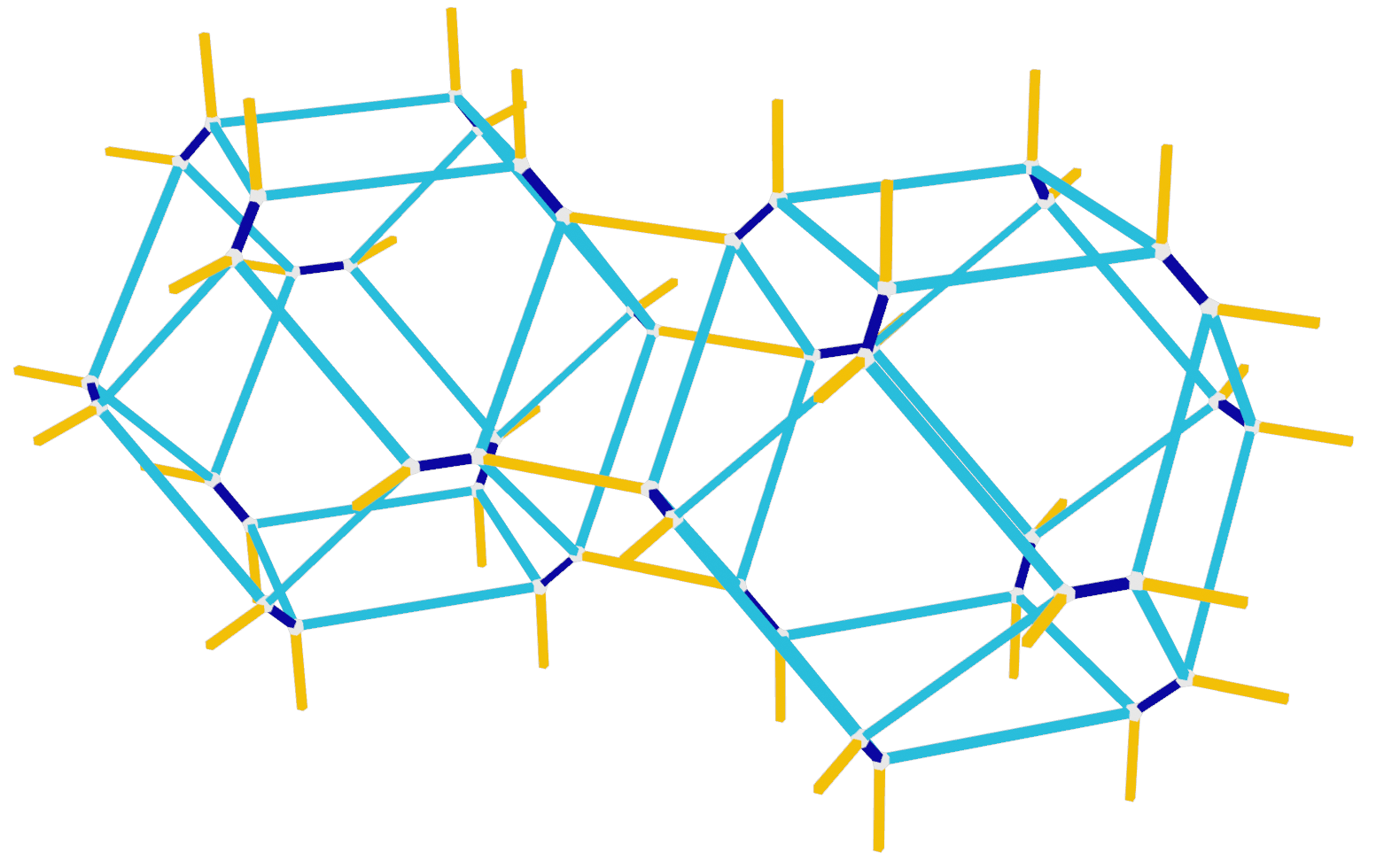}
        \begin{tikzpicture}
            \draw[->] (0,0) -- (1,0);
            \draw[draw=white] (0,0) -- (0,-1.5);
        \end{tikzpicture}
        \includegraphics[width=.26\textwidth]{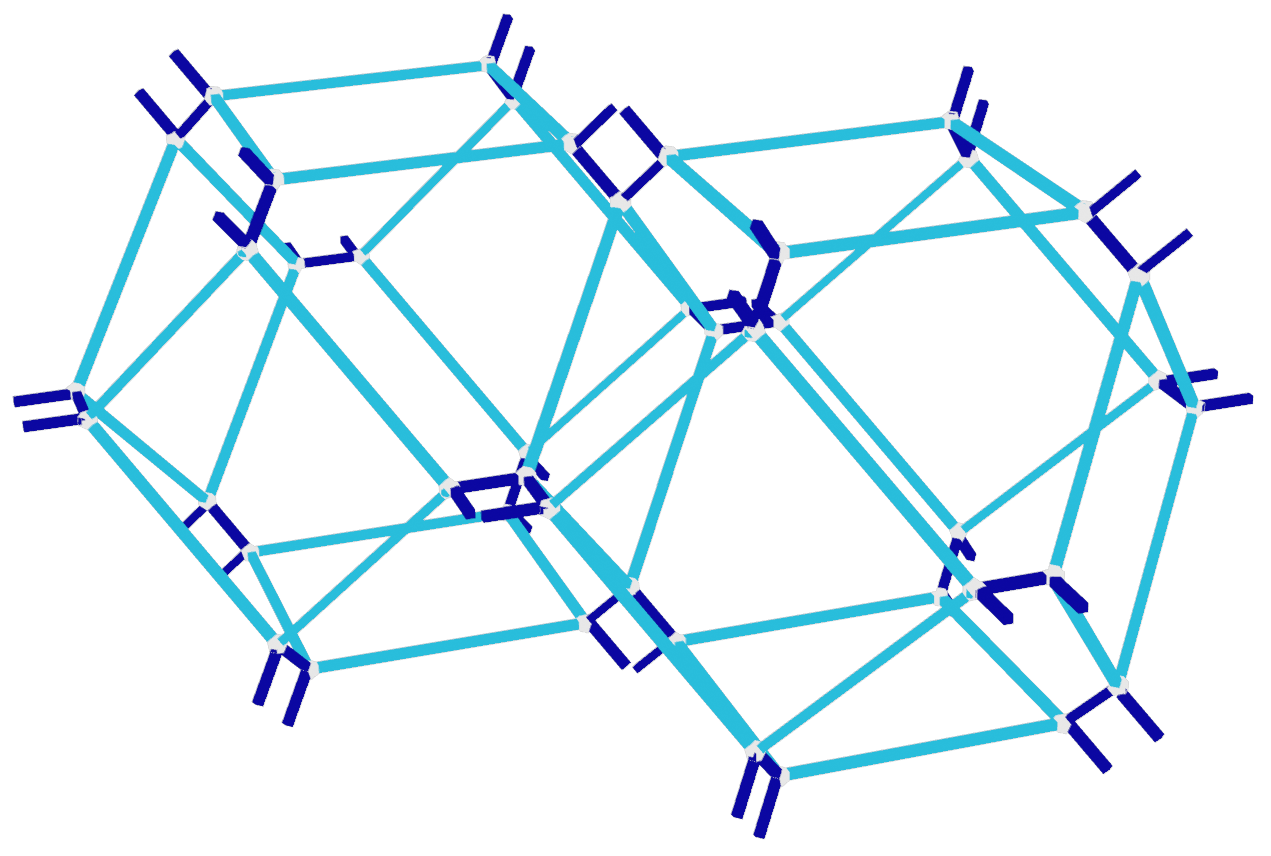}
    \caption{\textit{(Left)} Subgraph of the flag graph of a 3D cubic lattice, equivalent to a 3D colour code lattice. Colours are $c_0=$ dark red, $c_1=$ light blue, $c_2=$ dark blue, $c_3=$ light orange, so this subgraph contains two $\{c_0,c_1,c_2\}$-maximal ($=$ rainbow) subgraphs and a $\{c_0,c_1,c_3\}$-maximal ($=$ rainbow) subgraph. \textit{(Middle)} Graph obtained from contracting all $c_0$ edges of the previous graph, so that the vertices at the endpoints of these edges become associated and the edges themselves are deleted from the graph. The result is a subregion of a different 3D colour code lattice (up to recolouring of edges). \textit{(Right)} Graph obtained from contracting all $c_3$ edges of the previous graph. This is also a subregion of a different colour code lattice (up to recolouring of edges).}
    \label{fig:3d_contraction}
\end{figure*}

\begin{figure*}
    \centering
    \includegraphics[width=0.9\textwidth]{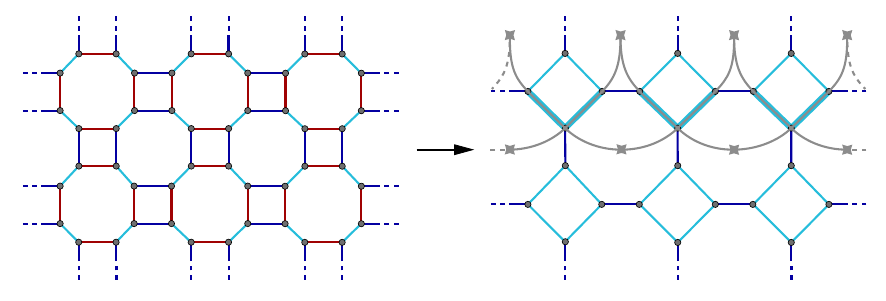}
    \caption{\textit{(Left)} a 2D colour code defined on the flag graph of a $3\times2$ square lattice on a torus. (\textit{Right}) $c_0$-contraction of this lattice resulting in a code with $X$ and $Z$ stabilisers on $\{\cancel{c_0},c_1\}$- and $\{c_1,c_2\}$-maximal/rainbow subgraphs. Shown in grey is a logical operator defined on pair of edges (thick light blue lines) and the stabilisers anticommuting with these edge operators (grey crosses). The operator must wrap twice around the torus in order to close, resulting in a loss of two encoded qubits relative to the uncontracted code.}
    \label{fig:2d_contraction_2}
\end{figure*}

In order to understand when and how this procedure might be generalised to other rainbow codes we first need to understand what makes it work in the cases discussed above. The most obvious requirement is that the stabilisers of the code must still have even weight after the edge contraction (this is necessary for the resulting code to still have a transversal non-Clifford gate due to the requirements discussed in the previous section). It is for this reason that we can contract the $c_0$ and $c_3$ edges of the $[\![384,9,4]\!]$ code but not the $c_1$ or $c_2$ edges, as all $2$-rainbow subgraphs containing $c_0$ or $c_3$ have size $0 \mod 4$, but $\{c_1,c_2\}$-rainbow subgraphs have size $6 = 2 \mod 4$, and so $\{\cancel{c_1},c_2\}$- or $\{c_1,\cancel{c_2}\}$-rainbows subgraphs would have size three. This fact can be understood as arising from the structure of the cubic lattice (which was used to obtain the original flag graph). This lattice has Coxeter diagram \CoxeterFour{}{4}{}{3}{}{4}{} which tells us that all $2$-rainbow subgraphs are generated by symmetries of even order except for $\{c_1,c_2\}$-subgraphs which are generated by symmetries of odd order. This insight can allow us to perform edge contraction on colour codes defined on hyperbolic lattices, for example we demonstrate how to halve the number of physical qubits used in a hyperbolic colour code based on the 4-3-5 3D tiling in the \href{https://github.com/m-webster/CSSLO/blob/main/rainbow_codes/D%20manifold-tilings.ipynb}{linked Jupyter notebook}. We can also see that, because the Coxeter diagram for a $D$-dimensional hypercubic lattice has the form
\begin{tikzpicture}
    \node[circle,draw,fill=black,inner sep=0,minimum size=3] at (0,0) {};
    \node[circle,draw,fill=black,inner sep=0,minimum size=3] at (0.5,0) {};
    \node[circle,draw,fill=black,inner sep=0,minimum size=3] at (2,0) {};
    \node[circle,draw,fill=black,inner sep=0,minimum size=3] at (2.5,0) {};
    \draw (0,0) -- (0.5,0);
    \node[anchor=south] at (0.25,-0.08) {\tiny 4};
    \draw (0.5,0) -- (0.85,0);
    \node[anchor=south] at (0.75,-0.08) {\tiny 3};
    \draw[dashed] (1,0) -- (1.5,0);
    \draw (1.65,0) -- (2,0);
    \node[anchor=south] at (1.75,-0.08) {\tiny 3};
    \draw (2,0) -- (2.5,0);
    \node[anchor=south] at (2.25,-0.08) {\tiny 4};
\end{tikzpicture}
we can in general only contract $c_0$ and $c_D$ edges in flag graphs obtained via the hypergraph product.

This is not the only requirement, however. If we consider a graph $\mathcal{G}$ obtained from a product of a length-$4$ and a length-$6$ cycle which defines a $[\![48,4,4]\!]$ 2D colour code then the graph $\mathrm{Cont}(c_0):\mathcal{G}$ defines a $[\![24,2,4]\!]$ code, so two logical qubits have been lost. We can see this in \cref{fig:2d_contraction_2}. We observe similar results in 3D, for example a product of three length-$6$ cycles defines a $[\![1296,9,6]\!]$ code but $c_0$ contraction results in a $[\![648,6,6]\!]$ code and an additional $c_3$ contraction results in a $[\![324,6,6]\!]$ code. This suggests that for the number of logical qubits to be preserved under contraction we should choose input graphs whose fundamental cycle bases contain only cycles of length $0 \mod 4$ and not $2 \mod 4$. 

We also need to understand how contraction works at seams between colour codes. An example is shown in \cref{fig:seam_contraction}. In this case we can see that, unlike in a colour code containing no seams, $\{\cancel{c_0},c_1\}$-rainbow subgraphs and $\{c_1,c_2\}$-maximal subgraphs are not required to have even intersection and instead can intersect at a single vertex. This makes assignment of stabilisers to $c_0$-contracted lattices containing type-$1$ seams (or $c_D$-contracted lattices containing type-$0$ seams) quite difficult. In contrast, because $c_D$-contraction does not modify type-$1$ seams (and similarly for $c_0$ contraction and type-$0$ seams) the same stabiliser as in unjoined colour codes can be applied in this case. This means that in codes obtained from products of graphs in which all level-$1$ (level-$0$) vertices have degree $2$ we can reliably perform $c_0$ ($c_D$) edge contraction and reliably obtain a commuting stabiliser group.

\begin{figure}
    \centering
    \includegraphics[width=0.5\linewidth]{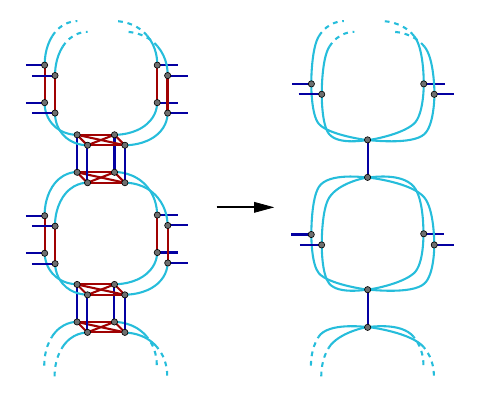}
    \caption{\textit{(Left)} Subgraph of the flag graph shown in \cref{fig:glued_ccs} that contains a seam. \textit{(Right)} Transformation of this seam under $c_0$-contraction. The intersection of $\{\cancel{c_0},c_1\}$-maximal and $\{c_1,c_2\}$-rainbow subgraphs can now be a single vertex.}
    \label{fig:seam_contraction}
\end{figure}

We believe that more sophisticated edge contraction techniques could lead to significant resource savings in triorthogonal quantum codes, but leave the development of these techniques as a problem for future work. 

\section{Unfolding of rainbow codes}
\label{app:unfolding}
An important characteristic of $D$-dimensional colour codes is the possibility to map them into $D$ copies of the toric code through a local Clifford unitary \cite{kubica2015unfolding}. This mapping, often called "unfolding" due its interpretation in two dimensions, has many useful consequences, such as the possibility to construct colour code decoders based on toric code ones \cite{kubica2023efficient} or the understanding of topological excitations and defects in the colour codes \cite{kesselring2018boundaries} 

We prove here the existence a local Clifford unitary map between any generic rainbow code constructed from a length-2 chain complex, and the tensor product of two codes whose qubits are supported on the 1-maximal subgraphs of the corresponding flag graph. We call this mapping "unfolding" by analogy with the colour code case, and prove its existence by generalizing the techniques developed in Ref.~\cite{kubica2023efficient}. When applying the same toolbox to rainbow codes in higher dimensions or beyond the generic case (e.g. mixed rainbow codes), those techniques seem to fail, which suggests that no such mapping exist in general. We leave this problem open for future work.

\subsection{General unfolding process}

To prove the existence of an unfolding map, we follow the procedure outlined in Ref.~\cite{kubica2015unfolding}, that we summarize now. The goal is to construct a local Clifford unitary between a stabilizer code $\mathcal{C}_A$ supported on qubits $\mathcal{Q}_A$ to a stabilizer code $\mathcal{C}_B$ supported on qubits $\mathcal{Q}_B$. 
For this, the idea is to partition the two qubit sets into small parts, $\{Q_A^{(i)}\}$ and $\{Q_B^{(i)}\}$, and construct a unitary $U_i$ for each element of the partition, such that the overall unitary operator can be written as a tensor product of all the $U_i$, making it local by construction. 
To see the effect of such unitary on the stabilizer group of $C_A$, we need to consider the overlap of all the stabilizers of $C_A$ with each qubit subset $Q_A^{(i)}$. If we know how all those stabilizer overlaps are mapped by all the $U_i$, we can determine the effect of the overall unitary on the code. The overlap of all the stabilizers of a code $\mathcal{C}$ with any given subset of qubits $Q$ forms a group, that we call the \textbf{overlap group} and write
\begin{align}
    \mathcal{O}(\mathcal{C},Q)=\{ P \in \mathcal{P}_n \; : \; P = S \cap Q,\; S \in \mathcal{S}(\mathcal{C}) \}
\end{align}
where $\mathcal{P}_n$ denotes the Pauli group on $n$ qubits (with $n$ the number of physical qubits of the code), and $\mathcal{S}(\mathcal{C})$ refers to the stabilizer group associated to $\mathcal{C}$.

An unfolding process is then given by the following data:
\begin{enumerate}
    \item A partition $Q_A=\bigsqcup_{i=1}^\ell Q_A^{(i)}$ and $Q_B=\bigsqcup_{i=1}^\ell Q_B^{(i)}$ of the qubits of the two codes
    \item For each subset of the partition, a group isomorphism $\mathcal{M}_i: \mathcal{O}_A^{(i)} \rightarrow \mathcal{O}_B^{(i)}$, where $O_A^{(i)}=\mathcal{O}(C_A,\mathcal{Q}_A^{(i)})$ and $O_B^{(i)}=\mathcal{O}(C_B,\mathcal{Q}_B^{(i)})$. Given a set of generators for the groups $O_A^{(i)}$ and $O_B^{(i)}$, the mapping $\mathcal{M}_i$ is entirely determined by how it maps generators of $\mathcal{O}_A^{(i)}$ into generators of $\mathcal{O}_B^{(i)}$.
\end{enumerate}

We say that an unfolding process is \textbf{valid} if for each $i$, there exists a unitary $U_i$ that implements the mapping $\mathcal{M}_i$. For a valid unfolding process, the unitary
\begin{align}
    U = \bigotimes_{i=1}^{\ell} U_i
\end{align}
therefore maps the code $C_A$ to the code $C_B$.

The following theorem (shown in \cite{kubica2015unfolding}) characterizes an unfolding process:

\begin{theorem}
    Let's consider an unfolding process equipped with a set of generators $\mathcal{G}(O_A^{(i)})$ and $\mathcal{G}(O_B^{(i)})$ for each overlap group, such that for all $1 \leq i\leq \ell$,
    \begin{enumerate}
        \item $G(O_A^{(i)})=G(O_B^{(i)})$
        \item $G(Z(O_A^{(i)}))=G(Z(O_B^{(i)}))$
        \item $[g,g']=[h,h']$ for each $g,g' \in \mathcal{G}(O_A^{(i)})$, $h=\mathcal{M}_i(g)$, $h'=\mathcal{M}_i(g')$.
    \end{enumerate}
    where $G(O)=|\mathcal{G}(O)|$ denotes the number of generators of the group $O$, and $Z(O)$ represents the center of the group, that is, all the group elements that commute with every other elements. Then, such an unfolding process is valid.
    \label{theorem:kubica-unfolding}
\end{theorem}
Note that the first and second conditions prove that the there exists an isomorphism between the two overlap groups, while the last condition ensures that there is unitary implementing the explicit group isomorphism given by $\mathcal{M}_i$. The proof of \cref{theorem:kubica-unfolding} heavily relies on the framework constructed in Ref.~\cite{yoshida2010framework}.

\subsection{Outline of the proof}

In the case of 2-rainbow codes, the unfolding process is defined as follows. The code $C_A$ corresponds to the rainbow code, with qubits $Q_A$ on vertices. 
The code $C_B$
has qubits on the $1$-maximal subgraphs of the flag graph, i.e. edges and higher-size cliques. 
We then choose two colors $c_a \neq c_b$ and consider the partition of the flag graph into its $\{c_a, c_b\}$-maximal subgraphs, that we call $\mathcal{G}_i$. 
The vertex set of each $(c_a, c_b)$-maximal subgraph then correspond to $Q_A^{(i)}$, while its cliques define  subsets of $Q_B^{(i)}$. 
Indeed, as we will see, there are more vertices than cliques in the flag graph (and in each $2$-maximal subgraphs), and we therefore also need to add some ancilla qubits in the set $Q_B^{(i)}$. One can easily see that the overlap groups $O_A^{(i)}$ consist of pairs of $X$s on every edge of $\mathcal{G}_i$, and $Z$ operators on the vertices of every clique. The overlap groups $O_B^{(i)}$ will be deduced from our choice of mappings $\mathcal{M}_i$. 

The rest of the proof will go as follow:
\begin{enumerate}
    \item Counting the number of ancilla qubits $A_i$ required for each subgraph $\mathcal{G}_i$ of the partition. That is, we will calculate $A_i=n_A^{(i)} - n_B^{(i)}$, where $n_A^{(i)}$ and $n_B^{(i)}$ are the number of physical qubits in $Q_A^{(i)}$ and $Q_B^{(i)}$ respectively 
    \item Constructing a group isomorphism $\mathcal{M}_i: O_A^{(i)} \rightarrow O_B^{(i)}$ by its effect on the generators of $O_A^{(i)}$.
    \item Proving that this mapping is valid by proving that the three conditions of \cref{theorem:kubica-unfolding} are satisfied.
    \item Showing that the resulting code $C_B$ is a tensor product of two codes supported on disjoint lattices.
\end{enumerate}

Those steps should show that there exists a local Clifford unitary between $C_A$ and $C_B$. 

For the rest of the proof, unless stated otherwise, let's choose a specific $\{c_a,c_b\}$-maximal subgraph $\mathcal{G} := \mathcal{G}_i$, associated to the qubit sets $Q_A := Q_A^{(i)}$ and $Q_B := Q_B^{(i)}$, and the overlap groups $O_A := O_A^{(i)}$ and $O_B := O_B^{(i)}$.

\subsection{Number of ancilla qubits}

Let's then determine the number of ancilla qubits. Denoting $m_k$ the number of $1$-maximal subgraphs made of $k$ vertices (i.e. the number of $k$-cliques) of the flag graph, the number of qubits in $\mathcal{C}_B$ is equal to 
\begin{align}
    n_B=\sum_{k=1}^\infty m_{2k}
\end{align}

\begin{lemma}
    The number of vertices $v$ of $\mathcal{G}$ is related to the number of $1$-maximal subgraphs via the following relation:
    \begin{align}
        v = \sum_{k=1}^{\infty} k m_{2k} = m_2 + 2m_4 + 3m_6 + \ldots
    \label{eq:vertex-clique-relation}
    \end{align}
    \label{lemma:vertex-clique-relation}
\end{lemma}
\begin{proof}
    Let's prove this lemma by induction on each $m_{2k}$ with $k\geq 2$. 
    We say that a flag graph has \textbf{clique cardinality} $(m_{2i})_{i\geq 1}$ if it has $m_{2i}$ cliques of size $i$.

    For the base case, let's consider a flag graph of clique cardinality $(m_2,0,0,\ldots)$. 
    The only type of flag graph with no cliques of size higher than $2$ are $2$-rainbow subgraphs. Since $2$-rainbow subgraphs are cycle graphs, the number of vertices is equal to the number of edges, or equivalently to the number of $2$-cliques $m_2$. Therefore $v=m_2$ and \cref{eq:vertex-clique-relation} is satisfied.
    
    Let's assume that \cref{eq:vertex-clique-relation} is true for any flag graph of clique cardinality $(m_{2i})_{i\geq 1}$. Let $k\geq 2$. Let's show that \cref{eq:vertex-clique-relation} is also true for a flag graph $\mathcal{F}$ of clique cardinality $(m_2,\ldots,m_{2k-2},m_{2k}+1,m_{2k+2},\ldots)$, that is, with one more clique of size $2k$. $\mathcal{F}$ has at least one $2k$-clique of size $2k$. Let's select one, partition its vertices into $k$ pairs, and delete every edge that does not correspond to a pair of the partition. This new graph, that we call $\widetilde{\mathcal{F}}$, is a valid flag graph, with the same number of vertices as $\mathcal{F}$, but one less $2k$-clique and $k$ more $2$-cliques (corresponding to the edges that have not been deleted). Therefore, $\widetilde{\mathcal{F}}$ has clique cardinality $(m_2+k,\ldots,m_{2k-2},m_{2k},m_{2k+2},\ldots)$. By the induction hypothesis, we know that the number of vertices of $\widetilde{\mathcal{F}}$ is given by \cref{eq:vertex-clique-relation}:
    \begin{align}
        v = (m_2+k) + 2m_4 + \ldots + km_{2k} + \ldots
    \end{align}
    which we can rewrite as:
    \begin{align}
        v = m_2 + 2m_4 + \ldots + k(m_{2k}+1) + \ldots
        \label{eq:vertex-clique-relation-induction-step}
    \end{align}
    Since the number of vertices is the same for $\widetilde{\mathcal{F}}$ and $\mathcal{F}$, we can deduce that \cref{eq:vertex-clique-relation-induction-step} gives the correct identity and \cref{eq:vertex-clique-relation} is verified for $\mathcal{F}$.
\end{proof}

Therefore, the number of ancilla qubits is given by

\begin{align}
    \begin{split}
        A &= n_A - n_B \\
        &= v - \sum_{k=1}^{\infty} m_{2k} \\
        &= \sum_{k=1}^{\infty} k m_{2k} - \sum_{k=1}^{\infty} m_{2k} \\
        A &= \sum_{k=2}^{\infty} (k-1) m_{2k} = m_4 + 2m_6 + 3 m_8 + \ldots
    \end{split}
\end{align}

The number of ancilla qubits can also be expressed in terms of the \textit{rainbow rank} of $\mathcal{G}$, a result that will be useful when constructing the mapping $\mathcal{M}$ in the next section. We define the \textbf{rainbow rank} of a flag graph as the number of independent rainbow subgraphs. By independence of rainbow subgraphs, we mean that the corresponding binary vectors (with a component for each vertex and a $1$ whenever a vertex is included in the subgraph) are linearly independent on $\mathbb{Z}_2$.
The following lemma shows that the rainbow rank is directly related to the number of ancilla qubits.

\begin{figure}
    \centering
    \subfloat[]{
        \includegraphics[width=0.3\textwidth]{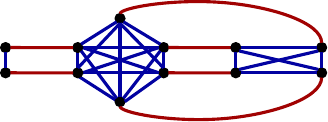}
        \label{fig:lemma-6-flag-graph}
    }
    \subfloat[]{
        \includegraphics[width=0.26\textwidth]{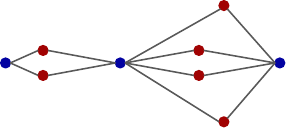}
        \label{fig:lemma-6-clique-graph}
    }
    \caption{Illustration of the proof of \cref{lemma:r-A-relations}.\textbf{(a)} Flag graph with one $6$-clique, one $4$-clique, and a rainbow rank $r=4$. \textbf{(b)} Corresponding clique graph, with one node per clique and an edge per connection between two cliques. Rainbow subgraphs of the flag graph correspond to cycles of the clique graph.}
    \label{fig:lemma-6}
\end{figure}

\begin{lemma} 
    The rainbow rank $r$ of a flag graph satisfies the following identity:
    \begin{align}
        r = 1 + \sum_{k=2}^{\infty} (k-1) m_{2k} = 1+A
    \label{eq:r-A-relations}
    \end{align}
    \label{lemma:r-A-relations}
\end{lemma}
\begin{proof}
    To prove this identity, let's define the \textit{clique graph} $\mathcal{C}(\mathcal{G})$ associated to $\mathcal{G}$ as a graph where each clique of $\mathcal{G}$ correspond to a vertex in $\mathcal{C}(\mathcal{G})$, and each edge of $\mathcal{G}$ connecting the vertices of two cliques becomes an edge connecting those two cliques in $\mathcal{C}(\mathcal{G})$. An example of clique graph is shown in \cref{fig:lemma-6}. Note that the clique graph is in reality a multigraph, as multiple edges can connect a given pair of nodes. The clique graph has the following useful properties:
    \begin{enumerate}
        \item It is a connected graph where each vertex has a degree equal to the size of the corresponding clique.
        \item Cycles of $\mathcal{C}(\mathcal{G})$ are in one-to-one correspondence with rainbow subgraphs of $\mathcal{G}$
    \end{enumerate}
    Denoting $\tilde{v}$ and $\tilde{e}$ the number of vertices and edges of $\mathcal{C}(\mathcal{G})$, we can deduce from those properties that the rainbow rank is given by
    \begin{align}
        r = \tilde{e} - \tilde{v} + 1
        \label{eq:r-e-v-relation}
    \end{align}
    Moreover, the number of vertices can easily be obtained as $\tilde{v}=\sum_{k \geq 1} m_{2k}$, while the number of edges is given by
    \begin{align}
        \tilde{e} = \frac{1}{2} \sum_{k \geq 1} 2k m_{2k}
    \end{align}
    Indeed, each $2k$-clique corresponds to a vertex of degree $2k$, and we divide the sum by two to avoid over-counting edges. Inserting this in \cref{eq:r-e-v-relation}, we get
    \begin{align}
        \begin{split}
            r &= 1 + \sum_{k \geq 1} k m_{2k} - \sum_{k \geq 1} m_{2k} \\
            &= 1 + \sum_{k \geq 2} (k-1) m_{2k}
        \end{split}
        \label{eq:r-e-v-relation}
    \end{align}
    which is the desired formula.
\end{proof}

\subsection{Mapping construction}

The mapping $\mathcal{M}$ is constructed by describing how generators of $O_A$ are mapped into operators supported on the cliques of $\mathcal{G}$. For a $Z$ generator supported on the vertices of a clique $\mathcal{C}$, we would like to map it into a single $Z$ operator living on $\mathcal{C}$. However, we note that the product of all the $Z$ generators of $O_A$ gives the identity, while the corresponding single-$Z$ operators does not cancel. Therefore, such mapping would map the identity operator into a non-identity operator, and thus cannot be an isomorphism. To solve this issue, we apply the mapping described above to all but one clique. For this clique, we transform the input generator by multiplying it to an all $X$ operator (which is an element of the center of $O_A$), and map this new generator to the single $Z$ operator on the clique. You can find an illustration of this in \cref{fig:unfolding-map-z}.

For $X$ operators, the idea is to map operators supported on the vertices of an edge into an operator acting on the two cliques incident to that edge. However, we have the converse problem as for $Z$ operators: every rainbow subgraph supports two non-trivial $X$ generators that would map to the identity, namely the product of all the $c$-colored edges for each choice of colors $c\in \{c_a,c_b\}$. To circumvent this issue, we make use of the ancilla qubits in the following way. Let's consider operators associated to edges of color $c\in \{c_a,c_b\}$. Remember from \cref{lemma:r-A-relations} that there are exactly $A=r-1$ ancilla qubits. Let's choose a basis of rainbow subgraphs $R_1,\ldots,R_r$. For each rainbow subgraph $R_k$ with $k \leq r-1$, we apply the mapping described above on all the $c$-colored edges but one. On the last edge, we multiply the output operator by an $X$ operator on the $k$th ancilla qubit. For the last edge of the last rainbow $R_r$, we transform the input generator by multiplying it by an all-vertex $X$ operator. The output can be obtained by applying the mapping on each $c$-colored edge for which it is already described. Note that all other edges of $\mathcal{G}$ can be written as a product of the edges already described, and the previous discussion therefore completely describes $\mathcal{M}$. An illustration of this is shown in \cref{fig:unfolding-map-x}.

\begin{figure*}[t!]
    \centering
    \subfloat[]{
        \includegraphics[valign=c,width=0.4\textwidth]{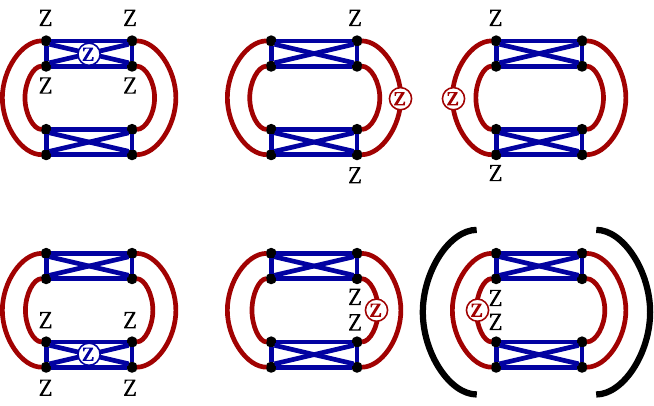}
        \vphantom{\includegraphics[valign=c,width=0.4\textwidth]{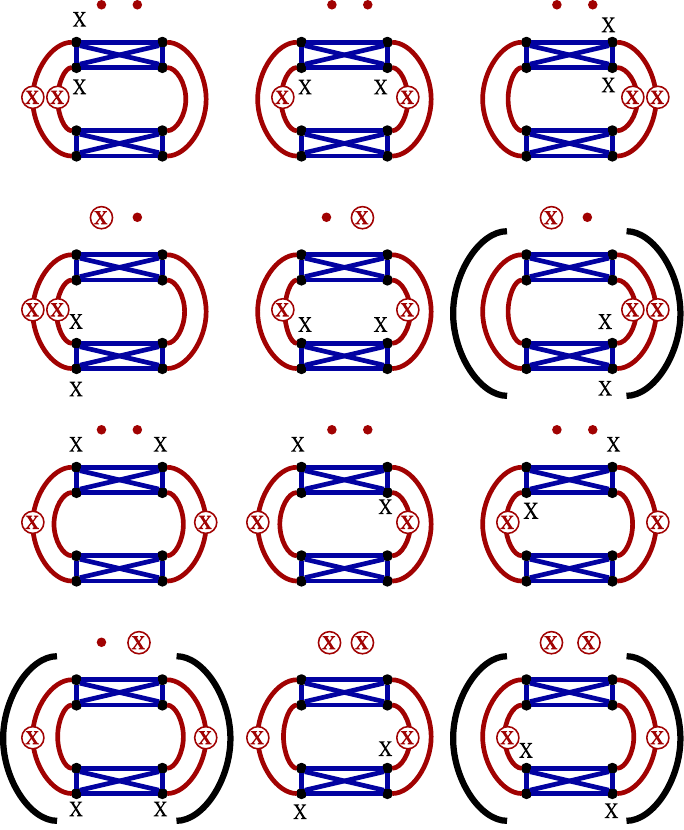}}
        \label{fig:unfolding-map-z}
    }
    \subfloat[]{
        \includegraphics[valign=c,width=0.4\textwidth]{figures/unfolding-map-x.pdf}
        \label{fig:unfolding-map-x}
    }
    \caption{Example of unfolding map on an overlap group. The black $X$s and $Z$s correspond to operators in $O_A$, while circled red/blue $X$s and $Z$s correspond to the output operators in $O_B$. Following the convention of Ref.~\cite{kubica2015unfolding}, we use parentheses when the input is multiplied by elements of the centralizer, which correspond here to an all $X$/$Z$ operator. \textbf{(a)} $Z$ operators on the vertices of every clique are mapped to single-qubit $Z$ operators acting on the clique, except for one operator (the bottom-right in this example), which is only mapped after multiplication by a $Z$ operators acting on all vertices. \textbf{(b)} $X$ operators on edges are mapped to the incident cliques of the opposite color. We look here at red edges only, but the same reasoning applies to blue edges. Dots on top of each graph represent ancilla qubits. The mapping on the first two rows are derived by choosing a basis of $r=3$ rainbow subgraphs (one rainbow subgraph per column) and making use of an ancilla qubit for the last edge of the first $r-1=2$ rainbow subgraphs. For the last edge of the last rainbow subgraph, we need to multiply the input operator by an all-vertex $X$ operator. The ancilla operators of the output are determined by performing this product and using all previously-constructed elements of the mappings. The third and fourth columns can be determined by linearity using elements of the first two columns}
    \label{fig:unfolding-map}
\end{figure*}

\subsection{Validity of the mapping}

To prove that the mapping is valid, we first need to show that the two groups $O_A$ and $O_B$ are indeed isomorphic (by proving that the two groups and their centers have the same number of generators), and then the mapping $\mathcal{M}$ preserves the commutation relations.

\paragraph{Existence of an isomorphism between $O_A$ and $O_B$} We show here that the number of generators of the two overlap groups are equal, as well as the number of generators of their center. Let's start by the group themselves:
\begin{proposition}
    The number of generators of the two overlap groups is given by
    \begin{align}
        G(O_A)=G(O_B)=m+v-2
    \end{align}
    where $m=\sum_{k=1}^\infty m_{2k}$ is the total number of $1$-maximal subgraphs and $v$ is the number of vertices of the flag graph.
\end{proposition}
\begin{proof}
    Let's start by computing the number of independent generators of $O_A$. The group is generated by $m$ $Z$ operators defined on the vertices of the $1$-maximal subgraphs, and $e$ $X$ operators defined on the $e$ edges of the graph. The $Z$ operators, that we denote $S_i^Z$, have a single global relation between them, given by
    \begin{align}
        \prod_i S_i^Z = I.
    \end{align}
    Indeed, each vertex is at the intersection of exactly two $1$-maximal subgraphs, so taking their product gives an identity operator on each vertex. 
    We can show that this is the only relation on $Z$ operators by contradiction. Let's assume that there exists another relation between the $S_i^Z$. Let's consider the subgraph made of all the $1$-maximal subgraphs involved in this relation. Each vertex of this subgraph must be connected to exactly two maximal subgraphs of this set (it must be even due to the relation, non-zero since we are including all the vertices in all the $1$-maximal subgraphs of this subgraph, but cannot be more than two). Therefore, this subgraph must be a $2$-maximal subgraph, which is only possible if it is the full $2$-maximal subgraph. Therefore, the only relation between $Z$ operators is the global relation involving all the $1$-maximal subgraphs. Hence, there are $m-1$ independent $Z$ operators in $O_A$

    Let's now take a look at $X$ operators in $O_A$. Since those are supported on the two vertices of each edge of the graph, simple cycles gives the relations between them. The number of simple cycles, or circuit rank, of any connected graph is given by $e-v+1$. Therefore, the graph contains $e-(e-v+1)=v-1$ independent $X$ operators. 

    Adding up $X$ and $Z$ operators, we found that
    \begin{align}
        G(O_A)=m+v-2
    \end{align}
    proving the first part of the proposition.

    We now turn to $O_B$. It contains single-qubit $Z$ operators for each $1$-maximal subgraph, which are all independent as single-qubit operators. There are therefore exactly $m$ independent $Z$ operators in $O_B$. On the other hand, $X$ operators are defined for each edge on the pair of $1$-maximal subgraphs incident to its two vertices. It is easy to check that each cycle of the graph defines a relation on those operators. 
    However, those are not the only relations: every rainbow cycle contains exactly two relations: one for each color of the rainbow cycle, since the product of the operators associated to all edges of the same color in a rainbow cycle is the identity. Therefore, there is a total of $(e-v+1)+r$ relations, and $e-(e-v+1)-r=v-r-1$ independent $X$ operators associated to edges. Adding this up with the number of ancilla qubits (as we have one independent $X$ operator per ancilla), we get:
    \begin{align}
        G(O_B)=m+v-r-1+A=m+v-2
    \end{align}
    where we used \cref{lemma:r-A-relations} to remove $A$ and $r$ from the equation.
\end{proof}
To finish the proof that $O_A$ and $O_B$ are isomorphic, we still need to show that $G(Z(O_A))=G(Z(O_B))$. This is the content of the following proposition
\begin{proposition} The number of generators of the centers of the two overlap groups is given by
\begin{align}
    G(Z(O_A))=G(Z(O_B))=r+1
\end{align}
\end{proposition}
\begin{proof}
    Let's start by determining the center of $O_A$, that is, all the elements that commute with all its generators. Remember that rainbow subgraphs have an even intersection with every clique. Therefore, since $Z$ operators live on the vertices of cliques, any $X$ operator supported on the vertices of a rainbow subgraph commutes with all the $Z$ operators. We can obtain such an $X$ operators by taking the product of $X$ generators on all the edges of one color in the rainbow subgraph. Therefore, such an operator belongs to $Z(O_A)$. Since there are $r$ independent rainbow subgraphs, there are $r$ independent $X$ generators in $Z(O_A)$. On the other hand, there is only one $Z$ operators in $Z(O_A)$, consisting in a $Z$ on all vertices of the flag graph. Therefore, $G(Z(O_A))=r+1$

    In $O_B$, every clique supports a single $Z$ operators, so the center does not contain any $X$ operator supported on cliques. However, $X$ operators on the ancilla qubits belong to the center. Therefore, there are exactly $A=r-1$ independent $X$ operators in $Z(O_B)$. Moreover, $Z$ operators supported on all cliques of a single color also commute with all $X$ operators. Since there are two colors in the flag graphs that we are considering, there are two independent $Z$ operators in $Z(O_B)$, giving a total number of generators $G(Z(O_B))=r+1$.
\end{proof}

\paragraph{Commutation relations} The final step to prove that there exists a unitary operator implementing the mapping $\mathcal{M}$ described above is to check that it preserves the commutation relation.

Let's choose an $X$ generator $g_X \in O_A$ and a $Z$ generator $g_z \in O_A$. Since $g_X$ is supported on two vertices, the intersection between $g_X$ and $g_Z$ must have weight $0$, $1$ or $2$. 

Let's start with the case where the two operators do not overlap. In this case, the edge that supports $g_X$ has no intersection with the clique that supports $g_Z$. Therefore, the two cliques in the support of $\mathcal{M}(g_X)$ will have no intersection with the clique corresponding to $\mathcal{M}(g_Z)$, and the commutation relations are preserved.

Let's now consider the case where the two operators overlap on two qubits. In this case, the edge corresponding to $g_X$ must be part of the clique corresponding to $g_Z$. Since $\mathcal{M}(g_X)$ is supported on a clique of a different color than the edge associated to $g_X$ (and therefore to the clique associated to $g_Z$), the two operators will have no overlap once $\mathcal{M}$ has been applied.

Finally, let's looks at the case where the two operators overlap on one qubit. This means that the edge associated to $g_X$ is adjacent to the clique associated to $g_Z$. Thus one of the two cliques in the support of $\mathcal{M}(g_X)$ is the one associated to $g_Z$, and the overlap of the mapped operators consists in exactly one qubit. This shows that the commutation relations are preserved by the mapping.

\subsection{Resulting code $C_B$}

Let's now prove that the code $C_B$ resulting from the mapping $\mathcal{M}$ can be written as the tensor product of two codes, one supported on $c_a$-colored cliques and one supported on $c_b$-colored cliques. To see this, let's take a look at the different types of stabilizers of $C_A$ and see how they transform through $\mathcal{M}$. 

Let's start with stabilizers supported on any of the graphs $\mathcal{G}_i$. Since $Z$ stabilizers correspond to $2$-maximal subgraphs, we should study how the $Z$ operator acting on all the vertices of $\mathcal{G}_i$ is mapped. Since for one of the clique, of color $c \in \{c_a,c_b\}$, we modified the generators via a product of $X$s, it means that the $Z$ stabilizers maps to a $Z$ operator acting on all the cliques of color opposite to $c$. 
On the other hand, multiplying by the $Z$ stabilizer by an $X$ stabilizer acting on all vertices of $\mathcal{G}_i$ (which is a stabilizer as a product of rainbow subgraphs), we get a $Y$ stabilizer that maps to a $Z$ operator acting on all the cliques of color $c$. 
Moreover, as established in the previous section, each $X$ stabilizer acting on a rainbow subgraph is mapped to $X$ operators on ancilla qubits. 
Therefore, each $\{c_a,c_b\}$-maximal subgraph supports exactly two operators: one acting all $c_a$-colored cliques, and one acting on all $c_b$-colored cliques.

Let's now look at stabilizers supported on a $\{c_a,c_c\}$-maximal subgraph, where $c_c \notin \{c_a,c_b\}$. A similar reasoning will also apply to $\{c_b,c_c\}$-maximal subgraphs. Let's first consider the $Z$ stabilizer supported on all vertices of this maximal subgraph. After a potential application of some $X$ stabilizers acting on all the vertices of some of the graphs $\mathcal{G}_i$, it maps to a $Z$ operator acting on all the $c_a$-colored cliques. Similarly, rainbow subgraphs can be decomposed into edges of color $c_a$, which map to $c_b$-colored cliques after the potential application of some $X$ operators.

Therefore, all the stabilizers of $C_B$ act either on $c_a$-colored or $c_b$-colored cliques, with no mixing between the two. Thus we can deduce that $C_B$ is a tensor product of two codes, $C_B=C_B^a \otimes C_B^b$, where $C_B^a$ corresponds to $c_a$-colored cliques and $C_B^b$ corresponds to $c_b$-colored cliques.

An illustration of such a mapping of stabilizers is shown in \cref{fig:unfolding-stabilizer-mapping}

\begin{figure}
    \centering
    \includegraphics[width=0.49\textwidth]{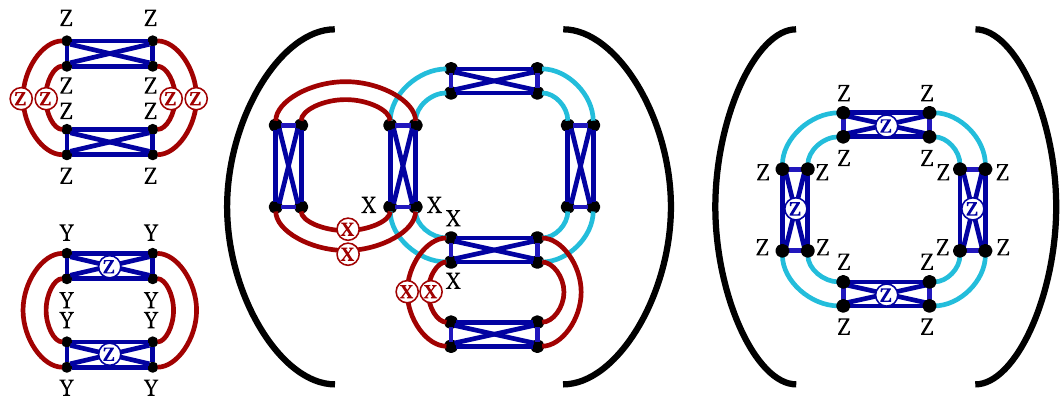}
    \caption{Unfolding of stabilizers}
    \label{fig:unfolding-stabilizer-mapping}
\end{figure}

\section{Constructing Maximal and Rainbow Subgraphs of Simplex Graphs}
\label{app:algorithms}
In this Appendix, we set out algorithms for constructing maximal and rainbow subgraphs that are  used to define check matrices for rainbow codes.
We assume that we are given a \textbf{simplicial complex} of dimension $D$ and that each $0$-cell has been allocated a colour from the range $[0..D]$. 

From the simplicial complex, we construct a \textbf{simplex graph} as outlined in \ref{subsection:rainbow_simplex} in which the vertices are the $D$-dimensional cells of the complex. Vertices are connected by an edge of colour $c$ if they differ only by a $0$-cell of colour $c$. 

Quantum codes are constructed by identifying qubits with the vertices of the simplex graph. Checks are identified by subgraphs of the simplex graph which are defined in terms of a set of colours $S := \{c_0,..,c_{d-1}\} \subseteq [0..D]$ and are of the following types:
\begin{enumerate}
    \item \textbf{Maximal Subgraphs:} subgraphs of vertices connected by edges of colours included in the set $S$;
\item \textbf{Rainbow Subgraphs:} subgraphs such that, for each colour $c_i$ in $S$, each vertex in the subgraph is connected to another vertex in the subgraph by exactly one edge of colour $c_i$.
\end{enumerate}

\subsection{Maximal Subgraphs}
Maximal subgraphs of type $S$ can be obtained by finding connected components of the simplex graph where vertices are considered adjacent only if the edge is of colour $c_i \in S$. In \Cref{alg:MSG}, we build maximal subgraphs recursively from subgraphs with a smaller number of colours using a spanning tree algorithm. To calculate the spanning tree, we consider sets of subgraphs $F$ connected by a set of edges of colour $c_0$ - see \Cref{alg:ST}.
The algorithm for computing  maximal subgraphs has complexity linear in the number of edges and vertices of the simplex graph.
\begin{figure}[H]
\begin{algorithm}[H]
\caption{Maximal Subgraphs}\label{alg:MSG}
\begin{algorithmic}
\State\textbf{Input:} 
\State A list of colours c := [c0,..,cd]
\State A simplex graph G
\State\textbf{Output:}  
\State Maximal subgraphs of G of type c
\Function{MSG}{G,c}
\If{len(c) = 0}
\State return [[v] for v in vertices(G)]
\Else
\Comment{Recursive call to MSG with d-1 colours}
\State F = \Call{MSG}{G, c[1:]}
\Comment{Spanning tree - edges of colour c0}
\State CC,ST,Cycles = \Call{SpanningTree}{G,c0,F}
\State return [vertices(C) for C in CC]
\EndIf
\EndFunction
\end{algorithmic}
\end{algorithm}
\end{figure}
\begin{figure}[H]
\begin{algorithm}[H]
\caption{Subgraph Spanning Tree}\label{alg:ST}
\begin{algorithmic}
\State\textbf{Input:} 
\State A simplex graph G
\State A list of subgraphs F
\State A colour c0 in [0,..,D]
\State\textbf{Output:}  
\State Connected components, spanning tree and unvisited edges (which correspond to cycles) for subgraphs F joined by edges of colour c0
\Function{SpanningTree}{G,c0,F}
\State Ftodo := F
\State CC := list()
\State ST := dict()
\State Cycles := set()
\While{len(Ftodo) $>$ 0}
\State f1 := Ftodo.pop()
\State Fvisited:=set(f1)
\State ST[f1] := None
\ForAll{edges (v1 $\in$ f1, v2 $\in$ f2 $\ne$ f1) of colour c0}
\If{f2 not in Fvisited}
\Comment{First time visiting f2 - update spanning tree}
\State ST[f2] := (f1, v1, v2)
\State Fvisited.add(f2)
\Else:
\Comment{Edges not followed when generating the }\Comment{spanning tree correspond to cycles}
\State Cycles.add((f1,f2,v1,v2))
\EndIf
\EndFor
\State CC.append(Fvisited)
\State Ftodo = Ftodo  - Fvisited
\EndWhile
\State return CC, ST, Cycles
\EndFunction
\end{algorithmic}
\end{algorithm}
\end{figure}
\subsection{Rainbow Subgraphs}

Rainbow subgraphs are more complex to obtain in general than maximal subgraphs. Rainbow subgraphs of type $S$ are subgraphs of a maximal subgraph of type $S$.  Finding rainbow subgraphs exhaustively would require us to consider all possible combinations of vertices within each maximal subgraph of the corresponding type, and is worst-case exponential complexity in the number of vertices in the maximal subgraphs.
We demonstrate  below efficient algorithms for finding rainbow subgraphs in the special case for rainbow codes where the Z-type checks are two-colour rainbow subgraphs.



\subsubsection{Two-Colour Rainbow Subgraphs}

In \Cref{alg:RSG2}, we show how to construct two-colour rainbow subgraphs of type $S = \{c_0,c_1\}$. We do this by first constructing the maximal subgraphs $F$ of type $\{c_1\}$ using \Cref{alg:MSG}. We then construct a spanning tree where subgraphs in $F$ are connected by edges of colour $c_0$ using \Cref{alg:ST}. This results in a set of connected components, a spanning tree and a set of unvisited edges which correspond to cycles. Unvisited edges are represented by a tuple $(f_1, f_2, v_1, v_2)$ where $(v_1, v_2)$ is an edge of colour $c_0$ joining two maximal subgraphs $f_1, f_2$. Using \Cref{alg:STPath}, we find the common ancestor $f$ of $f_1$ and $f_2$ and the set of vertices joining $f, f_1$ and $f_2$ which form the rainbow subgraph.
The complexity of this algorithm is linear in the number of vertices and edges of the simplex graph, but also depends on the number of vertices in the rainbow subgraphs.

\begin{algorithm}[H]
\caption{Two-Colour Rainbow Subgraphs}\label{alg:RSG2}
\begin{algorithmic}
\State\textbf{Algorithm: Two-Colour Rainbow Subgraphs}
\State\textbf{Input:} 
\State A list of colours c := [c0,c1] of size 2
\State A simplex graph G
\State\textbf{Output:}  
\State Rainbow subgraphs of G of type c
\Function{RSG2}{G,c}
\State c0, c1 := c
\Comment{F = maximal subgraphs of type c1} 
\State F := MSG(G, [c1])
\Comment{ST = spanning tree where vertices are }\Comment{MSG of type c1 joined by edges of type c0}
\Comment{Cycles = unfollowed edges when generating ST}
\State CC, ST, Cycles := \Call{SpanningTree}{G,c0,F}
\State RSG := set() 
\For{(f1,f2,u,v) in Cycles} 
\Comment{Generate paths from f1 and f2 to the root }\Comment{of the spanning tree}
\State (fList1, uList1,vList1):= \Call{STpath}{ST, f1}
\State (fList2, uList2,vList2) := \Call{STpath}{ST, f2}
\Comment{There is a common ancestor f of f1 and f2 }\Comment{such that fList1[j] = fList2[k]}
\State f := fList1 $\cap$ fList2
\State j := fList1.indexof(f)
\State k := fList2.indexof(f)
\Comment{Vertices leading to common ancestor plus }\Comment{the connecting edge (u,v) form an RSG}
\State r := uList1[:j-1] $\cup$ vList1[:j-1] \State  $\cup$ uList2[:k-1] $\cup$ vList2[:k-1] $\cup$ \{u,v\} 
\State RSG.add(r)
\EndFor
\State return RSG
\EndFunction
\end{algorithmic}
\end{algorithm}
\begin{figure}
\begin{algorithm}[H]
\caption{Spanning Tree Path}\label{alg:STPath}
\begin{algorithmic}
\State\textbf{Input:} 
\State A spanning tree ST
\State A subgraph f
\State\textbf{Output:}  
\State Path from f to the root of ST which includes a list of subgraphs (fList), and lists of vertices forming edges (uList, vList) of colour c0
\Function{STpath}{ST,f}
\State fList := list(f)
\State uList := list()
\State vList := list()
\While {ST[f] is not None}
\State (fi, ui, vi) := ST[f]
\State fList.append(fi)
\State uList.append(ui)
\State vList.append(vi)
\State f := fi
\EndWhile
\State return fList, uList, vList
\EndFunction
\end{algorithmic}
\end{algorithm}
\end{figure}
\subsubsection{Multi-colour Colour Rainbow Subgraphs}
In this work, we construct CSS codes whose  Z-checks are associated with $2$-colour subgraphs and  X-checks with $D$-colour subgraphs.
As a result, Z-checks can be generated efficiently, whether they are maximal (using \Cref{alg:MSG}) or rainbow type (using \Cref{alg:RSG2}). X-checks which are of maximal type can also be generated efficiently, leaving X-checks of rainbow type $S$ where there are more than two colours in $S$.

In \Cref{alg:RSGKer}, we show how to generate the remaining X-checks of rainbow type $S$. We first calculate the maximal subgraphs of type $S$ using \Cref{alg:MSG}. Rainbow subgraphs of type $S$ are contained within a maximal subgraph of type $S$. Any X-check must also commute with the Z-checks and so be in the intersection of a maximal subgraph and the kernel of the Z-checks. The intersection of spans can be calculated using linear algebra techniques as set out in \Cref{alg:SpanIntersection}.
The kernel calculation dominates the complexity of this algorithm and accordingly the overall complexity is polynomial in the number of vertices in the simplex graph.
\begin{figure}[H]
\begin{algorithm}[H]
\caption{Multi-Colour Rainbow Subgraphs}\label{alg:RSGKer}
\begin{algorithmic}
\State\textbf{Input:} 
\State A list of colours c = [c0,..,cd] where d $>$ 1
\State A simplex graph G
\State A generating set of Z-checks SZ in the form of a binary s $\times$ n matrix
\State A generating set of Z-logicals LZ in the form of a binary k $\times$ n matrix
\State \textbf{Output:} 
\State Rainbow Subgraphs of type c
\Function{RSGKer}{G,c,SZ,LZ}
\Comment{X-checks which must commute with Z-checks}\Comment {and so are in the Kernel of SZ  modulo 2}
\State K := \Call{KerModN}{SZ,2}
\State RSG := list()
\Comment{RSG of type c are contained}\Comment{in the MSG of type c}
\For{f in \Call{MSG}{G, c}}
\Comment{Add the intersection f and K to RSG}
\State R := \Call{SpanIntersection}{K,f,2}
\State RSG.extend(R)
\EndFor
\State return RSG
\EndFunction
\end{algorithmic}
\end{algorithm}
\end{figure}

\begin{figure}[H]
\begin{algorithm}[H]
\caption{Intersection of Spans}\label{alg:SpanIntersection}
\begin{algorithmic}
\State\textbf{Input:} 
\State Two matrices generating matrices A, B for spans modulo N
\State\textbf{Output:}  
\State A generating matrix for the intersection of the spans
\Function{SpanIntersection}{A,B,N}
\State r := len(A)
\Comment{Kernel of transpose of A and B stacked}\Comment{modulo 2}
\State K := \Call{KerModN}{(A.T $|$ B.T),2}
\Comment{Extract first r columns of K}
\State K1 := K[:,:r]
\Comment{Matrix product modulo 2}
\State return \Call{MatMulModN}{K1, A,2}
\EndFunction
\end{algorithmic}
\end{algorithm}
\end{figure}
\section{Coloured Logical Paulis}

In this Appendix, we show how to generate the coloured Paulis as defined in Section IV.B of \cite{vuillot2022quantum}. 
We set out a method to calculate coloured logical Z operators. Swapping Z-checks and logicals with X-checks and logicals results in a method for logical X operators. One of the inputs to the algorithm is a generating set of logical Z operators which are not necessarily coloured logicals. These can be generated, for instance, by using the method in Section 10.5.7 of \cite{nielsen_quantum_2010}.
\begin{figure}[H]
\begin{algorithm}[H]
\caption{Coloured Logical Pauli Operators}\label{alg:ColouredLZ}
\begin{algorithmic}
\State\textbf{Input:} 
\State Set of colours c = \{c0,..,cd\}
\State A simplex graph G
\State Z-checks SZ in the form of a binary s $\times$ n matrix
\State Z-logical Pauli operators LZ in the form of a binary k $\times$ n matrix.
\State \textbf{Output:} 
\State Coloured Logical Pauli Operators of type c 
\Function{ColouredLZ}{G,c,SZ,LZ}
\Comment{Inverted set of colours}
\State cInv := [0..D] - c
\Comment{Coloured LZ operators are combinations}\Comment{of MSG of type cInv...}
\State M := MSG(G, cInv)
\Comment{...that are also Logical Z operators}
\State SZLZ := stack(SZ,LZ)
\State L : = \Call{SpanIntersection}{M,SZLZ,2}
\Comment{Exclude elements of L which are stabilisers}
\State return \{z: z  a  row of L and  z $\notin \langle$ SZ $\rangle\}$
\EndFunction
\end{algorithmic}
\end{algorithm}
\end{figure}

\end{document}